\documentclass[11pt]{article}
\usepackage[left=1in,top=1in,right=1in,bottom=1in,head=0in,nofoot]{geometry}

\usepackage{setspace,url,bm,amsmath} 
\usepackage{scalerel}

\usepackage{xcolor}
 \usepackage{multirow}
 \usepackage{arydshln}
 \usepackage{mathtools}
\usepackage{titlesec} 
\titlelabel{\thetitle.\quad} 
\titleformat*{\section}{\bf\Large\center}

\usepackage{float}
\usepackage{graphicx} 
\usepackage{bbm}
\usepackage{latexsym}
\usepackage{caption}
\usepackage[margin=20pt]{subcaption}
\usepackage{enumerate}

\usepackage{amsthm}
\usepackage{amssymb}
\usepackage{amsmath}
\usepackage{color}

\usepackage[labelformat=simple]{subcaption}

\usepackage{booktabs,caption}
\usepackage{enumitem}
\usepackage{amsfonts}
\usepackage[all,import]{xy}
\def\T{{ \mathrm{\scriptscriptstyle T} }}
\newcommand{\ind}{\perp\!\!\!\!\perp} 
\newtheorem{definition}{Definition}
\newtheorem{lemma}{Lemma}

\newtheorem{theorem}{Theorem}
\newtheorem{proposition}{Proposition}
\newtheorem{corollary}{Corollary}

\usepackage{natbib} 
\bibpunct{(}{)}{;}{a}{}{,} 

\renewcommand{\refname}{References} 
\usepackage{etoolbox} 
\apptocmd{\sloppy}{\hbadness 10000\relax}{}{} 

\usepackage{multibib}
\newcites{sec}{References}

\usepackage{color}
\usepackage{listings}
\usepackage{hyperref}
\usepackage{booktabs}

\usepackage{lscape}
\usepackage{xurl}

\RequirePackage[normalem]{ulem}

\allowdisplaybreaks

\begin{document}
\onehalfspacing

\title{\bf  
Interpretable sensitivity analysis for the Baron--Kenny approach to mediation with unmeasured confounding
} 
\author{Mingrui Zhang and Peng Ding
\footnote{Mingrui Zhang, Division of Biostatistics, University of California, Berkeley, CA 94720 U.S.A. (E-mail: mingrui\_zhang@berkeley.edu). Peng Ding, Department of Statistics, University of California, Berkeley, CA 94720 U.S.A. (E-mail: pengdingpku@berkeley.edu). 
}
}
\date{}
 
\maketitle

\begin{abstract}
Mediation analysis assesses the extent to which the exposure affects the outcome indirectly through a mediator and the extent to which it operates directly through other pathways. The popular Baron--Kenny approach estimates the indirect and direct effects of the exposure on the outcome based on linear regressions. However, when the exposure and the mediator are not randomized, the estimates may be biased due to unmeasured confounding. We first derive general omitted-variable bias formulas in linear regressions with vector responses and regressors. We then use the formulas to develop a sensitivity analysis method for the Baron--Kenny approach in the presence of unmeasured confounding. To ensure interpretability, we express the sensitivity parameters to correspond to the natural factorization of the joint distribution of the direct acyclic graph. They measure the partial correlation between the unmeasured confounder and the exposure, mediator, and outcome, respectively. We further propose a novel measure called the ``robustness value for mediation'' or simply the ``robustness value'', to assess the robustness of results based on the Baron--Kenny approach with respect to unmeasured confounding. Intuitively, the robustness value measures the minimum value of the maximum proportion of variability explained by the unmeasured confounding, for the exposure, mediator, and outcome, to overturn the results of the direct and indirect effect estimates. Importantly, we prove that all our sensitivity bounds are attainable and thus sharp. 
\end{abstract}

\medskip 
\noindent 
{\bf Keywords}: 
causal inference; Cochran's formula; direct effect; indirect effect; omitted-variable bias; robustness value 

\newpage
\section{Introduction}\label{sec::introduction}

\subsection{Baron--Kenny approach and unmeasured confounding}\label{subsec::introduction-bk}
Researchers are often interested in the causal mechanism from an exposure to an outcome. Mediation analysis is a tool for quantifying the extent to which the exposure acts on the outcome indirectly through a mediator and the extent to which the exposure acts on the outcome directly through other pathways. That is, it estimates the indirect and direct effects of the exposure on the outcome. The Baron--Kenny approach to mediation \citep{baron1986moderator} is the most popular approach for empirical mediation analysis, with more than one hundred thousand Google Scholar citations when we were preparing for this manuscript. It decomposes the total effect of the exposure on the outcome based on linear structural equation models. However, with non-randomized exposure or mediator, unmeasured confounding can bias the estimates of the indirect and direct effects. A central question is then to assess the sensitivity of the estimates with respect to unmeasured confounding. This is often called {\it sensitivity analysis} in causal inference \citep[e.g.,][]{Cornfield::1959, Rosenbaum::1983JRSSB, lin1998assessing, imbens2003sensitivity, rothman2008modern, ding2016sensitivity, cinelli2020making}.

Sensitivity analysis for mediation analysis has been an important research topic in recent years. \citet{tchetgen2012semiparametric}, \citet{ding2016sharp}, and \citet{smith2019mediational} propose sensitivity analysis methods based on the counterfactual formulation of mediation analysis. They focus on nonlinear structural equation models. \citet{vanderweele2010bias} proposes a general formula for sensitivity analysis that can be used for the Baron--Kenny approach, and \citet{le2016bias} also derives a formula for the Baron--Kenny approach. However, their resulting sensitivity parameters are not easy to interpret due to the unknown scale of the unmeasured confounding. \citet{imai2010identification} propose a sensitivity analysis method for mediation analysis based on linear structural equation models. However, they do not discuss the most general form of the Baron--Kenny approach. \citet{cox2013sensitivity} compare and evaluate some existing methods for assessing sensitivity in mediation analysis.

It is still an open problem to develop an interpretable sensitivity analysis method for the most general form of the Baron--Kenny approach. It is perhaps a little surprising due to the popularity of the Baron--Kenny approach. The main challenge is to find parsimonious and interpretable sensitivity parameters that can quantify the strength of the unmeasured confounding with respect to the exposure, mediator, and outcome. We fill this gap by leveraging the insight from a recent result in \citet{cinelli2020making}, with similar results also in \citet{mauro1990understanding}, \citet{frank2000impact}, and \citet{hosman2010sensitivity}. The key idea of \citet{cinelli2020making} is to use the partial $R^2$ as sensitivity parameters that measure the proportions of variability explained by unmeasured confounding conditional on the observed variables. The $R^2$ parameters are interpretable and bounded between $0$ and $1$, and are utilized in various sensitivity analysis settings, including those in \citet{small2007sensitivity}, \citet{oster2019unobservable}, and \citet{freidling2022optimization}. However, the formula in \citet{cinelli2020making} cannot be applied immediately to sensitivity analysis for the Baron--Kenny approach, especially for the case with multiple mediators and multiple unmeasured confounders. We solve this problem by deriving a general omitted-variable bias (OVB) formula with an $R$ parameterization, which will be explained thoroughly in the following discussion.

\subsection{Our contributions}\label{sec::introduction-contribution}

In this paper, we develop sensitivity analysis methods for the Baron--Kenny approach. Our framework allows the mediator and unmeasured confounder to be scalar or vector. 

With a single unmeasured confounder, we first derive a general OVB formula with vector responses and regressors, which extends \citet{cinelli2020making}. We propose a novel $R$ measure and express the multivariate OVB in terms of the $R$ parameters\footnote{In earlier arXiv versions of this paper, we proposed a novel matrix $R^2$ parameter and derived similar formulas using the matrix $R^2$ parameterization. In this version, we adopt the $R$ parameterization because it leads to more elegant mathematical results.}. The $R^2$ parameters can be viewed as the Euclidean norm of the corresponding vector $R$ parameters, and thus, the $R$ parameters are bounded in a Euclidean unit ball. Our OVB formula serves as the basis for the sensitivity analysis for the Baron--Kenny approach, especially in the setting with possibly multiple mediators. 

Then, we apply our OVB formula to develop a sensitivity analysis method for the Baron--Kenny approach with a single unmeasured confounder. The main challenge is that a direct application of our OVB formula requires specifying the sensitivity parameters that do not correspond to the natural factorization of the directed acyclic graph for mediation analysis. This is also a challenge encountered by \citet{vanderweele2010bias} and \citet{ding2016sharp}. They do not provide any solutions. We solve the problem by leveraging the properties of the $R$ parameters. Based on the $R$ parameters, we propose a novel measure and call it the ``robustness value for mediation'' or simply the ``robustness value''. It is the minimum value of the maximum strength of these $R$ parameters measured by their Euclidean norms, or equivalently, their corresponding $R^2$ parameters, such that the point estimate can be reduced to $0$ or the confidence interval can be altered to cover $0$. Intuitively, the corresponding $R^2$ parameters can be interpreted as the proportions of variability explained by the unmeasured confounding in the exposure, mediator, and outcome models given the observed variables. 

With multiple unmeasured confounders, a common trick is to consider a linear combination of the unmeasured confounders from the least squares regression that could fully characterize the OVB, and thus the problem reduces to the discussion with a single unmeasured confounder \citep{hosman2010sensitivity, cinelli2020omitted}. Although this trick works well for many problems, it does not work in sensitivity analysis for the indirect effect with multiple mediators. It breaks down because the exposure could affect the outcome through the mediators in multiple paths. Thus, a single confounder could not characterize all the confounding in these paths. As a solution, we further derive a general OVB formula with vector responses, regressors, and unmeasured confounders based on $R$ parameters. 

Then, we apply the general OVB formula to develop a sensitivity analysis method for the Baron--Kenny approach with multiple unmeasured confounders. We extend the definition of robustness values, where the strength of $R$ parameters are measured by their spectral norms, because our $R$ parameters can be matrices with multiple mediators and multiple unmeasured confounders. We provide methods to compute the robustness values. To the best of our knowledge, sensitivity analysis with vector unmeasured confounders has not been done before even for the simple least squares with vector responses, let alone the more complicated mediation analysis. 

Last but not least, we show that all our sensitivity analysis methods are {\it{sharp}}. That is, given the observables, the sensitivity parameters can freely take arbitrary values in their feasible regions. 

\subsection{Overview of our sensitivity analysis method for practitioners}\label{sec::introduction-overview}

The Baron--Kenny approach targets direct and indirect effects under the linear structural models. Based on the observed data, we can obtain the point estimates and $t$ statistics for direct and indirect effects. With unmeasured confounders, the result could be biased due to the OVB. We express direct and indirect effects by moments of observables and $R$ parameters. The moments of observables can be estimated by the observed data. The $R$ parameters measure the strength of association between the unmeasured confounder and the exposure, the mediators, and the outcome, respectively. Our proposed sensitivity analysis methods aim to solve the following problems: 

\begin{enumerate}
\item[(P1)] Given the upper bound of the $R$ parameters, report the worst point estimate and $t$ statistic for the direct or indirect effect. 
\item[(P2)] Report the lower bound of the $R$ parameters, such that the absolute $t$ statistic for the direct or indirect effect could reduce to 1.96 or 0. We define the lower bound as the ``robustness value for mediation''. 
\end{enumerate}

(P1) and (P2) have been standard in non-mediation settings \citep[e.g.,][]{ding2016sensitivity, cinelli2020making}. We propose analogous methods for mediation analysis based on the Baron--Kenny approach. All the methods discussed in this paper are implemented in the R package \texttt{BaronKennyU}, with guidance provided in Section \ref{sec::package}. The complete replication files for this paper are provided at: \url{https://github.com/mingrui229/Sensitivity-analysis-for-Baron-Kenny}. 



\subsection{Notation and definitions}\label{subsec::introduction-notation}

Let $\mathrm{dim}(\cdot)$ denote the dimension of a vector. For vector $x$, let $x_j$ denote the $j$th coordinate of $x$, and $x_{-j}$ denote the vector after removing the $j$th coordinate of $x$. Let $I$ denote the identity matrix with the dimension implicit and clear in the context. For positive semi-definite matrix $A$, let $A^{1/2}$ denote its symmetric square root. We use $\mathrm{diag}(\cdot)$ to denote a diagonal matrix formed by keeping only the diagonal elements of a square matrix, while setting all off-diagonal elements to 0. 

Let $\|\cdot\|_{\infty}$ denote the maximum of the absolute values of a vector or a matrix. Let $\| \cdot \|_2$ denote the Euclidean norm of a vector. We also use $\|\cdot\|_2$ to denote the matrix norm induced by the Euclidean norm, i.e., $\|A\|_2=\sup_{x\not=0}\|Ax\|_2/\|x\|_2$ for matrix $A$, which is also known as the spectral norm. When $A$ reduces to a vector, the spectral norm of $A$ reduces to the Euclidean norm of $A$. Let $\mathbb{B}_d$ denote the open unit ball of $d$-dimensional vectors with respect to the Euclidean norm, i.e., $\mathbb{B}_d=\{x\in \mathbb{R}^d:\|x\|_2<1\}$. Let $\partial \mathbb{B}_d$ denote the boundary of $d$-dimensional unit ball with respect to the Euclidean norm, i.e., $\partial \mathbb{B}_d=\{x\in\mathbb{R}^d:\|x\|_2=1\}$, and $\overline{\mathbb{B}}_d$ denote the closed $d$-dimensional unit ball with respect to the Euclidean norm, i.e., $\overline{\mathbb{B}}_d=\{x\in\mathbb{R}^d:\|x\|_2\leq 1\}$. Let $\mathbb{B}_{d_1\times d_2}$ denote the open unit ball of $d_1\times d_2$ matrices with respect to the spectral norm, i.e., $\mathbb{B}_{d_1\times d_2}=\{x\in \mathbb{R}^{d_1\times d_2}:\|x\|_2<1\}$. Let $\mathbb{S}_n^{++}$ denote the sets of $n\times n$ positive definite matrices. 

Let $\mathrm{corr}(\cdot,\cdot)$ denote the correlation coefficient of two random variables. Let $\hat{\mathrm{var}}(\cdot)$ denote the sample variance of a vector, $\hat{\mathrm{cov}}(\cdot)$ denote the sample covariance matrix of a data matrix, and $\hat{\mathrm{corr}}(\cdot,\cdot)$ denote the sample correlation coefficient of two scalars. Let $\ind$ denote independence of random variables.

We review the definition of least squares below. For least squares, we use lowercase letters to denote random variables, uppercase letters to denote the observed vectors or matrices for independent samples from the corresponding random variables, and Greek letters to denote the coefficients in regressions, respectively. We assume that all random vectors are column vectors and the coefficients from least squares are vectors or matrices with suitable dimensions. 

\begin{definition}
\label{def::least-squares}

(i) For scalar $y$ and vector $x$, the population least squares regression of $y$ on $x$ is $y=\beta^{\T} x+y_{\perp x}$, where $\beta=\arg\min_bE\{(y-b^{\T} x)^2\}$, and $y_{\perp x}  = y - \beta^{\T} x$ denotes the population residual. 
If $y$ is a vector, then $\beta$ is a matrix that minimizes $E\{ \| y-b^{\T} x \|_2^2\}$.

(ii) For vector $Y$ and matrix $X$, the sample least squares regression of $Y$ on $X$ is $Y=X\hat{\beta}+Y_{\perp X}$, where $\hat{\beta}=\arg\min_b\|Y-Xb\|_2^2$, and $Y_{\perp X} = Y - X\hat{\beta}$ denotes the sample residual vector. 
If $Y$ is a matrix, then  $\hat{\beta}$ is a matrix with each column corresponding to the sample least squares coefficient of the corresponding column of $Y$ on $X$.  
\end{definition}

Definition \ref{def::least-squares} above does not assume that the linear model is correct. We also review the definition of the population $R^2$ and population partial $R^2$ in Definition \ref{def::traditional-r2} below. 

\begin{definition}
\label{def::traditional-r2}
For scalar $y$ and vectors $x$ and $z$, define the population $R^2$ of $y$ on $x$ as $R_{y\sim x}^2=1-\mathrm{var}(y_{\perp x})/\mathrm{var}(y)$ if $\mathrm{var}(y)\not=0$, and the population partial $R^2$ of $y$ on $x$ given $z$ as $R_{y\sim x\mid z}^2=1-\mathrm{var}(y_{\perp x,z})/\mathrm{var}(y_{\perp z})$ if $\mathrm{var}(y_{\perp z})\not=0$. 
\end{definition}

\section{Omitted-variable bias in least squares}\label{sec::ovb}

The Baron--Kenny approach for mediation analysis builds on linear structural equation models. To deal with unmeasured confounding for the Baron--Kenny approach, we first develop theoretical tools for the simpler least squares problem. Section \ref{subsec::ovb-r2} reviews the $R^2$ parameterization of the omitted-variable bias (OVB) formula from \citet{cinelli2020making}. Their formula applies only when the outcome, exposure, and unmeasured confounder are scalars, and is not general enough for the Baron–Kenny approach with multiple mediators. To extend their framework, Section \ref{subsec::new-r} introduces a novel $R$ measure for multivariate partial correlation. Building on the $R$ measure, Sections \ref{subsec::ovb-r} and \ref{subsec::ovb-vector-u} develop an OVB formula that accommodates vector outcome and exposure when the unmeasured confounder is scalar and vector, respectively. 

\subsection{Review of \citet{cinelli2020making}'s OVB with the $R^2$ parameterization}\label{subsec::ovb-r2}
We review \citet{cinelli2020making}'s key result on OVB based on the $R^2$ parametrization. Let $y$ denote the outcome, $a$ the exposure, $c$ the observed covariates including the constant $1$, and $u$ the unmeasured confounder. Assume $y,a,u$ are scalars and $c$ can be a vector. Consider the long and short population least squares regressions with and without $u$:
\begin{equation}\label{eq::ols-long-short}
\begin{split}
y&=\gamma_1 a+\gamma_2^{\T} c+\gamma_3u+y_{\perp a,c,u},\\
y&=\tilde{\gamma}_1a+\tilde{\gamma}_2^{\T} c+y_{\perp a,c}.
\end{split}
\end{equation}
The difference between the coefficients of $a$ in the long and short regressions in \eqref{eq::ols-long-short} equals 
\begin{equation}
\label{eq::ovb-cochran}
\tilde{\gamma}_1-\gamma_1=\gamma_3\cdot \rho_1,
\end{equation}
where $\rho_1$ is the coefficient of $a$ in the population least squares regression of $u$ on $(a,c)$. The result in \eqref{eq::ovb-cochran} is called Cochran's formula in statistics \citep{cochran1938omission, cox2007generalization} and the OVB formula in econometrics \citep{angrist2008mostly}. It has a simple form and is useful for sensitivity analysis with respect to unmeasured confounding. \citet{vanderweele2010bias} derives a similar formula without the linear forms in \eqref{eq::ols-long-short}. However, it can be difficult to interpret the regression coefficients $\gamma_3$ and $\rho_1$ in \eqref{eq::ovb-cochran} because $u$ is unmeasured with an unknown scale.

\citet{cinelli2020making} propose a novel $R^2$ parametrization of the OVB: 

\begin{equation}
\label{eq::ovb-r2}
\tilde{\gamma}_1 - \gamma_1 = \mathrm{sgn}(\tilde{\gamma}_1 - \gamma_1) \left(\frac{R_{y\sim u\mid a,c}^2 R_{a\sim u\mid c}^2}{1- R_{a\sim u\mid c}^2} \right)^{1/2} \left\{ \frac{\mathrm{var}(y_{\perp a,c})}{\mathrm{var}(a_{\perp c})} \right\}^{1/2}.
\end{equation}
The formula \eqref{eq::ovb-r2} states that the difference $\gamma_1-\tilde{\gamma}_1$ depends on three types of quantities. First, it depends on the variances of the residuals $\mathrm{var}(y_{\perp a,c})$ and $\mathrm{var}(a_{\perp c})$, which are determined by the moments of the observables. Second, it depends on two $R^2$ parameters: $R_{a\sim u\mid c}^2$ measures the confounder-exposure partial correlation given observed covariates $c$, and $R_{y\sim u\mid a,c}^2$ measures the confounder-outcome partial correlation given the exposure $a$ and observed covariates $c$. Third, it depends on the sign of the bias $\mathrm{sgn}(\tilde{\gamma}_1 - \gamma_1)$. Since $u$ is unobserved, the bias $\tilde{\gamma}_1 - \gamma_1 $ can be positive or negative. The formula \eqref{eq::ovb-r2} is more interpretable than \eqref{eq::ovb-cochran}. Even with an unknown scale of $u$, the $R^2$ is bounded between $0$ and $1$ and measures the proportion of variability explained by $u$. \citet{imbens2003sensitivity} and \citet{blackwell2014selection} also propose to use the $R^2$-type sensitivity parameters for other settings. 

The OVB formula \eqref{eq::ovb-r2} assumes that both $y$ and $a$ are scalar, and it cannot be immediately deployed to perform sensitivity analysis in the Baron--Kenny approach with multiple mediators. In the following part of this section, we extend the OVB formula to allow for both $y$ and $a$ to be vectors, based on novel $R$ parameters. 

\subsection{Introducing $R$ parameters: measuring multivariate partial correlation}\label{subsec::new-r}

We propose a novel $R$ measure for developing our sensitivity analysis results. It measures the correlation between two centered vectors, possibly after controlling for a third vector. We give the formal definition of the $R$ measure in Definition \ref{def::new-r} below. 
\begin{definition}
\label{def::new-r}
For centered vectors $y$ and $x$, define 
$$R_{y\sim x}=\mathrm{cov}(y)^{-1/2}\mathrm{cov}(y,x)\mathrm{cov}(x)^{-1/2}$$
if $\mathrm{cov}(y)$ and $\mathrm{cov}(x)$ are invertible. For vectors $y,x,z$, where $z$ includes the constant $1$, define 
$$R_{y\sim x\mid z}=\mathrm{cov}(y_{\perp z})^{-1/2}\mathrm{cov}(y_{\perp z},x_{\perp z})\mathrm{cov}(x_{\perp z})^{-1/2}$$
if $\mathrm{cov}(y_{\perp z})$ and $\mathrm{cov}(x_{\perp z})$ are invertible.  
\end{definition}
Consider the special case when $y$ and $x$ are scalars in Definition \ref{def::new-r}. Then, $R_{y\sim x}$ reduces to the population Pearson correlation coefficient of $y$ and $x$, and $R_{y\sim x\mid z}$ reduces to the population partial correlation coefficient of $y$ and $x$ after controlling for $z$. Moreover, the squared $R_{y\sim x}$ and $R_{y\sim x\mid z}$ in Definition \ref{def::new-r} are $R_{y\sim x}^2$ and $R_{y\sim x\mid z}^2$ in Definition \ref{def::traditional-r2}, respectively. 

\citet{cinelli2020making} also use the same notation $R$ for the partial correlation coefficient in their discussion, but our Definition \ref{def::new-r} allows for vectors $y$ and $x$. When $y$ and $x$ are vectors, $R_{y\sim x}$ or $R_{y\sim x\mid z}$ is a matrix with dimension $\mathrm{dim}(y)\times \mathrm{dim}(x)$, satisfying $R_{y\sim x}=R_{x\sim y}^{\T}$ and $R_{y\sim x\mid z}=R_{x\sim y\mid z}^{\T}$. 

To unify the definition of $R_{y\sim x}$ and $R_{y\sim x\mid z}$ in Definition \ref{def::new-r}, we can write $R_{y\sim x\mid z}=R_{y_{\perp z}\sim x_{\perp z}}$, where $y_{\perp z}$ and $x_{\perp z}$ are centered residual vectors. By the property of least squares, $R_{y\sim x\mid z}^2=R_{y_{\perp z}\sim x_{\perp z}}^2$, for $R^2$ and partial $R^2$ in Definition \ref{def::traditional-r2}. Thus, we can focus on the discussion of the $R$ measure of two centered vectors and the population $R^2$. Lemma \ref{lem::r2-and-r} states that the Euclidean norm or spectral norm of the $R$ measure in Definition \ref{def::new-r} is related to the $R^2$ in Definition \ref{def::traditional-r2}. 

\begin{lemma}
\label{lem::r2-and-r}
For centered scalar $y$ and centered vector $x$, we have $\|R_{y\sim x}^{\T}\|_2^2=R_{y\sim x}^2$, if $\mathrm{var}(y)\not=0$ and $\mathrm{cov}(x)$ is invertible. For centered vectors $y$ and $x$, $\|R_{y\sim x}\|_2^2$ is the maximum of $R_{\alpha^{\T} y\sim x}^2$ over all linear combinations $\alpha^{\T} y\not=0$, if $\mathrm{cov}(y)$ and $\mathrm{cov}(x)$ are invertible. 
\end{lemma}

By Lemma \ref{lem::r2-and-r}, when $y$ is a scalar, $\|R_{y\sim x}\|_2^2$ represents the maximum proportion of the variance in $y$ that can be explained by linear predictors of $x$. When $y$ is a vector, $\|R_{y\sim x}\|_2^2$ represents the maximum proportion of the variance in linear combinations of $y$ that can be explained by linear predictors of $x$. Therefore, in Section \ref{subsec::scale-u-report} and Section \ref{subsec::vector-u-report}, we use the Euclidean norm or spectral norm of an $R$ parameter to measure its strength of correlation when we report our sensitivity analysis results based on upper bounds of $R$ parameters. Relatedly, $\|R_{y\sim x}\|_2^2$ or $\|R_{x\sim y}\|_2^2$ equals the maximal squared canonical correlation coefficient of $y$ and $x$ \citep{anderson1984introduction}, which gives an alternative interpretation of the spectral norm of $R$ parameters. 

Our $R$ parameters are closely related to marginal correlation coefficients. In Section \ref{sec::marginal-cc}, we will provide a detailed discussion on the relationship between $R_{y\sim x}$ and the marginal correlation coefficients of $y$ and $x$, as well as the use of marginal correlation coefficients as sensitivity parameters.

\subsection{A novel OVB formula with a single unmeasured confounder based on the $R$ parameterization}\label{subsec::ovb-r}

Now, we present a novel OVB formula based on the $R$ parametrization. We allow the outcome $y$ and exposure $a$ to be vectors but still assume a scalar unmeasured confounder $u$. We consider the long and short population least squares regressions with and without $u$: 
\begin{equation}\label{eq::multi-ols-long-short}
\begin{split}
y&=\gamma_1 a+\gamma_2 c+\gamma_3u+y_{\perp a,c,u},\\
y&=\tilde{\gamma}_1 a+\tilde{\gamma}_2 c+y_{\perp a,c}.
\end{split}
\end{equation}
The coefficients $\gamma_j$'s and $\tilde{\gamma}_j$'s in \eqref{eq::multi-ols-long-short} are matrices with suitable dimensions. Similar to \eqref{eq::ovb-cochran} with matrix coefficients, the difference between the coefficients of $a$ in the long and short regressions in \eqref{eq::multi-ols-long-short} equals 
\begin{equation}
\label{eq::ovb-cochran-extend}
\tilde{\gamma}_1-\gamma_1=\gamma_3\rho_1,
\end{equation}
where $\rho_1$ is the coefficient of $a$ in the population least squares regression of $u$ on $(a,c)$. Since $u$ is a scalar, $\gamma_3$ and $\rho_1$ have rank one, so the difference $\tilde{\gamma}_1-\gamma_1$ must be a rank-one matrix. Theorem \ref{thm::multi-ols} below extends \citet{cinelli2020making} to represent $\tilde{\gamma}_1-\gamma_1$ by the $R$ parameters. 

\begin{theorem}\label{thm::multi-ols}
Consider the regressions in \eqref{eq::multi-ols-long-short} with possibly vector $y$, $a$ and $c$ but scalar $u$. Assume the covariance matrix of $(y,a,u)_{\perp c}$ is invertible. 

(i) We have
\begin{equation}
\label{eq::ovb-vector-ya}
\gamma_1=\tilde{\gamma}_1-\frac{\mathrm{cov}(y_{\perp a,c})^{1/2}R_{y\sim u\mid a,c}R_{a\sim u\mid c}^{\T}\mathrm{cov}(a_{\perp c})^{-1/2}}{(1-\|R_{a\sim u\mid c}\|_2^2)^{1/2}}.
\end{equation}

(ii) Given the observables $\{y,a,c\}$, the parameters $\{R_{y\sim u\mid a,c}$, $R_{a\sim u\mid c}\}$ can take arbitrary values in $\mathbb{B}_{\mathrm{dim}(y)}\times \mathbb{B}_{\mathrm{dim}(a)}$, recalling the notation of unit balls $\mathbb{B}_d$ in Section \ref{sec::introduction}. 
\end{theorem}

Theorem \ref{thm::multi-ols}(i) states that the difference $\tilde{\gamma}_1-\gamma_1$ depends on two types of quantities. First, it depends on the covariances of residuals $\mathrm{cov}(y_{\perp a,c})$ and $\mathrm{cov}(a_{\perp c})$, which are determined by the moments of the observables. Second, it depends on two $R$ parameters, $R_{a\sim u\mid c}$ and $R_{y\sim u\mid a,c}$, which measure the confounder-exposure partial correlation and confounder-outcome partial correlation, respectively. When both $y$ and $a$ are scalars, $R_{y\sim u\mid a,c}$ and $R_{a\sim u\mid c}$ are scalars, and thus the formula simplifies to 
\begin{equation}\label{eq::thm2-scalar-ay}
\gamma_1=\tilde{\gamma}_1-\frac{R_{y\sim u\mid a,c}R_{a\sim u\mid c}}{(1-R_{a\sim u\mid c}^2)^{1/2}}\cdot\left\{\frac{\mathrm{var}(y_{\perp a,c})}{\mathrm{var}(a_{\perp c})}\right\}^{1/2}.
\end{equation}
Formula \eqref{eq::thm2-scalar-ay} is slightly different from formula \eqref{eq::ovb-r2}. Formula \eqref{eq::thm2-scalar-ay} no longer needs the sign parameters because $R_{y\sim u\mid a,c}$ and $R_{a\sim u\mid c}$ contain additional sign information that determines $\mathrm{sgn}(\tilde{\gamma}_1-\gamma_1)$ in formula \eqref{eq::ovb-r2}.

Theorem \ref{thm::multi-ols}(ii) shows that \eqref{eq::ovb-vector-ya} is {\it{sharp}} in the sense that given the observed random variables, the two parameters $R_{y\sim u\mid a,c}$ and $R_{a\sim u\mid c}$ can freely take arbitrary values in $\mathbb{B}_{\mathrm{dim}(y)}$ and $\mathbb{B}_{\mathrm{dim}(a)}$, respectively. By Lemma \ref{lem::r2-and-r}, the Euclidean norm or spectral norm of any $R$ parameter must be less than $1$. Therefore, the two parameters $R_{y\sim u\mid a,c}$ and $R_{a\sim u\mid c}$ can freely take arbitrary feasible values. Our notion of sharpness implies that the sensitivity analysis procedure is not conservative under any feasible specification of the sensitivity parameters.

Since the parameters $\{R_{y\sim u\mid a,c},R_{a\sim u\mid c}\}$ depend on the unmeasured covariate $u$, we cannot estimate them based on the observed data and must treat them as sensitivity parameters. With known $R_{y\sim u\mid a,c}$ and $R_{a\sim u\mid c}$, Theorem \ref{thm::multi-ols} motivates the following plug-in estimator
\begin{equation}\label{eq::plug-in-estimator}
\hat{\gamma}_1=\hat{\gamma}_1^{\text{obs}}-\frac{\hat{\mathrm{cov}}(Y_{\perp A,C})^{1/2}R_{y\sim u\mid a,c}R_{a\sim u\mid c}^{\T}\hat{\mathrm{cov}}(A_{\perp C})^{-1/2}}{(1-\|R_{a\sim u\mid c}\|_2^2)^{1/2}},
\end{equation}
where $\hat{\gamma}_1^{\text{obs}}$ is the coefficients of $A$ of the sample least squares fit of $Y$ on $(A,C)$. Given the sensitivity parameters $\{R_{y\sim u\mid a,c},R_{a\sim u\mid c}\}$, the estimator $\hat{\gamma}_1$ in \eqref{eq::plug-in-estimator} is consistent for $\gamma_1$ and asymptotically normal. We can use the nonparametric bootstrap for statistical inference, because \eqref{eq::plug-in-estimator} is a function of sample moments given the sensitivity parameters. Varying these sensitivity parameters within a plausible region yields a range of sensitivity bounds. 


One may wonder whether Theorem \ref{thm::multi-ols}  is really needed for a vector $y$, given the classic result that least squares with vector outcome $y$ reduces to separate least squares for each component of $y$. Indeed, we can apply \eqref{eq::thm2-scalar-ay} to each coordinate of $y$ to conduct marginal sensitivity analysis. However, if we want to conduct a joint sensitivity analysis across coordinates of $y$, we must use Theorem \ref{thm::multi-ols}. Otherwise, the $R$ parameters $\{R_{y_j\sim u\mid a,c}:1\leq j\leq \mathrm{dim}(y)\}$ in the marginal sensitivity analysis cannot take arbitrary values in $(-1,1)^{\mathrm{dim}(y)}$ in general, due to the correlation of $y$'s. Put another way, Theorem \ref{thm::multi-ols} ensures a sharp joint sensitivity analysis. More crucially, the formula in Theorem \ref{thm::multi-ols} serves as the stepstone for sensitivity analysis with vector mediator $m$ in Section \ref{sec::scale-u} below.

\subsection{OVB formula with multiple unmeasured confounders}\label{subsec::ovb-vector-u}
Now, we present an even more general OVB formula with multiple unmeasured confounders. We allow the unmeasured confounder $u$ to be vector. We consider the long and short population least squares regressions with and without $u$: 
\begin{equation}\label{eq::most-general-long-short-ols}
\begin{split}
y&=\gamma_1a+\gamma_2c+\gamma_3u+y_{\perp a,c,u},\\
y&=\tilde{\gamma}_1a+\tilde{\gamma}_2c+y_{\perp a,c},\\
\end{split}
\end{equation}
where all of $y,a,c,u$ are allowed to be vectors. Again, the formula \eqref{eq::ovb-cochran-extend} still holds when $u$ is a vector, but unlike Section \ref{subsec::ovb-r}, the difference $\tilde{\gamma}_1-\gamma_1$ may not be a rank-one matrix. This gives the mathematical reason why we cannot simply assume $u$ is a scalar in many multivariate problems. Theorem \ref{thm::most-general-ovb} below extends Theorem \ref{thm::multi-ols} with vector $u$ to represent $\tilde{\gamma}_1-\gamma_1$ by the $R$ parameters. 

\begin{theorem}
\label{thm::most-general-ovb}
Consider the regressions in \eqref{eq::most-general-long-short-ols} with possibly vector $y,a,c,u$. Assume the covariance matrix of $(y,a,u)_{\perp c}$ is invertible.

 (i) We have
 $$\gamma_1=\tilde{\gamma}_1-\mathrm{cov}(y_{\perp a,c})^{1/2}R_{y\sim u\mid a,c}\mathrm{cov}(u_{\perp a,c})^{-1/2}\mathrm{cov}(u_{\perp c})^{1/2}R_{a\sim u\mid c}^{\T}\mathrm{cov}(a_{\perp c})^{-1/2},$$
 where
 $$\mathrm{cov}(u_{\perp a,c})=\mathrm{cov}(u_{\perp c})-\mathrm{cov}(u_{\perp c})^{1/2}R_{a\sim u\mid c}^{\T} R_{a\sim u\mid c}\mathrm{cov}(u_{\perp c})^{1/2}.$$

(ii) Given observables $\{y,a,c\}$, the parameters $\{R_{y\sim u\mid a,c},R_{a\sim u\mid c},\mathrm{cov}(u_{\perp c})\}$ can take arbitrary values in $\mathbb{B}_{\mathrm{dim}(y)\times\mathrm{dim}(u)}\times \mathbb{B}_{\mathrm{dim}(a)\times\mathrm{dim}(u)}\times \mathbb{S}_{\mathrm{dim}(u)}^{++}$. 
\end{theorem}

Theorem \ref{thm::most-general-ovb}(i) states that the difference $\tilde{\gamma}_1-\gamma_1$ depends on three types of quantities: moment of observables, $R$ parameters involving $u$, and $\mathrm{cov}(u_{\perp c})$. The quantity $\mathrm{cov}(u_{\perp c})$ in Theorem \ref{thm::most-general-ovb}(i) is new compared with Theorem \ref{thm::most-general-ovb}(i). If $u$ is a scalar, then $\mathrm{cov}(u_{\perp a,c})^{-1/2}\mathrm{cov}(u_{\perp c})^{1/2}$ reduces to $(1-\|R_{a\sim u\mid c}\|_2^2)^{-1/2}$, and Theorem \ref{thm::most-general-ovb}(i) reduces to Theorem \ref{thm::multi-ols}(i). However, if $u$ is a vector, the quantity $\mathrm{cov}(u_{\perp c})$ is involved. This is mainly because $R_{y\sim u\mid a,c}$ and $R_{a\sim u\mid c}$ are invariant to any invertible linear transformation of $u$ if and only if $u$ is a scalar. Theorem \ref{thm::most-general-ovb}(ii) states the sharpness. 

Since $u$ is unobserved, we need to specify the sensitivity parameters $\{R_{y\sim u\mid a,c},R_{a\sim u\mid c},\mathrm{cov}(u_{\perp c})\}$ to apply Theorem \ref{thm::most-general-ovb} and obtain the point estimate of $\gamma_1$. However, it is difficult to specify these sensitivity parameters, especially when the dimension of $u$ is undetermined. Fortunately, Theorem \ref{thm::most-general-ovb} is useful when we want to compute the sensitivity bound, because the worst point estimate or $t$ statistic under the upper bound of $R$ parameters does not depend on the choice of $\mathrm{cov}(u_{\perp c})$; see Proposition \ref{prop::assume-cov-u-ols} in the Supplementary Material for the precise statement.

Here the worst point estimate (or $t$ statistic) means the minimum point estimate (or $t$ statistic) when the observed point estimate (or $t$ statistic) is positive, and the maximum point estimate (or $t$ statistic) when the observed point estimate (or $t$ statistic) is negative. To simplify the presentation, we use the terminology ``upper/lower bound of $R$ parameters'' to refer to ``upper/lower bound of spectral norms of $R$ parameters''. Recall Lemma \ref{lem::r2-and-r} for an interpretation of spectral norms of $R$ parameters. We also use the terms ``spectral norm of $R$ parameters'' and ``strength of $R$ parameters'' interchangeably. 

 As a more concrete example, we will show how to compute the worst-case result for the Baron--Kenny approach with multiple unmeasured confounders in Section \ref{sec::vector-u}. To simplify the presentation and highlight the key concepts, we will start with the case with a single unmeasured confounder in Section \ref{sec::scale-u} below.

\section{Sensitivity analysis for the Baron--Kenny approach with a single unmeasured confounder}\label{sec::scale-u}

Equipped with the theoretical foundations in Section \ref{sec::ovb}, we now turn to the problem of sensitivity analysis for the Baron--Kenny approach. To fix ideas, we start with the case with a scalar unmeasured confounder $u$. Section \ref{subsec::scale-u-bk} reviews the linear structural models for the Baron--Kenny approach. Section \ref{subsec::scale-u-parameter} applies the OVB formula in Section \ref{subsec::ovb-r} to represent the direct and indirect effects in terms of the \textit{natural} $R$ parameters, which correspond to the natural factorization of the joint distribution based on the direct acyclic graph in Figure \ref{fg::DAG-mediation-confounding}. Section \ref{subsec::scale-u-report} discusses how to report the sensitivity analysis results to answer (P1) and (P2) in Section \ref{sec::introduction-overview}. Section \ref{subsubsec::fb} discusses the strategies of calibrating sensitivity parameters, a fundamental challenge in sensitivity analysis.

\subsection{Baron--Kenny approach with multiple mediators}\label{subsec::scale-u-bk}

\begin{figure} 
\begin{subfigure}{0.50\textwidth}
$$
\begin{xy}
\xymatrix{
&  & c \ar[dll] \ar[d] \ar[drr] \\
a \ar[rr] \ar@/_1.5pc/[rrrr] & & m\ar[rr] & & y}
\end{xy}
$$
\caption{without unmeasured confounding}
\label{fg::DAG-mediation}
\end{subfigure}\hspace*{\fill}
\begin{subfigure}{0.50\textwidth}
$$
\begin{xy}
\xymatrix{
&  & u \ar[dll] \ar[d] \ar[drr] \\
a \ar[rr] \ar@/_1.5pc/[rrrr] & & m\ar[rr] & & y}
\end{xy}
$$
\caption{with unmeasured $u$, conditional on $c$}
\label{fg::DAG-mediation-confounding}
\end{subfigure}
\caption{Directed acyclic graphs for mediation analysis.} 
\end{figure}
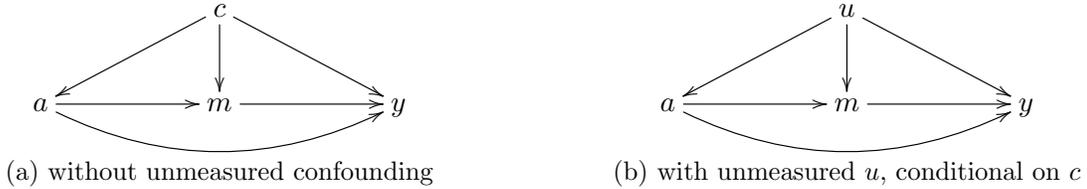

Without unmeasured confounding, Figure \ref{fg::DAG-mediation} illustrates the relationship among the exposure $a$, outcome $y$, mediator $m$, and observed covariates $c$ that include the constant $1$. We allow $m$ to be a vector, and all the discussion can be applied to the case with a scalar $m$. Mediation analysis quantifies the indirect effect of the exposure $a$ on the outcome $y$ that acts through the mediator $m$ and the direct effect that acts independently through other pathways. The Baron--Kenny approach to mediation analysis is based on the following two population least squares regressions
\begin{equation}
\label{eq::multi-med-long-short2}
\begin{alignedat}{3}
m&=\tilde{\beta}_1a+\tilde{\beta}_2c+m_{\perp a,c},\\
y&=\tilde{\theta}_1a+\tilde{\theta}_2^{\T} c+\tilde{\theta}_3^{\T} m+y_{\perp a,c,m}, 
\end{alignedat}
\end{equation}
and then interprets $\tilde{\theta}_1$ as the direct effect and $\tilde{\theta}_3^{\T} \tilde{\beta}_1$ as the indirect effect, respectively. If we further fit the short regression of the outcome on the exposure and covariates without the mediator
$$
y=\tilde{\gamma}_1a+\tilde{\gamma}_2^{\T} c+y_{\perp a,c},
$$
then Cochran's formula in \eqref{eq::ovb-cochran} implies the indirect effect also equals $\tilde{\theta}_3^{\T} \tilde{\beta}_1=\tilde{\gamma}_1-\tilde{\theta}_1$. Run sample least squares to obtain the estimators of these coefficients, with the ``hat'' notation and the superscript ``obs'' signifying the estimates. We can estimate the direct effect by $\hat{\theta}_{1}^{\text{obs}}$. We can estimate the indirect effect by either $(\hat{\theta}_{3}^{\text{obs}})^{\T} \hat{\beta}_{1}^{\text{obs}}$ or $\hat{\gamma}_{1}^{\text{obs}}-\hat{\theta}_{1}^{\text{obs}}$, with the former known as the {\it product method} and the latter known as the {\it difference method} for the indirect effect \citep{vanderweele2015explanation}.

With unmeasured confounding, Figure \ref{fg::DAG-mediation-confounding} illustrates how the causal relationship is confounded by $u$. We consider the population least squares regressions 
\begin{equation}
\label{eq::multi-med-long-short1}
\begin{alignedat}{3}
m&=\beta_1a+\beta_2c+\beta_3u+m_{\perp a,c,u},\\
y&=\theta_1a+\theta_2^{\T} c+\theta_3^{\T} m+\theta_4u+y_{\perp a,c,m,u},\\
\end{alignedat}
\end{equation}
and the true direct effect is $\theta_1$ and the true indirect effect is $\theta_3^{\T} \beta_1$ or $\gamma_1-\theta_1$ based on \eqref{eq::ols-long-short} and \eqref{eq::multi-med-long-short1}.

Our central goal is to develop a sensitivity analysis method for both the direct and indirect effects using the $R$ parameterization. The equivalence of the product method and the difference method leads to the equivalence of the corresponding sensitivity analysis. We focus on the product method in our discussion due to its simplicity. 

\subsection{Natural parameterization of direct and indirect effects}\label{subsec::scale-u-parameter}

\subsubsection{Direct effect}
We first discuss the direct effect. By Theorem \ref{thm::multi-ols}, we can express the true direct effect $\theta_1$ as
\begin{equation}
\label{eq::direct}
\theta_1=\tilde{\theta}_1-R_{y\sim u\mid a,m,c}\cdot \frac{R_{a\sim u\mid m,c}}{(1-R_{a\sim u\mid m,c}^2)^{1/2}}\cdot\left\{\frac{{\mathrm{var}}(y_{\perp a,m,c})}{{\mathrm{var}}(a_{\perp m,c})}\right\}^{1/2}. 
\end{equation}
Unfortunately, $R_{a\sim u\mid m,c}$ in \eqref{eq::direct} is not a parameter corresponding to the factorization of the joint distribution based on Figure \ref{fg::DAG-mediation-confounding}. It measures the partial correlation between $a$ and $u$ conditional on their common descendent $m$. \citet{vanderweele2010bias} and \citet{ding2016sharp} encounter similar difficulties in specifying the sensitivity parameters for mediation based on other approaches. They do not provide any solutions. Fortunately, we can express the term $R_{a\sim u\mid m,c}/(1-R_{a\sim u\mid m,c}^2)^{1/2}$ in \eqref{eq::direct} using other natural $R$ parameters corresponding to Figure \ref{fg::DAG-mediation-confounding}, as shown in Proposition \ref{prop::aucm-decom} below. 

\begin{proposition}\label{prop::aucm-decom}
Assume the covariance matrix of $(y,m,a,u)_{\perp c}$ is invertible for possibly vector $c$ but scalar $y,m,a,u$. We have
\begin{equation}\label{eq::multi-aucm-formula2}
\frac{R_{a\sim u\mid m,c}}{(1-R_{a\sim u\mid m,c}^2)^{1/2}}=\frac{R_{a\sim u\mid c}(1-\|R_{m\sim a\mid c}\|_2^2)^{1/2}}{(1-R_{a\sim u\mid c}^2)^{1/2}(1-\|R_{m\sim u\mid a,c}\|_2^2)^{1/2}}-\frac{R_{m\sim a\mid c}^{\T} \mathrm{cov}(m_{\perp c})^{-1/2}\mathrm{cov}(m_{\perp a,c})^{1/2}R_{m\sim u\mid a,c}}{(1-\|R_{m\sim a\mid c}\|_2^2)^{1/2}(1-\|R_{m\sim u\mid a,c}\|_2^2)^{1/2}}. 
\end{equation}
\end{proposition}

By Proposition \ref{prop::aucm-decom}, we can express $R_{a\sim u\mid m,c}/(1-R_{a\sim u\mid m,c}^2)^{1/2}$ using parameters $R_{a\sim u\mid c}$, $R_{m\sim u\mid a,c}$, and $R_{m\sim a\mid c}$. The parameter $R_{m\sim a\mid c}$ is estimable based on the observed data, whereas $R_{a\sim u\mid c}$ and $R_{m\sim u\mid a,c}$ depend on the unmeasured confounder $u$. Define 
\begin{equation}\label{eq::ru2}
R_u = \{R_{y\sim u\mid a,m,c},R_{m\sim u\mid a,c},R_{a\sim u\mid c}\}.
\end{equation}
These $R$ parameters in $R_u$ correspond to the natural factorization of the joint distribution based on Figure \ref{fg::DAG-mediation-confounding}, which facilitates the interpretability of our sensitivity analysis. We present the sharpness result in Corollary \ref{cor::direct} below.  
\begin{corollary}
\label{cor::direct}
Assume the covariance matrix of $(y,m,a,u)_{\perp c}$ is invertible for possibly vector $c$ but scalar $y,m,a,u$. Given the observables $\{y,m,a,c\}$, the parameters $\{R_{y\sim u\mid a,m,c},R_{m\sim u\mid a,c},R_{a\sim u\mid c}\}$ can take arbitrary values in $(-1,1)\times \mathbb{B}_{\mathrm{dim}(m)}\times (-1,1)$. 
\end{corollary}


\subsubsection{Indirect effect}
We next discuss the indirect effect based on the product method. By Theorem \ref{thm::multi-ols}, we can express $\beta_1$ and $\theta_3$ in the true indirect effect $\theta_3^{\T}\beta_1$ as
\begin{align}
\beta_1&=\tilde{\beta}_1-\frac{R_{a\sim u\mid c}}{(1-R_{a\sim u\mid c}^2)^{1/2}{\mathrm{var}}(a_{\perp c})^{1/2}}\cdot{\mathrm{cov}}(m_{\perp a,c})^{1/2}R_{m\sim u\mid a,c},\label{eq::multi-indirect-product1}\\
{\theta}_3&=\tilde{\theta}_3-\frac{R_{y\sim u\mid a,m,c}{\mathrm{var}}(y_{\perp a,c,m})^{1/2}}{(1-\|R_{m\sim u\mid a,c}\|_2^2)^{1/2}}\cdot{\mathrm{cov}}(m_{\perp a,c})^{-1/2}R_{m\sim u\mid a,c}.\label{eq::multi-indirect-product2}
\end{align}
The above formulas \eqref{eq::multi-indirect-product1} and \eqref{eq::multi-indirect-product2} are parameterized by natural parameters in $R_u$. The sharpness for indirect effect is also guaranteed by Corollary \ref{cor::direct} above. 

\subsubsection{Special case: randomized exposure}
As a special case, if the exposure is randomized under the assumption $a\ind u\mid c$, then the arrow from $u$ to $a$ disappears in Figure \ref{fg::DAG-mediation-confounding} and $R_{a\sim u\mid c}=0$. This simplifies our sensitivity analysis. For the direct effect, we can apply \eqref{eq::direct} and \eqref{eq::multi-aucm-formula2} to obtain 
\begin{equation}
\label{eq::random-experiment-multi-direct}
\theta_1=\tilde{\theta}_1+\frac{R_{y\sim u\mid a,m,c}{\mathrm{var}}(y_{\perp a,m,c})^{1/2}}{{\mathrm{var}}(a_{\perp m,c})^{1/2}}\cdot\frac{{R}_{m\sim a\mid c}^{\T}{\mathrm{cov}}(m_{\perp c})^{-1/2}{\mathrm{cov}}(m_{\perp a,c})^{1/2}R_{m\sim u\mid a,c}}{(1-\|{R}_{m\sim a\mid c}\|_2^2)^{1/2}(1-\|R_{m\sim u\mid a,c}\|_2^2)^{1/2}},
\end{equation}
with sensitivity parameters $\{R_{y\sim u\mid a,m,c},R_{m\sim u\mid a,c}\}$. For the indirect effect based on the product method, we have $\beta_1=\tilde{\beta}_1$ and $\theta_3$ in \eqref{eq::multi-indirect-product2} with sensitivity parameters $\{R_{y\sim u\mid a,m,c},R_{m\sim u\mid a,c}\}$.

\citet{imai2010identification} propose a sensitivity analysis method for mediation with a randomized exposure. Under linear structural equation models, their method is similar to ours, although they use a different parametrization. They impose homoskedasticity of the linear models, and their method is thus not as general as ours.

\subsection{Report the sensitivity analysis results}\label{subsec::scale-u-report}

With prespecified $R_u$, we can use the natural parameterization in Section \ref{subsec::scale-u-parameter} to compute the point estimates, standard errors, and $t$ statistics for the direct and indirect effects. Varying the sensitivity parameters, we can report a series of point estimates, standard errors, and $t$ statistics. Therefore, a straightforward way to report sensitivity analysis is to display a table or several tables of point estimates and confidence intervals corresponding to different combinations of the sensitivity parameters. This is a common strategy in sensitivity analysis in causal inference, e.g., \citet{Rosenbaum::1983JRSSB}. However, this approach can be cumbersome to implement in practice because it requires specifying all the sensitivity parameters and displaying a large number of tables or figures. Below, we focus on two simpler ways to report the sensitivity analysis results.

\subsubsection{Report the worst cases under upper bound of the $R$ parameters}\label{subsubsec::bound-r}
We can report the worst point estimate or $t$ statistic given the upper bound of $R$ parameters. To simplify the presentation, we assume the observed point estimate is positive throughout this section. With specified $0\leq \rho_y,\rho_m,\rho_a<1$, we compute the minimum point estimate and minimum $t$ statistic under the constraint 
\begin{equation}\label{eq::rho-yma}
R_{y\sim u\mid a,m,c}^2\leq \rho_y,\quad \|R_{m\sim u\mid a,c}\|_2^2\leq \rho_m,\quad R_{a\sim u\mid c}^2\leq \rho_a.
\end{equation}
In general, the point estimate or $t$ statistic may not be a monotone function of all of $R_{y\sim u\mid a,m,c}^2$, $\|R_{m\sim u\mid a,c}\|_2^2$, and $R_{a\sim u\mid c}^2$. Thus, the minimum $t$ statistic is not necessarily attained at $R_{y\sim u\mid a,m,c}^2=\rho_y$, $\|R_{m\sim u\mid a,c}\|_2^2=\rho_m$, and $R_{a\sim u\mid c}^2=\rho_a$. While explicit results can be derived by analyzing the point estimate and $t$ statistic for direct and indirect effects on a case-by-case basis, we focus below on a unified numerical approach to solving the optimization problem.

Based on the natural parameterization in Section \ref{subsec::scale-u-parameter}, we can re-express the direct and indirect effects as 
\begin{equation}\label{eq::direct-indirect-scalar-u}
\begin{split}
\theta_1&=\tilde{\theta}_1+T_1\phi_1+T_2^{\T}\phi_2,\\
\theta_3^{\T}\beta_1&=(\tilde{\theta}_3+T_3\phi_2)^{\T}(\tilde{\beta}_1+T_4\phi_3),
\end{split}
\end{equation}
where
\begin{equation}\label{eq::t1234}
\begin{split}
T_1&=-\frac{(1-\|R_{m\sim a\mid c}\|_2^2)^{1/2}\mathrm{var}(y_{\perp a,m,c})^{1/2}}{\mathrm{var}(a_{\perp m,c})^{1/2}},\\
T_2&=\frac{\mathrm{var}(y_{\perp a,m,c})^{1/2}\cdot \mathrm{cov}(m_{\perp a,c})^{1/2}\mathrm{cov}(m_{\perp c})^{-1/2}R_{m\sim a\mid c}}{(1-\|R_{m\sim a\mid c}\|_2^2)^{1/2}\mathrm{var}(a_{\perp m,c})^{1/2}},\\
T_3&=-\frac{\mathrm{cov}(m_{\perp a,c})^{1/2}}{\mathrm{var}(a_{\perp c})^{1/2}},\\
T_4&=-\mathrm{var}(y_{\perp a,c,m})^{1/2}\mathrm{cov}(m_{\perp a,c})^{-1/2}
\end{split}
\end{equation}
depend only on moments of observables, and 
\begin{equation}\label{eq::phi123}
\begin{split}
\phi_1&=\frac{R_{y\sim u\mid a,m,c}R_{a\sim u\mid c}}{(1-R_{a\sim u\mid c}^2)^{1/2}(1-\|R_{m\sim u\mid a,c}\|_2^2)^{1/2}},\\
\phi_2&=\frac{R_{y\sim u\mid a,m,c}\cdot R_{m\sim u\mid a,c}}{(1-\|R_{m\sim u\mid a,c}\|_2^2)^{1/2}},\\
\phi_3&=\frac{R_{a\sim u\mid c}\cdot R_{m\sim u\mid a,c}}{(1-R_{a\sim u\mid c}^2)^{1/2}}
\end{split}
\end{equation}
depend only on sensitivity parameters $R_u$. Proposition \ref{prop::scalar-u-range} below gives the range of $(\phi_1,\phi_2)$ for the direct effect and the range of $(\phi_2,\phi_3)$ for the indirect effect. 

\begin{proposition}\label{prop::scalar-u-range}
Under the constraint in \eqref{eq::rho-yma}, the range of $(\phi_1,\phi_2)$ is 
$$\left\{(x_1,x_2)\in \mathbb{R}\times \mathbb{R}^{\mathrm{dim}(m)}:x_1^2\leq \frac{\rho_a(\rho_y+\|x_2\|_2^2)}{1-\rho_a},\quad \|x_2\|_2^2\leq \frac{\rho_m\rho_y}{1-\rho_m}\right\},$$
and the range of $(\phi_2,\phi_3)$ is
$$\left\{(x_2,x_3)\in \mathbb{R}^{\mathrm{dim}(m)}\times \mathbb{R}^{\mathrm{dim}(m)}:\|x_2\|_2^2\leq \frac{\rho_m\rho_y}{1-\rho_m},\quad \|x_3\|_2^2\leq \frac{\rho_a\rho_m}{1-\rho_a}, \quad \text{there exists }t\in\mathbb{R} \text{ such that } x_2=tx_3\right\}.$$
\end{proposition}

Proposition \ref{prop::scalar-u-range} is the theoretical basis for computing the worst-case point estimates and $t$-statistics. For the point estimates, with estimated $T_j$ ($j=1,\ldots, 4$), the plug-in estimates for direct and indirect effects in  \eqref{eq::direct-indirect-scalar-u} are functions of $(\phi_1,\phi_2)$ and $(\phi_2,\phi_3)$, respectively. Given the $\phi$ terms in \eqref{eq::phi123} that depend only on the sensitivity parameters, the plug-in estimates for direct and indirect effects depend only on moments of observables, so we can use the nonparametric bootstrap to compute their standard errors and the corresponding $t$ statistics. Therefore, the $t$ statistics for direct and indirect effects are functions of $(\phi_1,\phi_2)$ and $(\phi_2,\phi_3)$, respectively. We can obtain the minimum point estimates or $t$ statistics for direct and indirect effects by solving optimization problems over $(\phi_1,\phi_2)$ and $(\phi_2,\phi_3)$, respectively, subject to the constraints in Proposition \ref{prop::scalar-u-range}. In both cases, the optimization problem involves variables of $(\mathrm{dim}(m)+1)$ degrees of freedom. 

In our R package, we implement the existing optimization solver NLopt to compute the worst-case point estimates and $t$-statistics. Our numerical tests show that NLopt outperforms other existing solvers, such as BARON, COBYLA, and Gurobi, in finding the global optimum for our specific problems. However, due to the complexity and non-convexity of the objective functions, NLopt may not find the global optimum. We recommend trying multiple initial values to improve the performance of NLopt.

\subsubsection{Report the robustness value for mediation analysis}\label{subsubsec::rv-scalar-u}

Researchers often ask the following questions: 
\begin{enumerate}
\item[(Q1)] What is the lower bound of the $R$ parameters such that the direct or indirect effect point estimate can be reduced to $0$?

\item[(Q2)] What is the lower bound of $R$ parameters such that the confidence interval for the direct or indirect effect can be altered to cover $0$?
\end{enumerate}
These questions motivate the development of the E-value for the risk ratio \citep{ding2016sensitivity, ding2016sharp, vanderweele2017sensitivity} and the ``robustness value'' in linear regression \citep{cinelli2020making}. 

To answer question (Q1), we need to find the minimal $\rho\geq 0$ such that 
\begin{equation}\label{eq::upper-bound-rho}
\max\{R_{y\sim u\mid a,m,c}^2,\|R_{m\sim u\mid a,c}\|_2^2,R_{a\sim u\mid c}^2\}\leq \rho
\end{equation} 
and the worst $t$ statistic is less than $0$. Analogously, to answer question (Q2), we need to find the minimal $\rho\geq 0$ such that \eqref{eq::upper-bound-rho} holds and the worst $t$ statistic is less than the $1-\alpha/2$ upper quantile of a standard normal random variable. For simplicity of description, we will focus on the threshold $1.96$ for $\alpha=0.05$. We have discussed how to obtain the minimum $t$ statistic under the constraint \eqref{eq::upper-bound-rho} in Section \ref{subsubsec::bound-r}. Here, we conduct the grid search to gradually increase $\rho$, and find the critical values of $\rho$ for the confidence interval and point estimate such that the $t$ statistic hits $1.96$ and $0$, respectively. We call the critical values the ``robustness values for mediation'' or simply the ``robustness values'' following \citet{cinelli2020making}. 

However, reporting the robustness values can be an overly conservative strategy when the $R$ parameters in \eqref{eq::upper-bound-rho} are not comparable in magnitude. If this kind of conservativeness is not desired in specific applications, then we recommend going back to the strategy in Section \ref{subsubsec::bound-r} or even going back to the strategy of reporting a series of estimates at the cost of specifying the sensitivity parameters.  

\subsection{Calibrate the sensitivity parameters}\label{subsubsec::fb}

As in all sensitivity analyses, it is fundamentally challenging to specify the threshold of the robustness values for mediation analysis without additional information. The threshold is often problem-specific. If a study has rich observed covariates, then researchers may not worry too much about OVB. On the contrary, if a study misses important covariates that may confound the exposure, mediator, and outcome, then researchers may specify a high threshold for the strength of $R$ parameters. Although we cannot solve this problem mathematically, we can report some observed sample $R^2$'s as the reference values. For instance, we can report the sample $R^2$'s, $\{\hat{R}_{y\sim a\mid m}^2, \hat{R}_{y\sim m\mid a}^2, \hat{R}_{m\sim a}^2\}$ when there is no observed covariate (See Section \ref{subsec::illustrations-single} below), and $\{\max_{j}\hat{R}_{y\sim c_j\mid a,m,c_{-j}}^2, \max_{j}\hat{R}_{m\sim c_j\mid a,c_{-j}}^2\}$ when there are observed covariates $c_j$'s (See Section \ref{subsec::illustrations-multiple} below).

To address the fundamental difficulty of specifying the range of the sensitivity parameters, \citet{cinelli2020making} propose an alternative strategy called {\it formal benchmarking}. We can also extend their strategy to the Baron--Kenny approach. Due to the space limit, we relegate the discussion to Section \ref{sec::formal-benchmarking} of the Supplementary Material and implement it in our R package.

\section{Sensitivity analysis for the Baron--Kenny approach with multiple unmeasured confounders}\label{sec::vector-u}

Previously, \citet{hosman2010sensitivity} and \citet{cinelli2020omitted} have shown that assuming a scalar unmeasured confounder $u$ is a conservative strategy in linear regression, even when the true $u$ is a vector. The intuition behind this is that we can always redefine a scalar $u$ as the linear combination that enters the outcome regression model. However, this intuition does not generalize to all sensitivity analysis problems. We have discussed the case with vector $y$ and vector $u$ in Section \ref{subsec::ovb-vector-u} before. This issue also arises in mediation analysis.  

We will show that when the exposure is nonrandomized and there are multiple mediators, assuming a scalar versus a vector unmeasured confounder $u$ can lead to different sensitivity analysis results for the indirect effect. We will also demonstrate that for all other cases, we can simply assume a scalar $u$. 

In this section, we allow $u$ to be a vector in \eqref{eq::multi-med-long-short1}, and consider the following long population least squares regression with $u$
\begin{equation}
\label{eq::multi-med-long-short1-vu}
\begin{alignedat}{3}
m&=\beta_1a+\beta_2c+\beta_3u+m_{\perp a,c,u},\\
y&=\theta_1a+\theta_2^{\T} c+\theta_3^{\T} m+\theta_4^{\T} u+y_{\perp a,c,m,u},\\
\end{alignedat}
\end{equation}
and short population least squares regression without $u$ in \eqref{eq::multi-med-long-short2}. We apply Theorem \ref{thm::most-general-ovb} to extend the results in Section \ref{sec::scale-u}. Section \ref{subsec::vector-u-parameter} extends Section \ref{subsec::scale-u-parameter} to parameterize direct and indirect effects using natural $R$ parameters $R_{y\sim u\mid a,m,c}$, $R_{m\sim u\mid a,c}$, $R_{a\sim u\mid c}$, and an additional nuisance parameter $\mathrm{cov}(u_{\perp c})$. Based on the natural parameterization, Section \ref{subsec::vector-u-report} extends Section \ref{subsec::scale-u-report} to discuss the strategies of reporting sensitivity analysis results. Section \ref{subsec::scalar-u-vs-vector-u} discusses when we cannot simply assume there is a single unmeasured confounder, based on results in Section \ref{subsec::scale-u-report} and Section \ref{subsec::vector-u-report}. 

\subsection{Natural parameterization of direct and indirect effects}\label{subsec::vector-u-parameter}


\subsubsection{Direct effect}

We first discuss the direct effect. By Theorem \ref{thm::most-general-ovb}, we can express the true direct effect $\theta_1$ as 
\begin{equation}\label{eq::direct-vu}
\theta_1=\tilde{\theta}_1-\mathrm{var}(y_{\perp a,m,c})^{1/2}R_{y\sim u\mid a,m,c}\mathrm{cov}(u_{\perp a,m,c})^{-1/2}\mathrm{cov}(u_{\perp c,m})^{1/2}R_{a\sim u\mid c,m}^{\T} \mathrm{var}(a_{\perp c,m})^{-1/2},
\end{equation}
where $\mathrm{cov}(u_{\perp a,m,c})=\mathrm{cov}(u_{\perp c,m})-\mathrm{cov}(u_{\perp c,m})^{1/2}R_{a\sim u\mid m,c}^{\T} R_{a\sim u\mid m,c}\mathrm{cov}(u_{\perp c,m})^{1/2}$. Similarly, an unnatural parameter $R_{a\sim u\mid m,c}$ appears in both \eqref{eq::direct-vu} and the relationship between $\mathrm{cov}(u_{\perp a,m,c})$ and $\mathrm{cov}(u_{\perp m,c})$. As a solution, we express the term $R_{a\sim u\mid c,m}\mathrm{cov}(u_{\perp c,m})^{1/2}$ using parameters $\{R_{a\sim u\mid c},R_{m\sim u\mid a,c},\mathrm{cov}(u_{\perp c}), \mathrm{cov}(u_{\perp a,c})\}$ in Proposition \ref{prop::aucm-vu} below. 

\begin{proposition}\label{prop::aucm-vu}
Assume the covariance matrix of $(y,m,a,u)_{\perp c}$ is invertible for possibly vector $c$ but scalar $y,m,a,u$. We have
\begin{equation}\label{eq::aucm-vu1}
\begin{split}
R_{a\sim u\mid c,m}\mathrm{cov}(u_{\perp c,m})^{1/2}=&(1-\|R_{m\sim a\mid c}\|_2^2)^{1/2}R_{a\sim u\mid c}\mathrm{cov}(u_{\perp c})^{1/2}\\
&-\frac{R_{m\sim a\mid c}^{\T} \mathrm{cov}(m_{\perp c})^{-1/2}\mathrm{cov}(m_{\perp a,c})^{1/2}R_{m\sim u\mid a,c}\mathrm{cov}(u_{\perp a,c})^{1/2}}{(1-\|R_{m\sim a\mid c}\|_2^2)^{1/2}}. 
\end{split}
\end{equation}
\end{proposition}

Moreover, we can express the additional terms $\mathrm{cov}(u_{\perp a,c})$ and $\mathrm{cov}(u_{\perp a,m,c})$ in terms of $\{R_{a\sim u\mid c}$, $R_{m\sim u\mid a,c}$, $\mathrm{cov}(u_{\perp c})\}$ below: 
\begin{align}
\mathrm{cov}(u_{\perp a,c})&=\mathrm{cov}(u_{\perp c})-\mathrm{cov}(u_{\perp c})^{1/2}R_{a\sim u\mid c}^{\T} R_{a\sim u\mid c}\mathrm{cov}(u_{\perp c})^{1/2},\label{eq::cov-uac}\\
\mathrm{cov}(u_{\perp a,m,c})&=\mathrm{cov}(u_{\perp a,c})-\mathrm{cov}(u_{\perp a,c})^{1/2}R_{m\sim u\mid a,c}^{\T} R_{m\sim u\mid a,c}\mathrm{cov}(u_{\perp a,c})^{1/2}. \label{eq::cov-uamc}
\end{align}

Combining \eqref{eq::direct-vu}--\eqref{eq::cov-uamc}, we express the true direct effect $\theta_1$ in terms of sensitivity parameters $R_{y\sim u\mid a,m,c}$, $R_{m\sim u\mid a,c}$, $R_{a\sim u\mid c}$, $\mathrm{cov}(u_{\perp c})$ and moments of observables. We renew the notation $R_u = \{R_{y\sim u\mid a,m,c}$, $R_{m\sim u\mid a,c}$, $R_{a\sim u\mid c}\}$ with vector $u$. Again, these $R$ parameters in $R_u$ correspond to the natural factorization of the joint distribution based on Figure \ref{fg::DAG-mediation-confounding}. We present the sharpness result in Corollary \ref{cor::vu-direct} below.  
\begin{corollary}
\label{cor::vu-direct}
Assume the covariance matrix of $(y,m,a,u)_{\perp c}$ is invertible for possibly vector $c$ but scalar $y,m,a,u$. Given the observables $\{y,m,a,c\}$, the parameters $\{R_{y\sim u\mid a,m,c},R_{m\sim u\mid a,c},R_{a\sim u\mid c}$, $\mathrm{cov}(u_{\perp c})\}$ can take arbitrary values in $(-1,1)\times \mathbb{B}_{\mathrm{dim}(m)}\times (-1,1)\times \mathbb{S}_{\mathrm{dim}(u)}^{++}$. 
\end{corollary}

\subsubsection{Indirect effect}
We next discuss the indirect effect based on the product method. By Theorem \ref{thm::most-general-ovb}, we can express $\beta_1$ and $\theta_3$ in the true indirect effect $\theta_3^{\T}\beta_1$ as
\begin{align}
\beta_1&=\tilde{\beta}_1-\mathrm{cov}(m_{\perp a,c})^{1/2}R_{m\sim u\mid a,c}\mathrm{cov}(u_{\perp a,c})^{-1/2}\mathrm{cov}(u_{\perp c})^{1/2}R_{a\sim u\mid c}^{\T} \mathrm{var}(a_{\perp c})^{-1/2},\label{eq::vu-multi-indirect-product1}\\
\theta_3&=\tilde{\theta}_3-\mathrm{var}(y_{\perp a,m,c})^{1/2}R_{y\sim u\mid a,m,c}\mathrm{cov}(u_{\perp a,m,c})^{-1/2}\mathrm{cov}(u_{\perp a,c})^{1/2}R_{m\sim u\mid a,c}^{\T} \mathrm{cov}(m_{\perp a,c})^{-1/2}.\label{eq::vu-multi-indirect-product2}
\end{align}
The above formulas \eqref{eq::vu-multi-indirect-product1} and \eqref{eq::vu-multi-indirect-product2} involve parameters $\{R_u,\mathrm{cov}(u_{\perp c}),\mathrm{cov}(u_{\perp a,c}),\mathrm{cov}(u_{\perp a,m,c})\}$. Moreover, the additional terms $\mathrm{cov}(u_{\perp a,c})$ and $\mathrm{cov}(u_{\perp a,m,c})$ in \eqref{eq::cov-uac}--\eqref{eq::cov-uamc} depend only on $\{R_{a\sim u\mid c}$, $R_{m\sim u\mid a,c}$, $\mathrm{cov}(u_{\perp c})\}$. Combining \eqref{eq::cov-uac}--\eqref{eq::cov-uamc} and \eqref{eq::vu-multi-indirect-product1}--\eqref{eq::vu-multi-indirect-product2}, we express the true indirect effect $\theta_3^{\T}\beta_1$ in terms of sensitivity parameters $\{R_u,\mathrm{cov}(u_{\perp c})\}$ and moments of observables. The sharpness for indirect effect is also guaranteed by Corollary \ref{cor::vu-direct} above. 

\subsection{Report the sensitivity analysis results with vector confounder}\label{subsec::vector-u-report}

This section extends Section \ref{subsec::scale-u-report} to the case with vector $u$. Similar to the discussion in Section \ref{subsec::ovb-vector-u}, it could be difficult to specify the sensitivity parameters $\{R_u,\mathrm{cov}(u_{\perp c})\}$, especially when the dimension of $u$ is undetermined. Therefore, reporting a series of estimates is not a practical method with vector $u$. Fortunately, we can derive similar results as in Section \ref{subsec::scale-u-report} to report the worst point estimate and $t$ statistic for the direct or indirect effect, and to report their robustness values.

\subsubsection{Report the worst case under upper bound of the $R$ parameters}\label{subsubsec::bound-r-vu}

In parallel with Section \ref{subsubsec::bound-r}, we now discuss the worst point estimate or $t$ statistic given the upper bound of $R$ parameters. Assume the observed point estimate is positive throughout this section. With specified $0\leq \rho_y,\rho_m,\rho_a<1$, we compute the minimum point estimate and minimum $t$ statistic under the constraint 
\begin{equation}\label{eq::rho-yma-2}
R_{y\sim u\mid a,m,c}^2\leq \rho_y,\quad \|R_{m\sim u\mid a,c}\|_2^2\leq \rho_m,\quad R_{a\sim u\mid c}^2\leq \rho_a.
\end{equation}
Similar to the discussion in Section \ref{subsec::ovb-vector-u}, the worst point estimate or $t$ statistic under the upper bound of $R$ parameters does not depend on the choice of $\mathrm{cov}(u_{\perp c})$; see Proposition \ref{prop::assume-cov-u-direct} and Proposition \ref{prop::assume-cov-u-indirect} in the Supplementary Material for the precise statement. For mathematical convenience, we assume $\mathrm{cov}(u_{\perp c})=(I-R_{a\sim u\mid c}^{\T} R_{a\sim u\mid c})^{-1}$, which implies $\mathrm{cov}(u_{\perp a,c})=I$ and $\mathrm{cov}(u_{\perp a,c,m})=I-R_{m\sim u\mid a,c}^{\T} R_{m\sim u\mid a,c}$. 

Based on the natural parameterization in Section \ref{subsec::vector-u-parameter}, we can re-express the direct and indirect effects as 
\begin{equation}\label{eq::direct-indirect-vector-u}
\begin{split}
\theta_1&=\tilde{\theta}_1+T_1\phi_1^*+T_2^{\T}\phi_2^*,\\
\theta_3^{\T}\beta_1&=(\tilde{\theta}_3+T_3\phi_2^*)^{\T}(\tilde{\beta}_1+T_4\phi_3^*),
\end{split}
\end{equation}
where $T_1,T_2,T_3,T_4$ in \eqref{eq::t1234} depend only on moments of observables and 
\begin{equation*}
\begin{split}
\phi_1^*&=\frac{R_{a\sim u\mid c}(I-R_{m\sim u\mid a,c}^{\T} R_{m\sim u\mid a,c})^{-1/2}R_{y\sim u\mid a,m,c}^{\T}}{(1-\|R_{a\sim u\mid c}\|_2^2)^{1/2}},\\
\phi_2^*&=R_{m\sim u\mid a,c}(I-R_{m\sim u\mid a,c}^{\T} R_{m\sim u\mid a,c})^{-1/2}R_{y\sim u\mid a,m,c}^{\T},\\
\phi_3^*&=\frac{R_{m\sim u\mid a,c}R_{a\sim u\mid c}^{\T}}{(1-\|R_{a\sim u\mid c}\|_2^2)^{1/2}}
\end{split}
\end{equation*}
depend only on sensitivity parameters $R_u$. Proposition \ref{prop::vector-u-range} below gives the range of $(\phi_1^*,\phi_2^*)$ for the direct effect, and the range of $(\phi_2^*,\phi_3^*)$ for the indirect effect. 

\begin{proposition}\label{prop::vector-u-range}
Under the constraint in \eqref{eq::rho-yma-2}, the range of $(\phi_1^*,\phi_2^*)$ is 
$$\left\{(x_1,x_2)\in \mathbb{R}\times \mathbb{R}^{\mathrm{dim}(m)}:x_1^2\leq \frac{\rho_a(\rho_y+\|x_2\|_2^2)}{1-\rho_a},\quad \|x_2\|_2^2\leq \frac{\rho_m\rho_y}{1-\rho_m}\right\},$$
and the range of $(\phi_2^*,\phi_3^*)$ is
$$\left\{(x_2,x_3)\in \mathbb{R}^{\mathrm{dim}(m)}\times \mathbb{R}^{\mathrm{dim}(m)}:\|x_2\|_2^2\leq \frac{\rho_m\rho_y}{1-\rho_m},\quad \|x_3\|_2^2\leq \frac{\rho_a\rho_m}{1-\rho_a}\right\}.$$
\end{proposition}

Proposition \ref{prop::vector-u-range} is the theoretical basis for computing the worst-case point estimates and $t$-statistics. Similar to the discussion in Section \ref{subsubsec::bound-r}, we can obtain the minimum point estimates or $t$ statistics for direct and indirect effects in \eqref{eq::direct-indirect-vector-u} by solving optimization problems over $(\phi_1^*,\phi_2^*)$ and $(\phi_2^*,\phi_3^*)$, subject to the constraints in Proposition \ref{prop::vector-u-range}. For the direct effect, the optimization problem involves variables of $(\mathrm{dim}(m)+1)$ degrees of freedom; and for the indirect effect, the optimization problem involves variables of $2\times \mathrm{dim}(m)$ degrees of freedom.

\subsubsection{Report the robustness value for mediation analysis}\label{subsubsec::rv-vu}
Based on the sensitivity parameters, we extend \eqref{eq::upper-bound-rho} to 
\begin{equation}\label{eq::upper-bound-rho-r-vu}
\max\{\|R_{y\sim u\mid a,m,c}\|_2^2,\|R_{m\sim u\mid a,c}\|_2^2,\|R_{a\sim u\mid c}\|_2^2\}\leq \rho, 
\end{equation}
with vector $u$. We define the robustness value for mediation as the minimum $\rho$ in \eqref{eq::upper-bound-rho-r-vu} such that the point estimate can be reduced to $0$ or the confidence interval can be altered to cover $0$. Similar to the discussion in Section \ref{subsubsec::rv-scalar-u}, we obtain the minimum $t$ statistic under the constraint \eqref{eq::upper-bound-rho-r-vu} in Section \ref{subsubsec::bound-r-vu}. Then, we conduct the grid search to gradually increase $\rho$ and find the critical values of $\rho$ for the confidence interval and point estimate such that the $t$ statistic hits $1.96$ and $0$, respectively.

\subsection{Scalar or vector confounder: when does it matter?}\label{subsec::scalar-u-vs-vector-u}

The following Proposition \ref{prop::range-connection} is a direct corollary of Proposition \ref{prop::scalar-u-range} and Proposition \ref{prop::vector-u-range}. 

\begin{proposition}\label{prop::range-connection}
(i) For the direct effect, the range of $(\phi_1,\phi_2)$ under the constraint \eqref{eq::rho-yma} is exactly the same as the range of $(\phi_1^*,\phi_2^*)$ under the constraint \eqref{eq::rho-yma-2}. (ii) For the indirect effect, when $\mathrm{dim}(m)=1$ or the exposure $a$ is randomized conditional on the observed covariates $c$, the range of $(\phi_2,\phi_3)$ under the constraint \eqref{eq::rho-yma} is exactly the same as the range of $(\phi_2^*,\phi_3^*)$ under the constraint \eqref{eq::rho-yma-2}. 
\end{proposition}

Proposition \ref{prop::range-connection}(i) states that assuming scalar or vector $u$ does not change the sensitivity analysis results for the direct effect. Proposition \ref{prop::range-connection}(ii) states that assuming scalar or vector $u$ does not change the sensitivity analysis results for indirect effect if the mediator $m$ is a scalar or the exposure $a$ is randomized conditional on observed covariates $c$. However, with nonrandomized exposure $a$ and vector mediator $m$, the sensitivity analysis results for indirect effect depends on whether $u$ is a scalar or vector, because the range of $(\phi_2,\phi_3)$ under the constraint \eqref{eq::rho-yma} is not the same as the range of $(\phi_2^*,\phi_3^*)$ under the constraint \eqref{eq::rho-yma-2}. In Section \ref{sec::computing-rv} of the Supplementary Material, we present a simulation study on how much we miss if we assume $u$ is a scalar to compute robustness values for the indirect effect with a vector mediator.  

In the following two studies of Section \ref{sec::illustrations}, we have that either $\mathrm{dim}(m)=1$ or the exposure is randomized, so it is not restrictive to assume $u$ is a scalar for the purpose of reporting sensitivity bounds under the upper bound of $R$ parameters or reporting robustness values. 

\section{Illustrations}\label{sec::illustrations}

\subsection{Mediation analysis with a single mediator}\label{subsec::illustrations-single}

\citet{jose2013doing} conducts a mediation analysis to assess the extent to which positive life events affect happiness through the level of gratitude. The data came from the first of five times of positive psychology measurement of 364 respondents to the International Wellbeing Study separated by three months each. The positive life events $a$, gratitude $m$, and happiness $y$ are measured by the sum of scores from some designed questions \citep{mccullough2002grateful, lyubomirsky1999measure}. In this study, $y,a,m$ are continuous variables, and no covariates are collected. By the Baron--Kenny approach, the estimate for the direct effect is $0.269$, with the $95\%$ confidence interval $[0.132,0.406]$, and the estimate for the indirect effect is $0.215$, with the $95\%$ confidence interval $[0.131,0.299]$. So, both are statistically significant. Because this is a non-randomized study, we reanalyze the data and assess the sensitivity of the conclusion with respect to unmeasured confounding.

The sensitivity analysis results are based on possibly multiple unmeasured confounders. By Proposition \ref{prop::range-connection}, it is not restrictive to assume $u$ is a scalar to compute the robustness values, since there is a single mediator. Figure \ref{fig::ex1} shows the minimum $t$ statistics under the condition $\max \{R_{y\sim u\mid a,m,c}^2,R_{m\sim u\mid a,c}^2,R_{a\sim u\mid c}^2\}\leq \rho$. Table \ref{tab::ex1} reports some statistics of the direct and indirect effects and their robustness values for the point estimates and $95\%$ confidence intervals. To interpret them, the robustness value for the estimate refers to the minimal $\rho$ such that the point estimate, or, equivalently, the $t$ statistic could be reduced to 0, under \eqref{eq::upper-bound-rho}; the robustness value for the confidence interval refers to the minimal $\rho$ such that the $t$ statistic could be reduced to $1.96$, under \eqref{eq::upper-bound-rho}. In this study, we need larger confounding to alter the conclusion about the indirect effect than about the direct effect.

\begin{figure}[h]
\begin{center}
\includegraphics[width=4in]{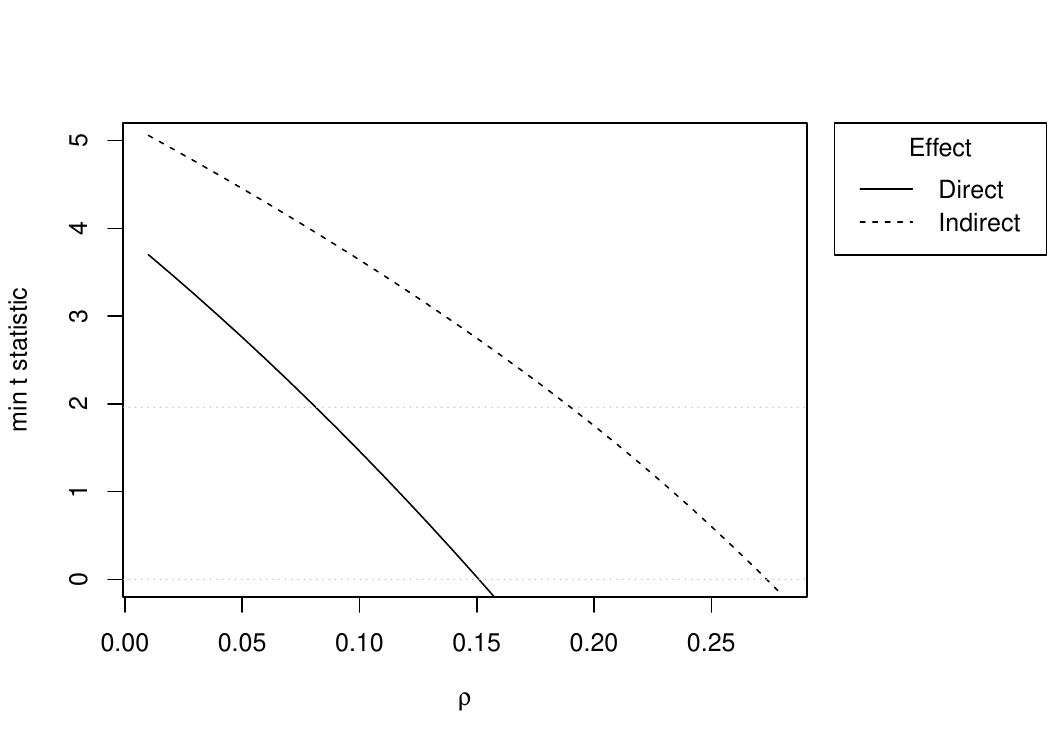}
\caption{The minimum $t$ statistics of the direct and indirect effects, under the restriction $\max \{R_{y\sim u\mid a,m,c}^2,R_{m\sim u\mid a,c}^2,R_{a\sim u\mid c}^2\}\leq \rho$. The grey reference lines indicate the $t=1.96$ and $t=0$ thresholds.}
\label{fig::ex1}
\end{center}
\end{figure}

\begin{table}[h]
\caption{Sensitivity analysis for \citet{jose2013doing}'s mediation analysis. Some reference sample $R^2$ values are $\hat{R}_{y\sim a\mid m}^2 =0.05$, $\hat{R}_{y\sim m\mid a}^2 =0.25$ and $\hat{R}_{m\sim a}^2=0.09$. The R.V. is short for the ``robustness value''. }\label{tab::ex1}
\begin{center}
\begin{tabular}{lccccc}
\hline 
& Est. & Std. Err. & $t$-value & R.V. for Est. & R.V. for $95\%$ C.I. \\
\hline 
Direct effect & $0.269$ & $0.069$ & $3.923$ & $0.15$ & $0.08$ \\
Indirect effect & $0.215$ & $0.041$ & $5.205$ & $0.27$ & $0.19$ \\
\hline 
\end{tabular}
\end{center}
\end{table}

\subsection{Mediation analysis with multiple mediators}\label{subsec::illustrations-multiple}

\citet{chong2016iron} conducts a randomized study on 219 students of a rural secondary school in the Cajamarca district of Peru during the 2009 school year. They are interested in whether or not iron deficiency could contribute to the intergenerational persistence of poverty by affecting the school performance and aspirations of anemic students. They encourage students to visit the clinic daily to take an iron pill. They randomly assign students to be exposed multiple times to one of the following three videos: in the first ``soccer'' group, a popular soccer player is encouraging iron supplements to maximize energy; in the second ``physician'' group, a doctor is encouraging iron supplements for overall health; in the third ``placebo'' group, the video does not mention iron at all. They measure various variables of the students after randomization, such as the anemia status at the follow-up survey, the cognitive function measured by the cognitive Wii game, and the school performance measured by average grade in the third and fourth quarters of 2009. 

We reanalyze their data. We are interested in the effect of the video assignment $a$ on the school performance $y$, with the anemia status $m_1$ and cognitive function $m_2$ as mediators. We focus on two levels of the exposure $a$: the ``physician'' group and the ``placebo'' group. The covariates $c$ include the gender, class level, anemia status at baseline survey, household monthly income, electricity in home, mother's years of education. By the Baron--Kenny approach, the estimate for the direct effect is $0.221$, with the $95\%$ confidence interval $[-0.167,0.609]$, and the estimate for the indirect effect is $0.138$, with the $95\%$ confidence interval $[-0.017,0.293]$. Neither is significant.

For illustration purposes, we only assess the sensitivity of the point estimates with respect to unmeasured confounding between the mediator and outcome. The sensitivity analysis results are based on possibly multiple unmeasured confounders. By Proposition \ref{prop::range-connection}, it is not restrictive to assume $u$ is a scalar to compute the robustness values, since $R_{a\sim u\mid c}=0$ by the design. Figure \ref{fig::ex1} shows the minimum $t$ statistics under the condition $\max \{R_{y\sim u\mid a,m,c}^2,R_{u\sim m\mid a,c}^2\}\leq \rho$. Table \ref{tab::ex2} reports some statistics of the direct and indirect effects and their robustness values for the point estimates. We present the sensitivity analysis result based on the formal benchmarking in Section \ref{sec::formal-benchmarking} of the Supplementary Material. 

In addition to Table \ref{tab::ex2}, we report three additional robustness values, pretending that \citet{chong2016iron}'s study is not a randomized experiment. For the direct effect, the robustness value for the point estimate is 0.07, which is substantially smaller than the previously reported value of 0.38 with $R_{a\sim u\mid c}=0$. For the indirect effect, the robustness value for the point estimate is 0.16 when $u$ is a scalar and 0.13 when $u$ is a vector. This also shows a small difference between robustness values between the scalar and vector $u$ cases.

\begin{figure}[h]
\begin{center}
\includegraphics[width=4in]{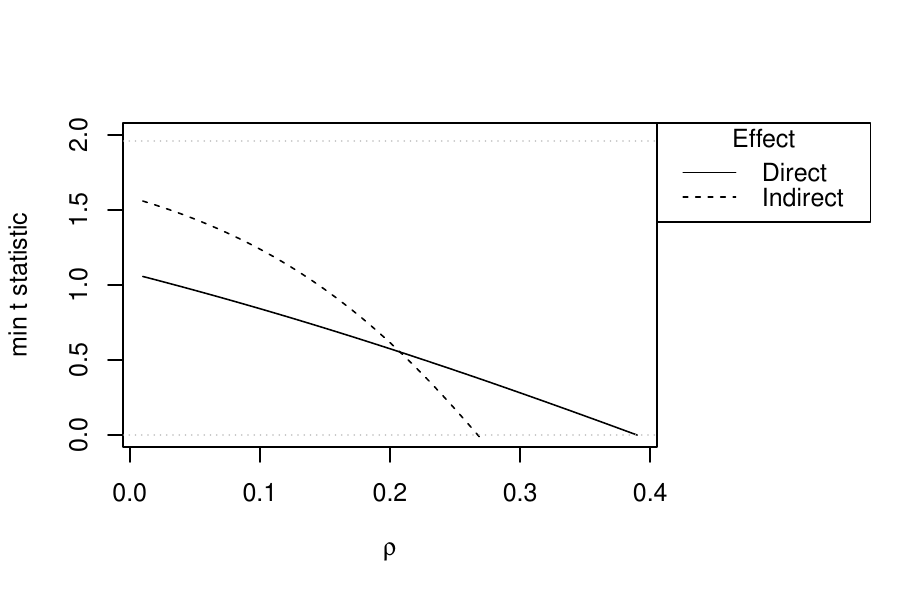}
\caption{The minimum $t$ statistics of the direct effect and indirect effect, under the restriction $\max \{R_{y\sim u\mid a,m,c}^2,\|R_{m\sim u\mid a,c}\|_2^2\}\leq \rho$. The grey reference lines indicate the $t=1.96$ and $t=0$ thresholds. }
\label{fig::ex2}
\end{center}
\end{figure}

\begin{table}[h]
\caption{Sensitivity analysis for mediation analysis based on \citet{chong2016iron}'s study. Some reference sample $R^2$ values are $\max_{j}\hat{R}_{y\sim c_j\mid a,c_{-j},m}^2=0.113$ and $\max_{j}\hat{R}_{c_j\sim m\mid a,c_{-j}}^2=0.135$. The R.V. is short for the ``robustness value''. }\label{tab::ex2}
\begin{center}
\begin{tabular}{lccccc} 
\hline 
& Est. & Std. Err. & $t$-value & R.V. for Est. & R.V. for $95\%$ C.I. \\ 
\hline 
Direct effect & 0.221 & 0.205 & 1.079 & 0.38 & 0.00 \\
Indirect effect & 0.138 & 0.087 & 1.586 & 0.26 & 0.00 \\
\hline 
\end{tabular}
\end{center}
\end{table}

\section{Sensitivity analysis based on marginal correlation coefficeints}\label{sec::marginal-cc}

In this section, we discuss an alternative sensitivity analysis strategy based on marginal correlation coefficients, which are closely related to our $R$ parameters. Let $\rho_{yx}\in\mathbb{R}^{\mathrm{dim}(y)\times \mathrm{dim}(x)}$ denote the marginal correlation matrix of $y$ and $x$, with the $(i,j)$th element of $\rho_{yx}$ equal to the correlation of $y_i$ and $x_j$. By Definition \ref{def::new-r}, we have
\begin{equation}
\label{eq::chol-cov}
R_{y\sim x}=\mathrm{cov}(y)^{-1/2}\mathrm{diag}\{\mathrm{cov}(y)\}^{1/2}\rho_{yx}\mathrm{diag}\{\mathrm{cov}(x)\}^{1/2}\mathrm{cov}(x)^{-1/2}=\Phi_y\rho_{yx}\Phi_x^{\T},
\end{equation}
where $\Phi_y=\mathrm{cov}(y)^{-1/2}\mathrm{diag}\{\mathrm{cov}(y)\}^{1/2}$ and $\Phi_x=\mathrm{cov}(x)^{-1/2}\mathrm{diag}\{\mathrm{cov}(x)\}^{1/2}$. The $\Phi_y$ matrix can be interpreted as the linear transformation matrix between the whitened variable $\mathrm{cov}(y)^{-1/2}y$ that has covariance matrix $I$ and the scaled variable $\mathrm{diag}\{\mathrm{cov}(y)\}^{-1/2}y$ that has unit marginal variances. A similar interpretation applies to $\Phi_x$. The $\Phi$ matrix reduces to 1 and the $R$ parameter reduces to the $\rho$ parameter under the scalar case. 

Equation \eqref{eq::chol-cov} suggests the possibility of using marginal correlations as sensitivity parameters. If some researchers find the Pearson correlation coefficients more interpretable than the $R$ parameters, then they can consider using them as the sensitivity parameters coupled with \eqref{eq::chol-cov}. Specifically, we denote these marginal correlation coefficients as $\rho_{a \sim u \mid c}$, $\rho_{m \sim u \mid a, c}$, and $\rho_{y \sim u \mid a, m, c}$ for the exposure-confounder, mediator-confounder, and outcome-confounder associations, respectively. In this section, we explore the sensitivity analysis using the $\rho$ parameters for the cases when $u$ is a scalar and when $u$ is a vector, respectively. With scalar $m$ and $u$, it is identical to our sensitivity analysis because the $R$ parameter equals the $\rho$ parameter for scalars. However, we will point out some difficulties in implementing the sensitivity analysis using the $\rho$ parameters with vector $m$ and $u$. In particular, we may lose the sharpness of the sensitivity analysis, and to ensure sharpness, we must rely on the results of the $R$ parameters.


\subsection{When $u$ is a scalar}\label{subsec::marginal-cc-scalar}

We start with the case when $u$ is a scalar. By \eqref{eq::chol-cov}, we can express $R_u$ in terms of new parameters $\{\rho_{a\sim u\mid c}, \rho_{m\sim u\mid a,c}, \rho_{y\sim u\mid a,m,c}\}$ as below: 
\begin{equation}\label{eq::r-rho}
\begin{split}
R_{a\sim u\mid c}&=\rho_{a\sim u\mid c},\\
R_{m\sim u\mid a,c}&=\mathrm{cov}(m_{\perp a,c})^{-1/2}\mathrm{diag}\{\mathrm{cov}(m_{\perp a,c})\}^{1/2}\rho_{m\sim u\mid a,c},\\
R_{y\sim u\mid a,m,c}&=\rho_{y\sim u\mid a,m,c}.
\end{split}
\end{equation}
Recall that \eqref{eq::direct-indirect-scalar-u} expresses the direct and indirect effects in terms of $T$'s from \eqref{eq::t1234}, which depend only on the sensitivity parameters $R_u$ and the $\phi$'s in \eqref{eq::phi123}. Combining \eqref{eq::direct-indirect-scalar-u}--\eqref{eq::phi123} with \eqref{eq::r-rho}, we can express the direct and indirect effects as functions of the sensitivity parameters $\{\rho_{a \sim u \mid c}, \rho_{m \sim u \mid a, c}, \rho_{y \sim u \mid a, m, c}\}$, the $T$'s in \eqref{eq::t1234}, and the covariance matrix $\mathrm{cov}(m_{\perp a, c})$.

With this parameterization of the direct and indirect effects, we are now ready to discuss how to report the results of the sensitivity analysis based on these new sensitivity parameters. In Section \ref{subsec::scale-u-report}, our approach uses the Euclidean norms of the $R$ parameters to measure their strengths, which are equivalent to the corresponding $R^2$ parameters and can be interpreted as the proportion of the variance explained by the unmeasured confounder $u$. However, the interpretability of $\|\rho_{m \sim u \mid a, c}\|_2^2$ becomes more complex when $m$ is a vector. Motivated by \eqref{eq::rho-yma} and \citet{freidling2022optimization}, we suggest reporting the worst point estimate or $t$ statistic under the constraint
\begin{equation}\label{eq::constraint-rhos}
|\rho_{y\sim u\mid a,m,c}|\leq \rho_y,\quad \|\rho_{m\sim u\mid a,c}\|_{\infty}\leq \rho_m,\quad |\rho_{a\sim u\mid c}|\leq \rho_a.
\end{equation}
The constraint $\|\rho_{m\sim u\mid a,c}\|_{\infty}\leq \rho_m$ is equivalent to $|\rho_{m_j\sim u\mid a,c}|\leq \rho_m$ for each $1\leq j\leq \mathrm{dim}(m)$. We focus on the constraint \ref{eq::constraint-rhos}, but similar discussions apply to other constraints for the sensitivity parameters $\{\rho_{a \sim u \mid c}, \rho_{m \sim u \mid a, c}, \rho_{y \sim u \mid a, m, c}\}$.

We can compute the worst point estimate or $t$ statistic by solving an optimization problem over parameters $\{\rho_{a\sim u\mid c}, \rho_{m\sim u\mid a,c}, \rho_{y\sim u\mid a,m,c}\}$ subject to the constraint \eqref{eq::constraint-rhos}. For the point estimates, with estimated $T_j$ ($j=1,\ldots, 4$) and estimated covariance matrix $\mathrm{cov}(m_{\perp a,c})$, the plug-in estimates for direct and indirect effects are functions of $\{\rho_{a\sim u\mid c}, \rho_{m\sim u\mid a,c}, \rho_{y\sim u\mid a,m,c}\}$. For the $t$ statistics, using the bootstrapped sample of $T_j$ ($j=1,\ldots, 4$) and $\mathrm{cov}(m_{\perp a,c})$, the nonparametric bootstrap standard errors for direct and indirect effects are also functions of $\{\rho_{a\sim u\mid c}, \rho_{m\sim u\mid a,c}, \rho_{y\sim u\mid a,m,c}\}$. Therefore, the point estimates and $t$ statistics for direct and indirect effects are functions of $\{\rho_{a\sim u\mid c}, \rho_{m\sim u\mid a,c}, \rho_{y\sim u\mid a,m,c}\}$, and we compute their minimum or maximum subject to the constraint \eqref{eq::constraint-rhos}. 

However, this approach might result in conservative sensitivity bounds because the sensitivity parameters $\{\rho_{a \sim u \mid c}, \rho_{m \sim u \mid a, c}, \rho_{y \sim u \mid a, m, c}\}$ may not attain all possible values within constraint \eqref{eq::constraint-rhos}. As discussed in Corollary \ref{cor::direct}, the parameters $\{R_{y \sim u \mid a, m, c}, R_{m \sim u \mid a, c}, R_{a \sim u \mid c}\}$ can take arbitrary values in $(-1, 1) \times \mathbb{B}_{\mathrm{dim}(m)} \times (-1, 1)$. This implies that $\rho_{a \sim u \mid c}$ and $\rho_{y \sim u \mid a, m, c}$ can take any value in $(-1, 1)$, while $\rho_{m \sim u \mid a, c}$ is constrained to
\begin{equation}\label{eq::rho-additional-constraint}
\left\{\mathrm{diag}\{\mathrm{cov}(m_{\perp a,c})\}^{-1/2}\mathrm{cov}(m_{\perp a,c})^{1/2}x:x\in\mathbb{R}^{\mathrm{dim}(m)},\ \|x\|_2^2\leq 1\right\}, 
\end{equation}
which depends on the distribution of the observables $\{a, m, c\}$. To ensure the sharpness of the sensitivity analysis based on the $\rho$ parameters, we must calculate the worst point estimate or $t$ statistic under the constraints in both \eqref{eq::constraint-rhos} \textit{and} \eqref{eq::rho-additional-constraint}. However, the constraint \eqref{eq::rho-additional-constraint}  involves unknown quantities that need to be estimated based on the observed data. We must modify the inferential procedures with an estimated feasible region, which is a nontrivial task compared with the inferential procedures based on the nonparametric bootstrap with the $R$ parameters. Moreover, deriving the feasible region of the $\rho$ parameters requires deriving the feasible region of the $R$ parameters. As such, we view the $R$ parameters as more natural choices for sensitivity analysis for mediation.

\citet{freidling2022optimization} consider general sensitivity analysis problems that incorporate multiple benchmarking constraints for least squares problems. We can also extend their results to mediation analysis under our optimization framework. Due to the space limit, we leave the numerical implementation for future research. 

\subsection{When $u$ is a vector}\label{subsec::marginal-cc-vector}
We next discuss the case when $u$ is a vector. By \eqref{eq::chol-cov}, we can express $R_u$ in terms of the new parameters $\{\rho_{a \sim u \mid c}, \rho_{m \sim u \mid a, c}, \rho_{y \sim u \mid a, m, c}\}$ as follows:
\begin{equation}\label{eq::r-rho-vu}
\begin{split}
R_{a\sim u\mid c}&=\rho_{a\sim u\mid c}\mathrm{diag}\{\mathrm{cov}(u_{\perp c})\}^{1/2}\mathrm{cov}(u_{\perp c})^{-1/2},\\
R_{m\sim u\mid a,c}&=\mathrm{cov}(m_{\perp a,c})^{-1/2}\mathrm{diag}\{\mathrm{cov}(m_{\perp a,c})\}^{1/2}\rho_{m\sim u\mid a,c}\mathrm{diag}\{\mathrm{cov}(u_{\perp a,c})\}^{1/2}\mathrm{cov}(u_{\perp a,c})^{-1/2},\\
R_{y\sim u\mid a,m,c}&=\rho_{y\sim u\mid a,m,c}\mathrm{diag}\{\mathrm{cov}(u_{\perp a,m,c})\}^{1/2}\mathrm{cov}(u_{\perp a,m,c})^{-1/2}.
\end{split}
\end{equation}
Recall that \eqref{eq::cov-uac}--\eqref{eq::cov-uamc} express $\mathrm{cov}(u_{\perp a,c})$ and $\mathrm{cov}(u_{\perp a,m,c})$ in terms of $\{R_{a\sim u\mid c}$, $R_{m\sim u\mid a,c}$, $\mathrm{cov}(u_{\perp c})\}$. Combining \eqref{eq::cov-uac}--\eqref{eq::cov-uamc} and \eqref{eq::r-rho-vu}, our original sensitivity parameters $\{R_u,\mathrm{cov}(u_{\perp c})\}$ are functions of new sensitivity parameters $\{\rho_{a\sim u\mid c}, \rho_{m\sim u\mid a,c}, \rho_{y\sim u\mid a,m,c},\mathrm{cov}(u_{\perp c})\}$. Thus, we can further write direct and indirect effects as functions of new sensitivity parameters $\{\rho_{a\sim u\mid c}, \rho_{m\sim u\mid a,c}, \rho_{y\sim u\mid a,m,c},\mathrm{cov}(u_{\perp c})\}$ and moment of observables. Similar to the scalar case in \eqref{eq::constraint-rhos}, we may wish to compute the worst point estimate or $t$ statistic under the constraint
\begin{equation}\label{eq::constraint-rhos2}
\|\rho_{y\sim u\mid a,m,c}\|_{\infty}\leq \rho_y,\quad \|\rho_{m\sim u\mid a,c}\|_{\infty}\leq \rho_m,\quad \|\rho_{a\sim u\mid c}\|_{\infty}\leq \rho_a. 
\end{equation}

However, unlike solving the optimization problems under the constraint \eqref{eq::rho-yma-2} with the $R$ parameterization, the corresponding optimization problems under the constraint \eqref{eq::constraint-rhos2} depend on the choice of $\mathrm{cov}(u_{\perp c})$. Moreover, we must specify the dimension of $u$; otherwise, the point estimate or $t$ statistic could become unbounded if the number of unmeasured confounders increases indefinitely with a constant correlation between $y$, $m$, and $a$. In practice, we can start by assuming $\mathrm{dim}(u) = 2$ or $\mathrm{dim}(u) = 3$. We then solve an appropriate optimization problem over the parameters $\{\rho_{a \sim u \mid c}, \rho_{m \sim u \mid a, c}, \rho_{y \sim u \mid a, m, c}, \mathrm{cov}(u_{\perp c})\}$ subject to the constraint \eqref{eq::constraint-rhos2}. When the dimension of $u$ is small, this optimization problem remains computationally feasible. However, similar to the scalar case, this approach may be prone to conservative sensitivity bounds due to the sharpness issue, as the range of $\rho_{m \sim u \mid a, c}$ depends on the observables $\{a, m, c\}$.

\section{Discussion}\label{sec::discussion}


Our sensitivity analysis focuses on computing sensitivity bounds for the direct or indirect effect separately. When $u$ is a scalar, Proposition \ref{prop::scalar-u-range} is useful to formulate it as an optimization problem over  $(\phi_1,\phi_2)$ or $(\phi_2,\phi_3)$; when $u$ is a vector, Proposition \ref{prop::vector-u-range} is useful to formulate it as an optimization problem over  $(\phi_1^*,\phi_2^*)$ or $(\phi_2^*,\phi_3^*)$. However, for joint sensitivity analysis involving direct and indirect effects simultaneously, Proposition \ref{prop::scalar-u-range} and Proposition \ref{prop::vector-u-range} are no longer enough. For scalar $u$, since the direct and indirect effects in \eqref{eq::direct-indirect-scalar-u} are expressed by $R$ parameters and moments of observables, the test statistic must be a function of $R$ parameters. Thus, we can formulate it as an optimization problem over $R$ parameters, which involve $(\mathrm{dim}(m)+2)$ degrees of freedom. For vector $u$, the problem is harder due to the undetermined dimension of $u$, and we present a numerical solution in Section \ref{sec::general-vector-u} of the Supplementary Material. 



Theorems \ref{thm::multi-ols}--\ref{thm::most-general-ovb} underpin our sensitivity analysis methods. The linear regression formulation is particularly useful for the Baron--Kenny approach. \citet{cinelli2019sensitivity} extend it to sensitivity analysis in linear structural causal models, and \citet{cinelli2020omitted} extend it to sensitivity analysis in the linear instrumental variable model. It is of interest to further explore the applications of Theorems \ref{thm::multi-ols}--\ref{thm::most-general-ovb} in these settings.  \citet{zheng2021copula} discuss sensitivity analysis with vector $y$ under additional factor-model-type assumptions. It is of interest to derive results that combine Theorem \ref{thm::multi-ols} and \citet{zheng2021copula}. \citet{chernozhukov2022long} provide an extension of the OVB formula under general nonlinear models. It is of interest to extend their formula to general mediation analysis.

 \citet{vanderweele2015explanation} surveys many advanced approaches for mediation analysis based on general definitions of direct and indirect effects. Despite the serious efforts made by \citet{vanderweele2010bias}, \citet{tchetgen2012semiparametric}, and \citet{ding2016sharp}, general sharp and interpretable sensitivity analysis methods are still lacking in the literature. This is an open research direction.

With multiple mediators, we follow the formulation of \citet{vanderweele2013mediation} and view the whole vector $m$ as the mediator. Sometimes, it is of scientific interest to model the causal mechanisms among the mediators and decompose the total effect of the exposure on the outcome into path-specific effects \citep{avin2005identifiability, miles2017quantifying}. Due to the possible unmeasured confounding, it is important to develop a sensitivity analysis method for path-specific effects. We leave it to future research.

We focus on a single mediator or a low-dimensional vector of mediators. Modern genetic studies often have high-dimensional mediators \citep{song2020bayesian, zhou2020estimation, yang2021estimation}. It is important to extend our sensitivity analysis method to handle this more challenging setting. Moreover, \citet{song2020bayesian} and \citet{yang2021estimation} discuss alternative measures of mediation. It is curious to extend our results to those measures.

\section*{Acknowledgement}

Ding was partially funded by the U.S. National Science Foundation (grant \# 1945136). The Associate Editor and two reviewers made helpful comments on our paper.

\bibliographystyle{Chicago}

\bibliography{bku}

\newpage

\appendix

\renewcommand{\thetheorem}{S\arabic{theorem}}
\setcounter{theorem}{0}
\renewcommand{\thelemma}{S\arabic{lemma}}
\setcounter{lemma}{0}
\renewcommand{\theproposition}{S\arabic{proposition}}
\setcounter{proposition}{0}
\renewcommand{\thecorollary}{S\arabic{corollary}}
\setcounter{corollary}{0}
\renewcommand{\thedefinition}{S\arabic{definition}}
\setcounter{definition}{0}
\renewcommand{\thealgo}{S\arabic{algo}}
\setcounter{algo}{0}
\renewcommand{\thepage}{S\arabic{page}}
\setcounter{page}{1}
\renewcommand{\theequation}{S\arabic{equation}}
\setcounter{equation}{0}
\renewcommand{\thetable}{S\arabic{table}}
\setcounter{table}{0}
\renewcommand{\thefigure}{S\arabic{figure}}
\setcounter{figure}{0}

\begin{center}
\huge \bfseries  Supplementary material
\end{center}

\bigskip 

Section \ref{sec::additional} presents additional results, including more rigorous statements about multiple unmeasured confounders in Section \ref{sec::vector-u-more-result}, details about formal benchmarking in Section \ref{sec::formal-benchmarking}, a simulation study on how much we miss if we assume $u$ is a scalar to compute robustness values in Section \ref{sec::computing-rv}, and discussion of general testing problems with vector $u$ in Section \ref{sec::general-vector-u}. 

Section \ref{sec::package} presents guidance to our R package \texttt{BaronKennyU}. 

Section \ref{sec::proofs} presents the proofs. 

\section{Additional results}\label{sec::additional}

\subsection{Rigorous statements with multiple unmeasured confounders}\label{sec::vector-u-more-result}

\subsubsection{Ordinary least squares}\label{subsec::vector-u-ols-discussion}
Proposition \ref{prop::assume-cov-u-ols} below shows that we can compute the worst point estimate or $t$ statistic for some functions of $\gamma_1$ in \eqref{eq::most-general-long-short-ols} by solving an easier optimization problem without the nuisance parameter $\mathrm{cov}(u_{\perp c})$. 

\begin{proposition}\label{prop::assume-cov-u-ols}
Assume the covariance matrix of $(y,a,u)_{\perp c}$ is invertible. Consider the optimization problem  
\begin{equation}\label{eq::optimization-ols-1}
\begin{aligned}
&\min\quad && f(R_{y\sim u\mid a,c}\mathrm{cov}(u_{\perp a,c})^{-1/2}\mathrm{cov}(u_{\perp c})^{1/2}R_{a\sim u\mid c}^{\T})\\
&\text{subject to }\quad &&\|R_{a\sim u\mid c}\|_2^2 \leq \rho_a, \quad \|R_{y\sim u\mid a,c}\|_2^2 \leq \rho_y
\end{aligned}
\end{equation}  
for any function $f$. The optimization problem does not depend on the choice of $\mathrm{cov}(u_{\perp c})$. Specifically, when $\Sigma=(I-R_{a\sim u\mid c}^{\T}R_{a\sim u\mid c})^{-1}$, we have $\mathrm{cov}(u_{\perp a,c})=I$ and the optimization problem reduces to  
\begin{equation}\label{eq::optimization-ols-2}
\begin{aligned}
&\min\quad && f(R_{y\sim u\mid a,c}(I-R_{a\sim u\mid c}^{\T} R_{a\sim u\mid c})^{-1/2}R_{a\sim u\mid c}^{\T})\\
&\text{subject to }\quad &&\|R_{a\sim u\mid c}\|_2^2 \leq \rho_a, \quad \|R_{y\sim u\mid a,c}\|_2^2 \leq \rho_y.
\end{aligned}
\end{equation} 
\end{proposition}

We can choose an appropriate function $f$ in Proposition \ref{prop::assume-cov-u-ols} to compute the worst point estimate or $t$ statistic for some functions of $\gamma_1$, under the upper bound of $R$ parameters. For the point estimate, the plug-in estimator
\begin{equation}\label{eq::gamma-plug-in}
\hat{\gamma}_1=\hat{\gamma}_1^{\text{obs}}-\hat{\mathrm{cov}}(Y_{\perp A,C})^{1/2}R_{y\sim u\mid a,c}\mathrm{cov}(u_{\perp a,c})^{-1/2}\mathrm{cov}(u_{\perp c})^{1/2}R_{a\sim u\mid c}^{\T}\hat{\mathrm{cov}}(A_{\perp C})^{-1/2}
\end{equation}
 is a function of $R_{y\sim u\mid a,c}\mathrm{cov}(u_{\perp a,c})^{-1/2}\mathrm{cov}(u_{\perp c})^{1/2}R_{a\sim u\mid c}^{\T}$, and so is the plug-in estimator of certain functions of $\gamma_1$. For the $t$ statistic, the nonparameteric bootstrap standard error of $\hat{\gamma}$ in \eqref{eq::gamma-plug-in} is a function of $R_{y\sim u\mid a,c}\mathrm{cov}(u_{\perp a,c})^{-1/2}\mathrm{cov}(u_{\perp c})^{1/2}R_{a\sim u\mid c}^{\T}$, and so is the $t$ statistic for certain functions of $\gamma_1$. 
 
 To solve the optimization problem \eqref{eq::optimization-ols-2}, we may want to derive the range of $R_{y\sim u\mid a,c}(I-R_{a\sim u\mid c}^{\T} R_{a\sim u\mid c})^{-1/2}R_{a\sim u\mid c}^{\T}$ under the constraint $\|R_{a\sim u\mid c}\|_2^2 \leq \rho_a$ and $\|R_{y\sim u\mid a,c}\|_2^2 \leq \rho_y$, which could further simplify the optimization problem \eqref{eq::optimization-ols-2} by reparameterization.
 
 \subsubsection{Baron--Kenny approach}
 
Proposition \ref{prop::assume-cov-u-direct} and Proposition \ref{prop::assume-cov-u-indirect} below show that we can compute the worst point estimate or $t$ statistic for the direct and indirect effects, respectively, by solving easier optimization problems without the nuisance parameter $\mathrm{cov}(u_{\perp c})$. 
 
 \begin{proposition}\label{prop::assume-cov-u-direct}
Assume the covariance matrix of $(y,m,a,u)_{\perp c}$ is invertible. Consider the optimization problem  
\begin{equation*}
\begin{aligned}
&\min\quad && f(R_{y\sim u\mid a,m,c}\mathrm{cov}(u_{\perp a,m,c})^{-1/2}\mathrm{cov}(u_{\perp c})^{1/2}R_{a\sim u\mid c}^{\T}, R_{m\sim u\mid a,c}\mathrm{cov}(u_{\perp a,c})^{1/2}\mathrm{cov}(u_{\perp a,m,c})^{-1/2}R_{y\sim u\mid a,m,c}^{\T})\\
&\text{subject to }\quad &&\|R_{a\sim u\mid c}\|_2^2 \leq \rho_a, \quad \|R_{m\sim u\mid a,c}\|_2^2 \leq \rho_m,\quad \|R_{y\sim u\mid a,m,c}\|_2^2\leq \rho_y
\end{aligned}
\end{equation*}  
for any function $f$. The optimization problem does not depend on the choice of \( \Sigma \). Specifically, when $\Sigma=(I-R_{a\sim u\mid c}^{\T}R_{a\sim u\mid c})^{-1}$, we have $\mathrm{cov}(u_{\perp a,c})=I$ and $\mathrm{cov}(u_{\perp a,m,c})=(I-R_{m\sim u\mid a,c}^{\T}R_{m\sim u\mid a,c})^{1/2}$. The optimization problem reduces to  
\begin{equation}\label{eq::reduced-vu-direct}
\begin{aligned}
&\min\quad && f\left(\frac{R_{y\sim u\mid a,m,c}(I-R_{m\sim u\mid a,c}^{\T}R_{m\sim u\mid a,c})^{-1/2}R_{a\sim u\mid c}^{\T}}{(1-\|R_{a\sim u\mid c}\|_2^2)^{1/2}},R_{m\sim u\mid a,c}(I-R_{m\sim u\mid a,c}^{\T} R_{m\sim u\mid a,c})^{-1/2}R_{y\sim u\mid a,m,c}^{\T}\right)\\
&\text{subject to }\quad &&\|R_{a\sim u\mid c}\|_2^2 \leq \rho_a, \quad \|R_{m\sim u\mid a,c}\|_2^2 \leq \rho_m,\quad \|R_{y\sim u\mid a,m,c}\|_2^2\leq \rho_y.
\end{aligned}
\end{equation} 
\end{proposition}

 \begin{proposition}\label{prop::assume-cov-u-indirect}
Assume the covariance matrix of $(y,m,a,u)_{\perp c}$ is invertible. Consider the optimization problem  
\begin{equation*}
\begin{aligned}
&\min\quad && f(R_{m\sim u\mid a,c}\mathrm{cov}(u_{\perp a,c})^{1/2}\mathrm{cov}(u_{\perp a,m,c})^{-1/2}R_{y\sim u\mid a,m,c}^{\T},R_{m\sim u\mid a,c}\mathrm{cov}(u_{\perp a,c})^{-1/2}\mathrm{cov}(u_{\perp c})^{1/2}R_{a\sim u\mid c}^{\T})\\
&\text{subject to }\quad &&\|R_{a\sim u\mid c}\|_2^2 \leq \rho_a, \quad \|R_{m\sim u\mid a,c}\|_2^2 \leq \rho_m,\quad \|R_{y\sim u\mid a,m,c}\|_2^2\leq \rho_y
\end{aligned}
\end{equation*}  
for any function $f$. The optimization problem does not depend on the choice of \( \Sigma \). Specifically, when $\Sigma=(I-R_{a\sim u\mid c}^{\T}R_{a\sim u\mid c})^{-1}$, we have $\mathrm{cov}(u_{\perp a,c})=I$ and $\mathrm{cov}(u_{\perp a,m,c})=(I-R_{m\sim u\mid a,c}^{\T}R_{m\sim u\mid a,c})^{1/2}$. The optimization problem reduces to  
\begin{equation}\label{eq::reduced-vu-indirect}
\begin{aligned}
&\min\quad && f\left(R_{m\sim u\mid a,c}(I-R_{m\sim u\mid a,c}^{\T} R_{m\sim u\mid a,c})^{-1/2}R_{y\sim u\mid a,m,c}^{\T},\frac{R_{m\sim u\mid a,c}R_{a\sim u\mid c}^{\T}}{(1-\|R_{a\sim u\mid c}\|_2^2)^{1/2}}\right)\\
&\text{subject to }\quad &&\|R_{a\sim u\mid c}\|_2^2 \leq \rho_a, \quad \|R_{m\sim u\mid a,c}\|_2^2 \leq \rho_m,\quad \|R_{y\sim u\mid a,m,c}\|_2^2\leq \rho_y.
\end{aligned}
\end{equation}
\end{proposition}

Based on the natural parameterization of direct and indirect effects in Section \ref{subsec::vector-u-parameter}, we can choose an appropriate function $f$ in Proposition \ref{prop::assume-cov-u-direct} and Proposition \ref{prop::assume-cov-u-indirect} to compute the worst point estimate or $t$ statistic for the direct and indirect effects, respectively, under the upper bound of $R$ parameters. The argument is similar to that in Section \ref{subsec::vector-u-ols-discussion}. The range result in Proposition \ref{prop::vector-u-range} is useful to numerically compute the optimization problem \eqref{eq::reduced-vu-direct} in Proposition \ref{prop::assume-cov-u-direct} and optimization problem \eqref{eq::reduced-vu-indirect} in Proposition \eqref{prop::assume-cov-u-indirect}.

\subsection{Report the sensitivity bound with formal benchmarking}\label{sec::formal-benchmarking}

\subsubsection{General method}

In this section, we focus on the discussion with a scalar $u$. Instead of directly bounding the original $R$ parameters, formal benchmarking bounds the relative strengths of $u$ compared to one covariate $c_j$, after controlling for the remaining covariates $c_{-j}$. Following \citet{cinelli2020making}, we define
\begin{equation}
\label{eq::ks-vector-u}
k_a^{(j)}=\frac{R_{a\sim u\mid c_{-j}}^2}{\hat{R}_{a\sim c_j\mid c_{-j}}^2},\quad k_m^{(j)}=\frac{\|R_{m\sim u\mid a,c_{-j}}\|_2^2}{\|\hat{R}_{m\sim c_j\mid a,c_{-j}}\|_2^2},\quad k_y^{(j)}=\frac{R_{y\sim u\mid a,m,c_{-j}}^2}{\hat{R}_{y\sim c_j\mid a,m,c_{-j}}^2},
\end{equation}
where $\hat{R}_{a\sim c_j\mid c_{-j}}$, $\hat{R}_{m\sim c_j\mid a,c_{-j}}$, $\hat{R}_{y\sim c_j\mid a,m,c_{-j}}$ are estimated $R$ parameters based on the observed data. They measure the norm ratios of the $R$ parameters after omitting the unmeasured confounder $u$ and after omitting the observed covariate $c_j$. 

The strategy of formal benchmarking aims to obtain the sensitivity bound based on the knowledge of the relative strengths $\{k_a^{(j)},k_m^{(j)},k_y^{(j)}\}$. We can assume that the maximum of $\{k_a^{(j)},k_m^{(j)},k_y^{(j)}\}$ is no greater than some threshold $\Delta$, for example, $\Delta=1$. It means that the strength of unmeasured confounder $u$ is no greater than $\Delta$ times the strength of the covariate $c_j$, after controlling for the remaining covariates $c_{-j}$. Then, we can report the worst point estimate or the worst $t$ statistic under the constraint $\max\{k_a^{(j)},k_m^{(j)},k_y^{(j)}\}\leq \Delta$. Another option is to report the minimum relative strength such that the point estimate can be reduced to $0$ or the confidence interval can be altered to cover $0$, similar to questions (Q1) and (Q2). We can also conduct the grid search to gradually increase $\Delta$, and find the critical values for the confidence interval and point estimate such that the $t$ statistic hits $1.96$ and $0$, respectively. 

Since $\hat{R}_{a\sim c_j\mid c_{-j}}$, $\hat{R}_{m\sim c_j\mid a,c_{-j}}$ and $\hat{R}_{y\sim c_j\mid a,m,c_{-j}}$ in \eqref{eq::ks-vector-u} are fixed and known from the observed data, we view formal benchmarking as optimization problems with respect to parameters $\{R_{a\sim u\mid c_{-j}}$, $R_{m\sim u\mid a,c_{-j}}$, $R_{y\sim u\mid a,m,c_{-j}}\}$. Following \citet{cinelli2020making}, we assume that $u$ is linearly independent of $c$, or equivalently, we consider only the part of $u$ not linearly explained by $c$. Proposition \ref{prop::benchmarking} below states that the original sensitivity parameters $R_u$ are functions of $\{R_{a\sim u\mid c_{-j}}$, $R_{m\sim u\mid a,c_{-j}}$, $R_{y\sim u\mid a,m,c_{-j}}\}$ and moment of observables. This implies that we can rewrite all the formulas for direct and indirect effects in terms of parameters $\{R_{a\sim u\mid c_{-j}}$, $R_{m\sim u\mid a,c_{-j}}$, $R_{y\sim u\mid a,m,c_{-j}}\}$.

\begin{proposition}
\label{prop::benchmarking}
Assume that the covariance matrix of $(y,m,a,c_j,u)_{\perp c_{-j}}$ is invertible for possibly vector $m,c$ but scalar $y,a,u$, and $u$ is linearly independent of $c$. We have
\begin{equation*}
\begin{split}
R_{a\sim u\mid c}&=\frac{R_{a\sim u\mid c_{-j}}}{(1-R_{a\sim c_j\mid c_{-j}}^2)^{1/2}}, \\
R_{m\sim u\mid a,c}&=\frac{\mathrm{cov}(m_{\perp a,c})^{-1/2}\mathrm{cov}(m_{\perp a,c_{-j}})^{1/2}(R_{m\sim u\mid a,c_{-j}}-R_{m\sim c_j\mid a,c_{-j}}R_{c_j\sim u\mid a,c_{-j}})}{(1-R_{c_j\sim u\mid a,c_{-j}}^2)^{1/2}}, \\
R_{y\sim u\mid a,m,c}&=\frac{R_{y\sim u\mid a,m,c_{-j}}-R_{y\sim c_j\mid a,m,c_{-j}}R_{c_j\sim u\mid a,m,c_{-j}}}{(1-R_{y\sim c_j\mid a,m,c_{-j}}^2)^{1/2}(1-R_{c_j\sim u\mid a,m,c_{-j}}^2)^{1/2}},\\
\end{split}
\end{equation*}
where
\begin{equation*}
\begin{split}
R_{c_j\sim u\mid a,c_{-j}}&=-\frac{R_{a\sim c_j\mid c_{-j}}R_{a\sim u\mid c_{-j}}}{(1-R_{a\sim u\mid c_{-j}}^2)^{1/2}(1-R_{a\sim c_j\mid c_{-j}}^2)^{1/2}},\\
R_{c_j\sim u\mid a,m,c_{-j}}&=\frac{R_{c_j\sim u\mid a,c_{-j}}-R_{m\sim c_j\mid a,c_{-j}}^{\T} R_{m\sim u\mid a,c_{-j}}}{(1-R_{c_j\sim m\mid a,c_{-j}}^2)^{1/2}\{1-\|R_{m\sim u\mid a,c_{-j}}\|_2^2)\}^{1/2}}.
\end{split}
\end{equation*}
\end{proposition}

Moreover, Corollary \ref{cor::benchmarking} below states that any feasbile specification of $\{R_{a\sim u\mid c_{-j}}$, $R_{m\sim u\mid a,c_{-j}}$, $R_{y\sim u\mid a,m,c_{-j}}\}$ is achievable, and thus guarantees the sharpness of the sensitivity bound. 

\begin{corollary}
\label{cor::benchmarking}
Assume that the covariance matrix of $(y,m,a,c_j,u)_{\perp c_{-j}}$ is invertible for possibly vector $m,c$ but scalar $y,a,u$, and $u$ is linearly independent of $c$. Given the observables $\{y,m,a,c\}$, the parameters $\{R_{y\sim u\mid a,m,c_{-j}},R_{m\sim u\mid a,c_{-j}},R_{a\sim u\mid c_{-j}}\}$ can take arbitrary values in $(-1,1)\times \mathbb{B}_{\mathrm{dim}(m)}\times (-1,1)$. 
\end{corollary}

\subsubsection{Illustration}

For illustration purposes, we report some sensitivity bounds with formal benchmarking based on \citet{chong2016iron}'s study in Section \ref{subsec::illustrations-multiple}. We consider the constraint $k_a^{(j)}=0$, $k_m^{(j)}\leq 1$, and $k_y^{(j)}\leq 1$ in \eqref{eq::ks-vector-u} for any single covariate $c_j$, and we report the worst point estimate for the direct and indirect effects. From Figure \ref{fig::fb1}, if the strength of unmeasured confounder could be as large as any single missing covariate, the worst point estimates would be still away from $0$. 

Since the gender variable is the most important covariate in Figure \ref{fig::fb1}, we focus on the specific covariate and report the minimum point estimates under the constraint $k_a^{(j)}=0$, $k_m^{(j)}\leq \Delta$, and $k_y^{(j)}\leq \Delta$ with varying $\Delta$. From Figure \ref{fig::fb2}, only if the strength of unmeasured confounder could be at least $3.9$ times as a missing gender variable, the point estimate for the indirect effect could be reduced to $0$; and only if the strength of unmeasured confounder could be at least $5.5$ times as a missing gender variable, the point estimate for the direct effect could be reduced to $0$.

\begin{figure} 
\begin{subfigure}{0.50\textwidth}
\includegraphics[width=3in]{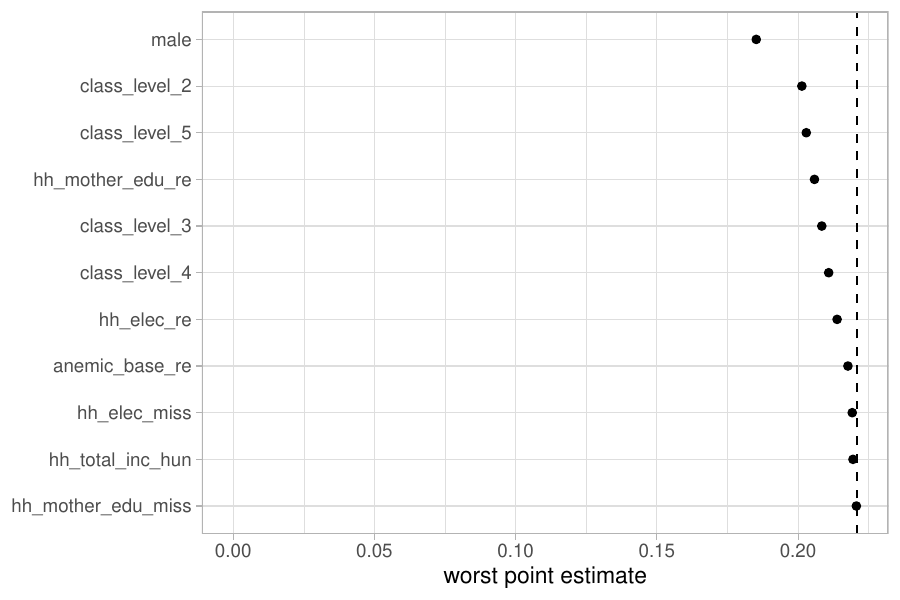}
\caption{direct effect}
\end{subfigure}\hspace*{\fill}
\begin{subfigure}{0.50\textwidth}
\includegraphics[width=3in]{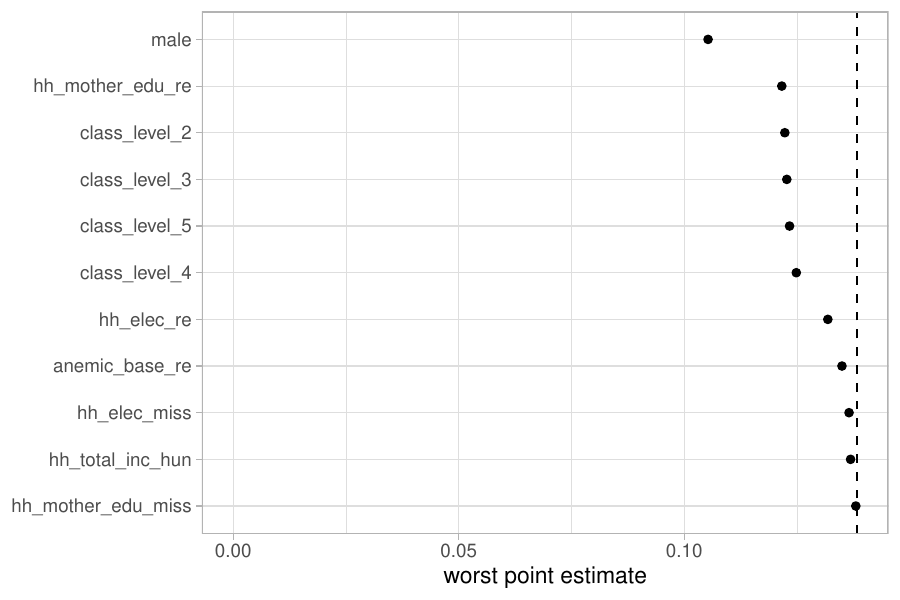}
\caption{indirect effect}
\end{subfigure}
\caption{The worst point estimates under the constraint $k_a^{(j)}=0$, $k_m^{(j)}\leq 1$, $k_y^{(j)}\leq 1$ for each covariate $c_j$. The reference lines indicate the observed point estimates. } 
\label{fig::fb1}
\end{figure}

\begin{figure} 
\centering
\includegraphics[width=4in]{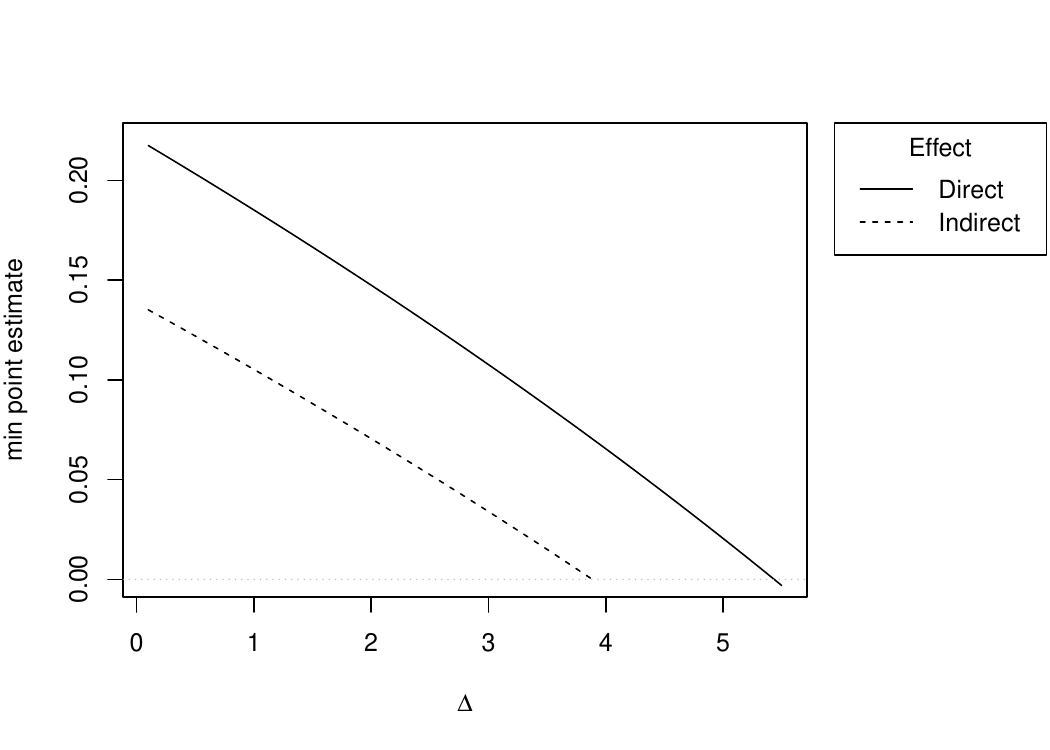}
\caption{The worst point estimates under the constraint $k_a^{(j)}=0$, $k_m^{(j)}\leq \Delta$, $k_y^{(j)}\leq \Delta$ for the gender variable $c_j$. } 
\label{fig::fb2}
\end{figure}

\subsection{Simulation: how much do we miss by assuming a scalar $u$?}\label{sec::computing-rv}

As we discuss in Section \ref{subsubsec::rv-vu}, it is restrictive to assume $u$ is a scalar to compute the robustness values for the indirect effect. In this section, we study how much we miss if we assume $u$ is a scalar by simulations. We consider the following simulation setting. The exposure $a$ follows a standard normal distribution $N(0,1)$. Given $a$, the mediator $m$ follows a multivariate normal distribution with mean $\alpha_1a$ and variance $\Sigma$, where 
$$\alpha_1=\lambda_1\left(0,\frac{1}{\mathrm{dim}(m)-1},\frac{2}{\mathrm{dim}(m)-1},\dots,\frac{\mathrm{dim}(m)-2}{\mathrm{dim}(m)-1},1\right)$$
with $\lambda_1$ selected later, and $\Sigma_{ij}=0.5^{|i-j|}$, for all $1\leq i,j\leq \mathrm{dim}(m)$. Given $a$ and $m$, the outcome $y$ follows a normal distribution with mean $a+\alpha_2^{\T} m$ and variance $1$, where 
$$\alpha_2=\lambda_2\left(1,1-\frac{1.5}{\mathrm{dim}(m)-1},1-\frac{3}{\mathrm{dim}(m)-1},\dots,1-\frac{1.5\{\mathrm{dim}(m)-2\}}{\mathrm{dim}(m)-1},-0.5\right)$$
with $\lambda_2$ selected later. 

We vary the three factors 
$$\mathrm{dim}(m)\in \{2,4,6\}, \quad R_{a\sim m\mid c}^2\in\{0.1,0.3,0.5,0.7\}, \quad R_{y\sim m\mid a,c}^2\in\{0.1,0.3,0.5,0.7\}$$
in our simulation design. For the $R^2$ parameters, we select $\lambda_1>0$ to achieve a specific level of $R_{a\sim m\mid c}^2$ and we select $\lambda_2>0$ to achieve a specific level of $R_{y\sim a\mid m,c}^2$. Tables \ref{tab::simulation1}--\ref{tab::simulation2} compare the robustness values with scalar $u$ and vector $u$, for the point estimates and $95\%$ confidence intervals, respectively. Numerically, we can observe that the robustness values with vector $u$ are about $3/4$ times the robustness values with scalar $u$. 

We also conduct an analysis-of-variance test to analyze the effect of the three factors on the ratio of robustness values with vector $u$ and with scalar $u$. From Table \ref{tab::simulation3}, we can find that the dimension of $m$ is the most important factor. However, we observe that when the dimension of $m$ increases, the average ratio of robustness values increases. Therefore, the higher dimension of $m$ would not lead to a smaller ratio of the robustness values, based on the current simulation results.

\begin{table*}[t]%
\centering
\caption{Simulation results of the average robustness values for point estimates with scalar $u$, over $20$ data replications. The values in the parentheses are the average robustness values for point estimates with vector $u$. }\label{tab::simulation1}%
\begin{tabular*}{\textwidth}{@{\extracolsep\fill}llllll@{\extracolsep\fill}}%
\toprule
\multirow{2}{*}{$\mathrm{dim}(m)$} & \multirow{2}{*}{$R_{a\sim m\mid c}^2$} &  \multicolumn{4}{c}{$R_{y\sim m\mid a,c}^2$}    \\
\cmidrule(lr){3-6}
&  &  $0.1$ & $0.3$ & $0.5$ & $0.7$   \\
\midrule
\multirow{4}{*}{$2$} & $0.1$ & 0.090 (0.064) & 0.126 (0.096) & 0.171 (0.136) & 0.173 (0.150) \\
& $0.3$ & 0.138 (0.104) & 0.236 (0.171) & 0.236 (0.176) & 0.280 (0.220) \\
& $0.5$ & 0.125 (0.102) & 0.245 (0.185) & 0.312 (0.232) & 0.364 (0.280) \\
& $0.7$ & 0.134 (0.115) & 0.264 (0.214) & 0.359 (0.279) & 0.429 (0.329) \\
\midrule
\multirow{4}{*}{$4$} & $0.1 $ & 0.054 (0.036) & 0.082 (0.062) & 0.078 (0.062) & 0.079 (0.064) \\
& $0.3$ & 0.071 (0.052) & 0.102 (0.072)  & 0.126 (0.092) & 0.118 (0.090) \\
& $0.5$ & 0.086 (0.068) & 0.116 (0.084) & 0.165 (0.119) & 0.180 (0.132)  \\
& $0.7$ & 0.080 (0.067) &  0.130 (0.100) & 0.162 (0.121)  & 0.238 (0.172) \\
\midrule
\multirow{4}{*}{$6$} & $0.1$ & 0.036 (0.024) & 0.056 (0.040) & 0.072 (0.054) & 0.061 (0.050) \\
& $0.3$ & 0.058 (0.043) & 0.089 (0.060) & 0.097 (0.069) & 0.093 (0.067) \\
& $0.5$ & 0.080 (0.061) & 0.095 (0.064)  & 0.119 (0.078) & 0.118 (0.082) \\
& $0.7$ & 0.076 (0.061) & 0.118 (0.085) & 0.103 (0.069) & 0.179 (0.126) \\
\bottomrule
\end{tabular*}
\end{table*}

\begin{table*}[t]%
\centering
\caption{Simulation results of the average robustness values for $95\%$ confidence intervals with scalar $u$, over $20$ data replications. The values in the parentheses are the average robustness values for $95\%$ confidence intervals with vector $u$. }\label{tab::simulation2}%
\begin{tabular*}{\textwidth}{@{\extracolsep\fill}llllll@{\extracolsep\fill}}%
\toprule
\multirow{2}{*}{$\mathrm{dim}(m)$} & \multirow{2}{*}{$R_{a\sim m\mid c}^2$} &  \multicolumn{4}{c}{$R_{y\sim m\mid a,c}^2$}    \\
\cmidrule(lr){3-6}
&  &  $0.1$ & $0.3$ & $0.5$ & $0.7$   \\
\midrule
\multirow{4}{*}{$2$} & $0.1$ & 0.002 (0.001) & 0.018 (0.013) & 0.062 (0.049) & 0.060 (0.052) \\
& $0.3$ & 0.024 (0.020) & 0.121 (0.086) & 0.124 (0.092) & 0.172 (0.134) \\
& $0.5$ & 0.022 (0.015) & 0.128 (0.094) & 0.207 (0.152) & 0.265 (0.199) \\
& $0.7$ & 0.035 (0.031) & 0.148 (0.118) & 0.259 (0.199) & 0.342 (0.258) \\
\midrule
\multirow{4}{*}{$4$} & $0.1 $ & 0.000 (0.000) & 0.006 (0.004) & 0.004 (0.002) & 0.010 (0.008) \\
& $0.3$ & 0.005 (0.003) & 0.016 (0.011)  & 0.024 (0.017) & 0.020 (0.014) \\
& $0.5$ & 0.005 (0.004) & 0.018 (0.012) & 0.048 (0.034) & 0.062 (0.045)  \\
& $0.7$ & 0.006 (0.006) &  0.026 (0.020) & 0.057 (0.042)  & 0.121 (0.087) \\
\midrule
\multirow{4}{*}{$6$} & $0.1$ & 0.000 (0.000) & 0.000 (0.000) & 0.006 (0.004) & 0.000 (0.000) \\
& $0.3$ & 0.002 (0.002) & 0.008 (0.005) & 0.014 (0.009) & 0.012 (0.010) \\
& $0.5$ & 0.002 (0.002) & 0.014 (0.010)  & 0.018 (0.010) & 0.020 (0.013) \\
& $0.7$ & 0.005 (0.004) & 0.013 (0.010) & 0.015 (0.009) & 0.065 (0.042) \\
\bottomrule
\end{tabular*}
\end{table*}

\begin{table*}[t]%
\centering
\caption{Analysis of variance for the ratio of robustness values with vector $u$ and with scalar $u$. }\label{tab::simulation3}%
\begin{tabular*}{\textwidth}{@{\extracolsep\fill}llllll@{\extracolsep\fill}}%
\toprule
Ratio of robustness values & Factor & df & MSE & $F$-value & $p$-value\\
 \midrule
\multirow{3}{*}{point estimate} & $\mathrm{dim}(m)$ & 2 & 0.387 & 26.075 & $9.63\times 10^{-12}$ \\
 & $R_{a\sim m\mid c}^2$ & 3 & 0.076 & 5.086 & 0.0017 \\
 & $R_{y\sim m\mid a,c}^2$ & 3 & 0.044 & 2.992 & 0.0301 \\
 \midrule
 \multirow{3}{*}{95\% confidence interval} & $\mathrm{dim}(m)$ & 2 & 0.215 & 9.250 & $1.17\times 10^{-4}$ \\
& $R_{a\sim m\mid c}^2$  & 3 & 0.119 & 5.125 & 0.0017 \\
 & $R_{y\sim m\mid a,c}^2$  & 3 & 0.057 & 2.465 & 0.0618 \\
\bottomrule
\end{tabular*}
\end{table*}

\subsection{General sensitivity analysis with multiple unmeasured confounders}\label{sec::general-vector-u}

In this section, we discuss sensitivity analysis in the context of general hypothesis testing. Previously, we emphasized testing whether the direct effect is zero and whether the indirect effect is zero. However, other hypotheses may also be of scientific interest. For example, we may wish to test whether both the direct and indirect effects are zero or assess a specific linear combination of these effects. In Sections \ref{sec::scale-u} and \ref{sec::vector-u}, we demonstrated that the direct and indirect effects are functions of the moments of observables and sensitivity parameters. Therefore, the test statistics for these general hypotheses can also be expressed in terms of sensitivity parameters. When there is a single unmeasured confounder, these problems are not more challenging. We can use the strategies in Section \ref{sec::scale-u} to report the sensitivity analysis based on the appropriate formulation of optimization problems with respect to $R$ parameters. When there are multiple unmeasured confounders, based on the strategy in Section \ref{sec::vector-u}, we can formulate it as the optimization problem
\begin{equation*}
\begin{aligned}
    &\min\quad &&f(R_{a\sim u\mid c}, R_{m\sim u\mid a,c}, R_{y\sim u\mid a,m,c},\mathrm{cov}(u_{\perp c}))\\
    &\text{subject to }\quad &&\|R_{a\sim u\mid c}\|_2^2\leq \rho_a,\quad  \|R_{m\sim u\mid a,c}\|_2^2\leq \rho_m,\quad  \|R_{y\sim u\mid a,m,c}\|_2^2\leq \rho_y
\end{aligned}
\end{equation*}
for some functions $f$ as our test statistic. The problem is challenging due to the unknown dimension of $u$. In this section, we provide a solution to the above optimization problem without the need to specify the dimension of $u$. The following Proposition \ref{prop::general-sensitivity} is our key result.

\begin{proposition}\label{prop::general-sensitivity}
    Assume that the covariance matrix of $(a,m,y,u)_{\perp c}$ is invertible and assume $\mathrm{cov}(u_{\perp a,c})=I$. Let 
    \begin{equation*}
    R=\begin{pmatrix}
    R_{a\sim u\mid c}\\
    R_{m\sim u\mid a,c}\\
    R_{y\sim u\mid a,m,c}
    \end{pmatrix}.
    \end{equation*}
    The covariance matrix of $(a_{\perp c,u}, m_{\perp c,u},y_{\perp c,u})$ depends only on the moments of observables and $RR^{\T}$. 
\end{proposition}

We make the following remarks regarding Proposition \ref{prop::general-sensitivity} below. First, we would like to compute the sensitivity bounds under the upper bound restriction of $R$ parameters. Similar to the discussion in Section \ref{subsec::ovb-vector-u}, we can assume $\mathrm{cov}(u_{\perp c})=(I-R_{a\sim u\mid c}^{\T} R_{a\sim u\mid c})^{-1}$, which implies $\mathrm{cov}(u_{\perp a,c})=I$ and $\mathrm{cov}(u_{\perp a,c,m})=I-R_{m\sim u\mid a,c}^{\T} R_{m\sim u\mid a,c}$. This is not restrictive for the purpose of reporting sensitivity bound under the upper bound of $R$ parameters.

Second, by Proposition \ref{prop::general-sensitivity}, the true direct effect $\theta_1$ and true indirect effect $\theta_3^{\T}\beta_1$ are functions of moments of observables and $RR^{\T}$. Therefore, the test statistic $f$ is also a function of $RR^{\T}$, regardless of the dimension of $u$. 

Third, by Proposition \ref{prop::general-sensitivity}, we can search over the space of $RR^{\T}$ instead, the space of $(\mathrm{dim}(m)+2)\times (\mathrm{dim}(m)+2)$ positive semidefinite matrices. For a specific positive semidefinite matrix $X$, there are numerous solutions to the equation $X=RR^{\T}$. However, the function $f$ will take the same value for all different solutions of $R$, and thus we only need to choose a specific solution, such as $R=X^{1/2}$. We use functions $\xi_1(R)$ to denote the first row of matrix $R$, $\xi_3(R)$ to denote the last row of matrix $R$, and $\xi_2(R)$ to denote all other rows of matrix $R$. Under the assumption $\mathrm{cov}(u_{\perp a,c})=I$, we have $\mathrm{cov}(u_{\perp c})=(I-R_{a\sim u\mid c}^{\T}R_{a\sim u\mid c})^{-1}$. Then, we can reformulate the optimization problem below: 
\begin{equation*}
\begin{aligned}
    &\min\quad &&f(\xi_1(X^{1/2}), \xi_2(X^{1/2}), \xi_3(X^{1/2}),\{I-\xi_1(X^{1/2})^{\T}\xi_1(X^{1/2})\}^{-1})\\
    &\text{subject to }\quad && X \text{ is positive semidefinite, }\\
    & &&\|\xi_1(X^{1/2})\|_2^2\leq \rho_a,\quad \|\xi_2(X^{1/2})\|_2^2\leq \rho_m,\quad \|\xi_3(X^{1/2})\|_2^2\leq \rho_y.
\end{aligned}
\end{equation*}
This is an optimization problem without the need to specify the dimension of $u$. 

Fourth, the above optimization problem is more complicated than semidefinite programming because the objective function is not linear in $X$. We can leverage the result in \citet{zhang2025stronger} to solve the problem, which provides a reparameterization to nonlinear semidefinite programming by translating the positive semidefinite constraint into some elementwise constraints.

\section{R package: \texttt{BaronKennyU}}\label{sec::package}
We provide an R package \texttt{BaronKennyU} to conduct the sensitivity analysis for the Baron--Kenny approach. The package contains R functions to compute the point estimates, standard errors, and $t$ statistics for direct and indirect effects with prespecified $R$ parameters when $u$ is a scalar. As discussed in Section \ref{subsec::scale-u-report} and Section \ref{subsec::vector-u-report}, the package also contains R functions to compute the worst-case result under the upper bound $R$ parameters, and to determine the robustness values for mediation analysis, both when $u$ is a scalar and when $u$ is a vector. Additionally, as discussed in Section \ref{sec::formal-benchmarking}, the package contains R functions to compute the sensitivity bounds with formal benchmarking when $u$ is a scalar. The package can be installed by the R command: 
\begin{verbatim}
devtools::install_github("mingrui229/BaronKennyU", upgrade = "never")
\end{verbatim}

First, with prespecified $R$ parameters, the R functions \texttt{bku\_direct} and \texttt{bku\_indirect} compute the point estimate, standard error, and $t$ statistic for the direct and indirect effects, respectively. Second, with prespecified norm constraints $\rho_a,\rho_m,\rho_y$, the R functions \texttt{bku\_worst\_direct} and \texttt{bku\_worst\_indirect} compute the worst point estimate and worst $t$ statistic for the direct and indirect effects, respectively. Third, the R functions \texttt{bku\_rv\_direct} and \texttt{bku\_rv\_indirect} compute the robustness values for both point estimates and 95\% confidence intervals for the direct and indirect effects, respectively. It returns a summary table like Table \ref{tab::ex1} or Table \ref{tab::ex2}. Fourth, the R functions \texttt{bku\_fb\_direct} and \texttt{bku\_fb\_indirect} compute the worst point estimate and $t$ statistic with formal benchmarking for the direct and indirect effects, respectively. Besides the data $a$, $m$, $y$, $c_j$ and $c_{-j}$, users also need to specify the upper bound of $k_a^{(j)}$, $k_m^{(j)}$, $k_y^{(j)}$ in \eqref{eq::ks-vector-u}. 

We take the example in Section \ref{subsec::illustrations-multiple} to illustrate the usage of the above R functions. With prespecified $R$ parameters $\{R_{y\sim u\mid a,m,c}=0.6,R_{m\sim u\mid a,c}=(0.5,0.5)^{\T}, R_{a\sim u\mid c}=0\}$, we can compute the sensitivity bound for the direct and indirect effects using the R code
\begin{verbatim}
bku_direct(y = y, m = m, a = a, c = c, Ry = 0.6, Rm = c(0.5, 0.5), Ra = 0)
bku_indirect(y = y, m = m, a = a, c = c, Ry = 0.6, Rm = c(0.5, 0.5), Ra = 0)
\end{verbatim}
We can compute the robustness values using the R code
\begin{verbatim}
bku_rv(y = y, m = m, a = a, c = c, randomized = T)
\end{verbatim}
With prespecified upper bound $k_a^{(j)}=0$, $k_m^{(j)}\leq 1$, $k_y^{(j)}\leq 1$ in \eqref{eq::ks-vector-u} for covariate $c_1$, we can compute the sensitivity bounds with formal benchmarking 
\begin{verbatim}
bku_fb_direct(y = y, m = m, a = a, c = c, j = 1, 
              ky_bound = 1, km_bound = 1, ka_bound = 0)
bku_fb_indirect(y = y, m = m, a = a, c = c, j = 1, 
                ky_bound = 1, km_bound = 1, ka_bound = 0)
\end{verbatim}

\section{Proofs}\label{sec::proofs}

\subsection{Some useful lemmas and Proof of Lemma \ref{lem::r2-and-r}}\label{subsec::further-useful-lemmas}

First, we present formulas of covariance of residuals in ordinary least squares in Lemma \ref{lem::cov-formula}, which is useful for the properties of both $R$ parameters and $R^2$ parameters. 

\begin{lemma}
\label{lem::cov-formula}
For centered vectors $x,y,z$, we have
\begin{equation*}
\begin{split}
\mathrm{cov}(y_{\perp x})&=\mathrm{cov}(y)-\mathrm{cov}(y,x)\mathrm{cov}(x)^{-1}\mathrm{cov}(x,y),\\
\mathrm{cov}(y_{\perp x},z_{\perp x})&=\mathrm{cov}(y,z)-\mathrm{cov}(y,x)\mathrm{cov}(x)^{-1}\mathrm{cov}(x,z),\\
\end{split}
\end{equation*}
if $\mathrm{cov}(x)$ is invertible. 
\end{lemma}

\begin{proof}[Proof of Lemma \ref{lem::cov-formula}]
Consider the population least squares $y=\beta x+y_{\perp x}$ with 
$\beta=\mathrm{cov}(y,x)\mathrm{cov}(x)^{-1}$
and 
$y_{\perp x}=y-\mathrm{cov}(y,x)\mathrm{cov}(x)^{-1}x.$
Similarly, we have
$z_{\perp x}=z-\mathrm{cov}(z,x)\mathrm{cov}(x)^{-1}x.$
We can obtain the desired results by direct calculation. 
\end{proof}

Next, we give a proof of Lemma \ref{lem::r2-and-r}, which connects $R$ parameters and $R^2$ parameters. Then, we present some more useful properties of $R$ parameters in Lemma \ref{lem::r-exp} and $R^2$ parameters in Lemma \ref{lem::rep1}--\ref{lem::rep2} below. 

\begin{proof}[Proof of Lemma \ref{lem::r2-and-r}]
For centered scalar $y$ and centered vector $x$, by definition, we have
$$\|R_{y\sim x}^{\T}\|_2^2=\frac{\mathrm{cov}(y,x)\mathrm{cov}(x)^{-1}\mathrm{cov}(x,y)}{\mathrm{var}(y)}$$
and by Lemma \ref{lem::cov-formula}, we have
$$R_{y\sim x}^2=1-\frac{\mathrm{var}(y_{\perp x})}{\mathrm{var}(y)}=\frac{\mathrm{cov}(y,x)\mathrm{cov}(x)^{-1}\mathrm{cov}(x,y)}{\mathrm{var}(y)}.$$
Thus, we have $\|R_{y\sim x}^{\T}\|_2^2=R_{x\sim y}^2$. For centered vectors $y$ and $z$, $\|R_{y\sim x}\|_2^2$ is the largest eigenvalue of $R_{y\sim x}R_{y\sim x}^{\T}$. That is, 
\begin{equation*}
\begin{split}
\|R_{y\sim x}^{\T}\|_2^2&=\max_{\beta\not=0}\frac{\beta^{\T} \mathrm{cov}(y)^{-1/2}\mathrm{cov}(y,x)\mathrm{cov}(x)^{-1}\mathrm{cov}(x,y)\mathrm{cov}(y)^{-1/2}\beta}{\beta^{\T} \beta}\\
&=\max_{\alpha\not=0}\frac{\alpha^{\T}\mathrm{cov}(y,x)\mathrm{cov}(x)^{-1}\mathrm{cov}(x,y)\alpha}{\alpha^{\T}\mathrm{cov}(y)^{-1}\alpha}\\
&=\max_{\alpha\not=0} R_{\alpha^{\T}y\sim x}^2. 
\end{split}
\end{equation*}
\end{proof}

\begin{lemma}
\label{lem::r-exp}
For centered vectors $x,y,z$, we have
$$R_{y\sim z\mid x}=\mathrm{cov}(y_{\perp x})^{-1/2}\mathrm{cov}(y)^{1/2}(R_{y\sim z}-R_{y\sim x}R_{x\sim z})\mathrm{cov}(z)^{1/2}\mathrm{cov}(z_{\perp x})^{-1/2},$$
if $\mathrm{cov}\{(x,y,z)\}$ is invertible. 
\end{lemma}

\begin{proof}[Proof of Lemma \ref{lem::r-exp}]
By Lemma \ref{lem::cov-formula}, we have
\begin{equation*}
\begin{split}
&\mathrm{cov}(y_{\perp x})^{-1/2}\mathrm{cov}(y)^{1/2}(R_{y\sim z}-R_{y\sim x}R_{x\sim z})\mathrm{cov}(z)^{1/2}\mathrm{cov}(z_{\perp x})^{-1/2}\\
=&\mathrm{cov}(y_{\perp x})^{-1/2}\left\{\mathrm{cov}(y,z)-\mathrm{cov}(y,x)\mathrm{cov}(x)^{-1}\mathrm{cov}(x,z)\right\}\mathrm{cov}(z_{\perp x})^{-1/2}\\
=&\mathrm{cov}(y_{\perp x})^{-1/2}\mathrm{cov}(y_{\perp x},z_{\perp x})^{-1/2}\mathrm{cov}(z_{\perp x})^{-1/2}\\
=&R_{y\sim z\mid x}.
\end{split}
\end{equation*}
\end{proof}

\begin{lemma}
\label{lem::rep1}
For scalar $y$ and vectors $x$ and $z$, we have
$$R_{y\sim z\mid x}^2=\frac{R_{y\sim (z,x)}^2-R_{y\sim x}^2}{1-R_{y\sim x}^2}=1-\frac{(1-R_{y\sim x\mid z}^2)(1-R_{y\sim z}^2)}{1-R_{y\sim x}^2},$$
if $\mathrm{cov}\{(x,y,z)\}$ is invertible. 
\end{lemma}
\begin{proof}[Proof of Lemma \ref{lem::rep1}]
By Definition \ref{def::traditional-r2}, we have
\begin{equation}
\label{eq::lem-rep1-1}
R_{y\sim z\mid x}^2=\frac{\mathrm{var}(y_{\perp x})-\mathrm{var}(y_{\perp x,z})}{\mathrm{var}(y_{\perp x})}=\frac{R_{y\sim (z,x)}^2-R_{y\sim x}^2}{1-R_{y\sim x}^2}.
\end{equation}
Swapping $x$ and $z$, we also have
\begin{equation}
\label{eq::lem-rep1-2}
R_{y\sim x\mid z}^2=\frac{R_{y\sim (x,z)}^2-R_{y\sim z}^2}{1-R_{y\sim z}^2}.
\end{equation}
Combining \eqref{eq::lem-rep1-1}--\eqref{eq::lem-rep1-2}, we obtain the desired result. 
\end{proof}

\begin{lemma}
\label{lem::rep2}
For centered nonzero scalars $y,x,z$, we have
$$
R_{y\sim z}^2=\left[S(R_{y\sim x}^2R_{z\sim x}^2)^{1/2}+S'\{R_{y\sim z\mid x}^2(1-R_{y\sim x}^2)(1-R_{z\sim x}^2)\}^{1/2}\right]^2
$$
if $\mathrm{cov}\{(x,y,z)\}$ is invertible, where 
$$
S=\mathrm{sgn}\{\mathrm{corr}(y,x)\cdot\mathrm{corr}(z,x)\},\quad 
S'=\mathrm{sgn}\{\mathrm{corr}(y_{\perp x},z_{\perp x})\}.
$$ 
\end{lemma}
Lemma \ref{lem::rep2} is a direct corollary of Lemma \ref{lem::r-exp}. The following Lemmas \ref{lem::block-matrix-psd}--\ref{lem::general-matrix-sharpness} are important to derive our sharpness results. 
\begin{lemma}
\label{lem::block-matrix-psd}
Assume $A$ is symmetric and $C$ is positive definite. Then
\begin{equation*}
M=\begin{pmatrix}
A & B \\
B^{\T}  & C
\end{pmatrix}
\end{equation*}
is positive definite if and only if $A-BC^{-1}B^{\T} $ is positive definite. 
\end{lemma}

Lemma \ref{lem::block-matrix-psd} is a standard property of the Schur complement. We omit its proof. 

\begin{lemma}
\label{lem::general-matrix-sharpness}
(i) Assume any given random vectors $a,b$ such that the covariance matrix $\mathrm{cov}\{(a,b)\}$ is invertible. For any positive definite matrices
\begin{equation*}\Sigma_1=
\begin{pmatrix}
\Sigma_{aa} &  \Sigma_{au}\\
\Sigma_{ua} & \Sigma_{uu}
\end{pmatrix},\quad \Sigma_2=
\begin{pmatrix}
\Sigma_{bb\mid a} & \Sigma_{bu\mid a}\\
\Sigma_{ub\mid a} & \Sigma_{uu\mid a}
\end{pmatrix},
\end{equation*}
satisfying $\Sigma_{uu\mid a}=\Sigma_{uu}-\Sigma_{ua}\Sigma_{aa}^{-1}\Sigma_{au}$, $\mathrm{cov}(a)=\Sigma_{aa}$, and $\mathrm{cov}(b_{\perp a})=\Sigma_{bb\mid a}$, then there exists a vector $u$ such that $\mathrm{cov}\{(a,u)\}=\Sigma_1$, $\mathrm{cov}\{(b_{\perp a},u_{\perp a})\}=\Sigma_2$, and $\mathrm{cov}\{(a,b,u)\}$ is invertible. 

(ii) Assume any given random vectors $a,b,c$ such that the covariance matrix $\mathrm{cov}\{(a,b,c)\}$ is invertible. For any positive definite matrices
\begin{equation*}\Sigma_1=
\begin{pmatrix}
\Sigma_{aa} &  \Sigma_{au}\\
\Sigma_{ua} & \Sigma_{uu}
\end{pmatrix},\quad \Sigma_2=
\begin{pmatrix}
\Sigma_{bb\mid a} & \Sigma_{bu\mid a}\\
\Sigma_{ub\mid a} & \Sigma_{uu\mid a}
\end{pmatrix},\quad \Sigma_3=
\begin{pmatrix}
\Sigma_{cc\mid a,b} & \Sigma_{cu\mid a,b}\\
\Sigma_{uc\mid a,b} & \Sigma_{uu\mid a,b}
\end{pmatrix}
\end{equation*}
satisfying $\Sigma_{uu\mid a}=\Sigma_{uu}-\Sigma_{ua}\Sigma_{aa}^{-1}\Sigma_{au}$, $\Sigma_{uu\mid a,b}=\Sigma_{uu\mid a}-\Sigma_{ub\mid a}\Sigma_{bb\mid a}^{-1}\Sigma_{bu\mid a}$, $\mathrm{cov}(a)=\Sigma_{aa}$, $\mathrm{cov}(b_{\perp a})=\Sigma_{bb\mid a}$ and $\mathrm{cov}(c_{\perp a,b})=\Sigma_{cc\mid a,b}$, then there exists a vector $u$ such that $\mathrm{cov}\{(a,u)\}=\Sigma_1$, $\mathrm{cov}\{(b_{\perp a},u_{\perp a})\}=\Sigma_2$, $\mathrm{cov}\{(c_{\perp a,b},u_{\perp a,b})\}=\Sigma_3$, and $\mathrm{cov}\{(a,b,c,u)\}$ is invertible. 
\end{lemma}

\begin{proof}[Proof of Lemma \ref{lem::general-matrix-sharpness}]
We only prove (ii) because the proof of (i) is similar. Let 
$$\Sigma_{vv}=\Sigma_{uu\mid a,b}-\Sigma_{uc\mid a,b}\Sigma_{cc\mid a,b}^{-1}\Sigma_{cu\mid a,b}.$$
By Lemma \ref{lem::block-matrix-psd}, $\Sigma_{vv}$ is positive definite. Construct 
$$u=\Sigma_{ua}\Sigma_{aa}^{-1}a+\Sigma_{ub\mid a}\Sigma_{bb\mid a}^{-1}b_{\perp a}+\Sigma_{uc\mid a,b}\Sigma_{cc\mid a,b}^{-1}c_{\perp a,b}+v,$$
where $v$ is independent of $(a,b,c)$ with $\mathrm{cov}(v)=\Sigma_{vv}$. We can check that the conditions hold. 
\end{proof}

Finally, the following Lemma \ref{lem::linear-algebra-matrix} is useful to prove Proposition \ref{prop::general-sensitivity}. 

\begin{lemma}\label{lem::linear-algebra-matrix}
Let $X\in \mathbb{R}^{p_1\times q}$, $Y\in \mathbb{R}^{p_2\times q}$, $Z\in \mathbb{R}^{p_3\times q}$ with $\|X\|_2,\|Y\|_2,\|Z\|_2<1$. Then $X(I-Y^{\T}Y)^{1/2}Z^{\T}$ and $X(I-Y^{\T}Y)^{-1/2}Z^{\T}$ depend only on $XY^{\T}, XZ^{\T}, YY^{\T}, YZ^{\T}$. 
\end{lemma}

\begin{proof}[Proof of Lemma \ref{lem::linear-algebra-matrix}]
Since matrix multiplication is associative and linear, it suffices to prove the statement for the case where $p_1 = p_3 = 1$, i.e., when $X$ and $Z$ are row vectors. The general case follows by applying the result row-wise. 

Suppose the matrix $YY^{\T}$ has $k$ nonzero eigenvalues $\lambda_1,\dots,\lambda_k$, corresponding to $k$ orthonormal eigenvectors $v_1,\dots,v_k$, respectively. Then we have the eigendecomposition of $YY^{\T}$:
$$YY^{\T}=\sum_{i=1}^k\lambda_iv_iv_i^{\T},$$
which implies
$$Y^{\T}YY^{\T}v_i=\lambda_iY^{\T}v_i,\quad \text{for all }1\leq i\leq k,$$
and
$$v_i^{\T}YY^{\T}v_j=0,\quad \text{for all }1\leq i<j\leq k.$$
By normalizing, the matrix $Y^{\T}Y$ has $k$ nonzero eigenvalues $\lambda_1,\dots,\lambda_k$, corresponding to $k$ orthonormal eigenvectors $Y^{\T}v_1/\|Y^{\T}v_1\|_2$,\dots, $Y^{\T}v_k/\|Y^{\T}v_k\|_2$, respectively. Thus, $(I-Y^{\T}Y)^{1/2}$ has $k$ eigenvalues $(1-\lambda_1)^{1/2},\dots,(1-\lambda_k)^{1/2}$, corresponding to $k$ orthonormal eigenvectors $Y^{\T}v_1/\|Y^{\T}v_1\|_2$,\dots, $Y^{\T}v_k/\|Y^{\T}v_k\|_2$, respectively, and has $(q-k)$ eigenvalues equal to 1, corresponding to $(q-k)$ orthonormal vectors $w_1,\dots,w_{q-k}\in\mathbb{R}^q$, which form an orthonormal basis for the null space of $Y$ (the subspace orthogonal to its row space). Therefore, the eigendecomposition of $(I-Y^{\T}Y)^{1/2}$ is
$$(I-Y^{\T}Y)^{1/2}=\sum_{i=1}^k(1-\lambda_i)^{1/2}\frac{Y^{\T}v_i}{\|Y^{\T}v_i\|_2}\left(\frac{Y^{\T}v_i}{\|Y^{\T}v_i\|_2}\right)^{\T}+\sum_{i=1}^{q-k}w_iw_i^{\T},$$
which implies
    $$X(I-Y^{\T}Y)^{1/2}Z^{\T}=\sum_{i=1}^k(1-\lambda_i)^{1/2}\frac{XY^{\T}v_iv_i^{\T}YZ^{\T}}{v_i^{\T}YY^{\T}v_i}+\sum_{i=1}^{q-k}Xw_iw_i^{\T}Z^{\T}.$$
Since $\lambda_i$ and $v_i$ depend only on $YY^{\T}$, it remains to show that the second term $\sum_{i=1}^{q-k}Xw_iw_i^{\T}Z^{\T}$ depends only on $XY^{\T}$, $XZ^{\T}$, $YY^{\T}$, and $YZ^{\T}$.  

The term $\sum_{i=1}^{q-k} w_i w_i^{\T}$ is the projection matrix onto the null space of $Y$. Since $Y$ may be rank-deficient, this projection matrix can be written as
$$\sum_{i=1}^{q-k}w_iw_i^{\T}=I-Y^{\T}(YY^{\T})^\dagger Y,$$
where $(YY^{\T})^\dagger$ is the Moore--Penrose pseudoinverse of $YY^{\T}$. Thus, applying this projection to $X$ and $Z$, we obtain
    $$\sum_{i=1}^{q-k}Xw_iw_i^{\T}Z^{\T}=XZ^{\T}-XY^{\T}(YY^{\T})^\dagger YZ^{\T}.$$
Since all terms on the right-hand side depend only on $XY^{\T}$, $XZ^{\T}$, $YY^{\T}$, and $YZ^{\T}$, it follows that $X(I - Y^{\T} Y)^{1/2} Z^{\T}$ does as well. 
    
    Similarly, we can show that $X(I-Y^{\T}Y)^{-1/2}Z^{\T}$ depends only on $XY^{\T}, XZ^{\T}, YY^{\T}, YZ^{\T}$. 
\end{proof}

\subsection{Proofs of Theorems \ref{thm::multi-ols}--\ref{thm::most-general-ovb}}\label{subsec::proofs-theorem}

\begin{proof}[Proof of Theorem \ref{thm::multi-ols}]
For (i), let $\rho_1$ be the coefficient of $a$ of the population least squares of $u$ on $(a,c)$. By Cochran's formula, we have
\begin{equation*}
\begin{split}
\gamma_1-\tilde{\gamma}_1&=-\gamma_3\rho_1\\
&=-\frac{\mathrm{cov}(y_{\perp a,c},u_{\perp a,c})\mathrm{cov}(u_{\perp c},a_{\perp c})\mathrm{cov}(a_{\perp c})^{-1}}{\mathrm{var}(u_{\perp a,c})}\\
&=-\frac{\mathrm{cov}(y_{\perp a,c})^{1/2}R_{y\sim u\mid a,c}R_{u\sim a\mid c}\mathrm{cov}(a_{\perp c})^{-1/2}}{\mathrm{var}(u_{\perp a,c})^{-1/2}\mathrm{var}(u_{\perp c})^{1/2}}\\
&=-\frac{\mathrm{cov}(y_{\perp a,c})^{1/2}R_{y\sim u\mid a,c}R_{a\sim u\mid c}^{\T}\mathrm{cov}(a_{\perp c})^{-1/2}}{(1-R_{u\sim a\mid c}^2)^{1/2}}\\
&=-\frac{\mathrm{cov}(y_{\perp a,c})^{1/2}R_{y\sim u\mid a,c}R_{a\sim u\mid c}^{\T}\mathrm{cov}(a_{\perp c})^{-1/2}}{(1-\|R_{a\sim u\mid c}\|_2^2)^{1/2}}. 
\end{split}
\end{equation*}

For (ii), for any $R_1\in \mathbb{B}_{\mathrm{dim}(a)}$ and $R_2\in \mathbb{B}_{\mathrm{dim}(y)}$, Lemma \ref{lem::block-matrix-psd} implies that
\begin{equation*}\Sigma_1=
\begin{pmatrix}
\mathrm{cov}(a_{\perp c}) &  \mathrm{cov}(a_{\perp c})^{1/2}R_1\mathrm{var}(u_{\perp c})^{1/2}\\
\mathrm{var}(u_{\perp c})^{1/2}R_1^{\T} \mathrm{cov}(a_{\perp c})^{1/2} & \mathrm{var}(u_{\perp c})
\end{pmatrix}
\end{equation*}
and 
\begin{equation*}
\Sigma_2=
\begin{pmatrix}
\mathrm{cov}(y_{\perp a,c}) & \mathrm{cov}(y_{\perp a,c})^{1/2}R_2\mathrm{var}(u_{\perp a,c})^{1/2}\\
 \mathrm{var}(u_{\perp a,c})^{1/2}R_2^{\T} \mathrm{cov}(y_{\perp a,c})^{1/2} & \mathrm{var}(u_{\perp a,c})
\end{pmatrix}
\end{equation*}
are positive definite. By Lemma \ref{lem::general-matrix-sharpness}, there exists $u$ such that $R_1=R_{a\sim u\mid c}$ and $R_2=R_{y\sim u\mid a,c}$. Thus, $R_{a\sim u\mid c}\times R_{y\sim u\mid a,c}$ can take arbitrary values in $\mathbb{B}_{\mathrm{dim}(y)}\times \mathbb{B}_{\mathrm{dim}(a)}$. 
\end{proof}

\begin{proof}[Proof of Theorem \ref{thm::most-general-ovb}]
For (i), let $\rho_1$ be the coefficient of $a$ of the population least squares of $u$ on $(a,c)$. By Cochran's formula, we have
\begin{equation*}
\begin{split}
\gamma_1-\tilde{\gamma}_1&=-\gamma_3\rho_1\\
&=-\mathrm{cov}(y_{\perp a,c},u_{\perp a,c})\mathrm{cov}(u_{\perp a,c})^{-1}\mathrm{cov}(u_{\perp c},a_{\perp c})\mathrm{cov}(a_{\perp c})^{-1}\\
&=-\mathrm{cov}(y_{\perp a,c})^{1/2}R_{y\sim u\mid a,c}\mathrm{cov}(u_{\perp a,c})^{-1/2}\mathrm{cov}(u_{\perp c})^{1/2}R_{u\sim a\mid c}\mathrm{cov}(a_{\perp c})^{-1/2}\\
&=-\mathrm{cov}(y_{\perp a,c})^{1/2}R_{y\sim u\mid a,c}\mathrm{cov}(u_{\perp a,c})^{-1/2}\mathrm{cov}(u_{\perp c})^{1/2}R_{a\sim u\mid c}^{\T}\mathrm{cov}(a_{\perp c})^{-1/2}.\\
\end{split}
\end{equation*}
Moreover, consider the population least squares $u_{\perp c}=\mathrm{cov}(u_{\perp c},a_{\perp c})\mathrm{cov}(a_{\perp c})^{-1}a_{\perp c}+u_{\perp a,c}$. This implies 
\begin{equation*}
\begin{split}
\mathrm{cov}(u_{\perp a,c})&=\mathrm{cov}(u_{\perp c})-\mathrm{cov}(u_{\perp c},a_{\perp c})\mathrm{cov}(a_{\perp c})^{-1}\mathrm{cov}(a_{\perp c},u_{\perp c})\\
&=\mathrm{cov}(u_{\perp c})-\mathrm{cov}(u_{\perp c})^{1/2}R_{a\sim u\mid c}^{\T} R_{a\sim u\mid c}\mathrm{cov}(u_{\perp c})^{1/2}.
\end{split}
\end{equation*}

For (ii), for any $Q_1\in \mathbb{S}_{\mathrm{dim}(u)}^{++}$, $R_1\in \mathbb{B}_{\mathrm{dim}(a)\times \mathrm{dim}(u)}$, $R_2\in \mathbb{B}_{\mathrm{dim}(y)\times \mathrm{dim}(u)}$, and $Q_2=Q_1-Q_1^{1/2}R_1^{\T} R_1Q_1^{1/2}$, Lemma \ref{lem::block-matrix-psd} implies that
\begin{equation*}\Sigma_1=
\begin{pmatrix}
\mathrm{cov}(a_{\perp c}) &  \mathrm{cov}(a_{\perp c})^{1/2}R_1Q_1^{1/2}\\
Q_1^{1/2}R_1^{\T} \mathrm{cov}(a_{\perp c})^{1/2} & Q_1
\end{pmatrix}
\end{equation*}
and 
\begin{equation*}
\Sigma_2=
\begin{pmatrix}
\mathrm{cov}(y_{\perp a,c}) & \mathrm{cov}(y_{\perp a,c})^{1/2}R_2Q_2^{1/2}\\
Q_2^{1/2}R_2^{\T} \mathrm{cov}(y_{\perp a,c})^{1/2} & Q_2
\end{pmatrix}.
\end{equation*}
are positive definite. By Lemma \ref{lem::general-matrix-sharpness}, there exists $u$ such that $Q_1=\mathrm{cov}(u_{\perp c})$, $R_1=R_{a\sim u\mid c}$ and $R_2=R_{y\sim u\mid a,c}$. This shows that $\{R_{y\sim u\mid a,c},R_{a\sim u\mid c},\mathrm{cov}(u_{\perp c})\}$ can take arbitrary values in $\mathbb{B}_{\mathrm{dim}(y)\times \mathrm{dim}(u)}\times \mathbb{B}_{\mathrm{dim}(a)\times \mathrm{dim}(u)}\times \mathbb{S}_{\mathrm{dim}(u)}^{++}$. 

\end{proof}

\subsection{Proofs of Propositions \ref{prop::aucm-decom}--\ref{prop::range-connection} and \ref{prop::assume-cov-u-ols}--\ref{prop::general-sensitivity}}\label{subsec::proofs-proposition}

Propositions \ref{prop::aucm-decom}, \ref{prop::aucm-vu}, \ref{prop::benchmarking} are corollaries of Lemma \ref{lem::r-exp}. 

\begin{proof}[Proof of Proposition \ref{prop::scalar-u-range}]

We first show the range of $(\phi_1,\phi_2)$ is the set of all possible $(x_1,x_2)\in \mathbb{R}\times \mathbb{R}^{\mathrm{dim}(m)}$ satisfying
$$x_1^2\leq \frac{\rho_a(\rho_y+\|x_2\|_2^2)}{1-\rho_a},\quad \|x_2\|_2^2\leq \frac{\rho_m\rho_y}{1-\rho_m}.$$
\begin{itemize}
\item Since 
$$\|\phi_2\|_2^2=\frac{R_{y\sim u\mid a,m,c}^2\|R_{m\sim u\mid a,c}\|_2^2}{1-\|R_{m\sim u\mid a,c}\|_2^2}\leq \frac{\rho_m\rho_y}{1-\rho_m}$$
and 
$$\|\phi_1\|_2^2=\frac{R_{a\sim u\mid c}^2}{1-R_{a\sim u\mid c}^2}\cdot (R_{y\sim u\mid a,m,c}^2+\|\phi_2\|_2^2)\leq \frac{\rho_a(\rho_y+\|\phi_2\|_2^2)}{1-\rho_a}$$
we conclude that the range of $(\phi_1,\phi_2)$ is a subset of all possible $(x_1,x_2)\in \mathbb{R}\times \mathbb{R}^{\mathrm{dim}(m)}$ satisfying
$$x_1^2\leq \frac{\rho_a(\rho_y+\|x_2\|_2^2)}{1-\rho_a},\quad \|x_2\|_2^2\leq \frac{\rho_m\rho_y}{1-\rho_m}.$$

\item For any $(x_1,x_2)\in \mathbb{R}\times \mathbb{R}^{\mathrm{dim}(m)}$ satisfying
$$x_1^2\leq \frac{\rho_a(\rho_y+\|x_2\|_2^2)}{1-\rho_a},\quad \|x_2\|_2^2\leq \frac{\rho_m\rho_y}{1-\rho_m},$$
we can choose 
$$R_{a\sim u\mid c}=\frac{\rho_m^{1/2}x_1}{(\rho_mx_1^2+x_2^2)^{1/2}},\quad R_{m\sim u\mid a,c}=\frac{\rho_m^{1/2}x_2}{\|x_2\|_2},\quad R_{y\sim u\mid a,m,c}=\frac{\|x_2\|_2(1-\rho_m)^{1/2}}{\rho_m^{1/2}}.$$
Then, we have
\begin{equation*}
\begin{split}
\phi_1&=\frac{R_{y\sim u\mid a,m,c}R_{a\sim u\mid c}}{(1-R_{a\sim u\mid c}^2)^{1/2}(1-\|R_{m\sim u\mid a,c}\|_2^2)^{1/2}}=x_1,\\
\phi_2&=\frac{R_{y\sim u\mid a,m,c}\cdot R_{m\sim u\mid a,c}}{(1-\|R_{m\sim u\mid a,c}\|_2^2)^{1/2}}=x_2,
\end{split}
\end{equation*}
and $R_{a\sim u\mid c}^2\leq \rho_a$, $\|R_{m\sim u\mid a,c}\|_2^2\leq \rho_m$, and $R_{y\sim u\mid a,m,c}^2\leq \rho_y$. Therefore, $(\phi_1,\phi_2)$ can take all possible $(x_1,x_2)\in \mathbb{R}\times \mathbb{R}^{\mathrm{dim}(m)}$ satisfying
$$x_1^2\leq \frac{\rho_a(\rho_y+\|x_2\|_2^2)}{1-\rho_a},\quad \|x_2\|_2^2\leq \frac{\rho_m\rho_y}{1-\rho_m}.$$
\end{itemize}

We then show the range of $(\phi_2,\phi_3)$ is the set of all possible $(x_2,x_3)\in \mathbb{R}^{\mathrm{dim}(m)}\times \mathbb{R}^{\mathrm{dim}(m)}$ satisfying 
$$\|x_2\|_2^2\leq \frac{\rho_m\rho_y}{1-\rho_m},\quad \|x_3\|_2^2\leq \frac{\rho_a\rho_m}{1-\rho_a},$$
and there exists a real number $t$ such that $x_2=tx_3$. 
\begin{itemize}
\item Since 
$$\|\phi_2\|_2^2=\frac{R_{y\sim u\mid a,m,c}^2\|R_{m\sim u\mid a,c}\|_2^2}{1-\|R_{m\sim u\mid a,c}\|_2^2}\leq \frac{\rho_m\rho_y}{1-\rho_m},$$
$$\|\phi_3\|_2^2=\frac{R_{a\sim u\mid c}^2\|R_{m\sim u\mid a,c}\|_2^2}{1-\|R_{m\sim u\mid a,c}\|_2^2}\leq \frac{\rho_a\rho_m}{1-\rho_m},$$
and $\phi_2,\phi_3$ must be in the same direction as $R_{m\sim u\mid a,c}$, the range of $(\phi_2,\phi_3)$ is a subset of all possible $(x_2,x_3)\in \mathbb{R}^{\mathrm{dim}(m)}\times \mathbb{R}^{\mathrm{dim}(m)}$ satisfying 
$$\|x_2\|_2^2\leq \frac{\rho_m\rho_y}{1-\rho_m},\quad \|x_3\|_2^2\leq \frac{\rho_a\rho_m}{1-\rho_a},$$
and there exists a real number $t$ such that $x_2=tx_3$.

\item For any $(x_2,x_3)\in \mathbb{R}^{\mathrm{dim}(m)}\times \mathbb{R}^{\mathrm{dim}(m)}$ satisfying 
$$\|x_2\|_2^2\leq \frac{\rho_m\rho_y}{1-\rho_m},\quad \|x_3\|_2^2\leq \frac{\rho_a\rho_m}{1-\rho_a},$$
and there exists a real number $t$ such that $x_2=tx_3$, when $x_2\not=0$, we can choose 
$$R_{a\sim u\mid c}=\frac{(1-\rho_m)^{1/2}\|x_2\|_2}{t\rho_m^{1/2}},\quad R_{m\sim u\mid a,c}=\frac{\rho_m^{1/2}x_2}{\|x_2\|_2},\quad R_{y\sim u\mid a,m,c}=\frac{\|x_2\|_2(1-\rho_m)^{1/2}}{\rho_m^{1/2}};$$
and when $x_2=0$ and $x_3\not=0$, we can choose
$$R_{a\sim u\mid c}=\frac{(1-\rho_m)^{1/2}\|x_3\|_2}{\rho_m^{1/2}},\quad R_{m\sim u\mid a,c}=\frac{\rho_m^{1/2}x_3}{\|x_3\|_2},\quad R_{y\sim u\mid a,m,c}=0;$$
and when $x_2=x_3=0$, we can choose all $R$ parameters to be 0. Then, we have
\begin{equation*}
\begin{split}
\phi_2&=\frac{R_{y\sim u\mid a,m,c}\cdot R_{m\sim u\mid a,c}}{(1-\|R_{m\sim u\mid a,c}\|_2^2)^{1/2}},\\
\phi_3&=\frac{R_{a\sim u\mid c}\cdot R_{m\sim u\mid a,c}}{(1-R_{a\sim u\mid c}^2)^{1/2}},
\end{split}
\end{equation*}
and $R_{a\sim u\mid c}^2\leq \rho_a$, $\|R_{m\sim u\mid a,c}\|_2^2\leq \rho_m$, and $R_{y\sim u\mid a,m,c}^2\leq \rho_y$. Therefore, $(\phi_2,\phi_3)$ can take all possible $(x_2,x_3)\in \mathbb{R}^{\mathrm{dim}(m)}\times \mathbb{R}^{\mathrm{dim}(m)}$ satisfying 
$$\|x_2\|_2^2\leq \frac{\rho_m\rho_y}{1-\rho_m},\quad \|x_3\|_2^2\leq \frac{\rho_a\rho_m}{1-\rho_a},$$
and there exists a real number $t$ such that $x_2=tx_3$. 
\end{itemize}
\end{proof}

\begin{proof}[Proof of Proposition \ref{prop::vector-u-range}]

We first show the range of $(\phi_1^*,\phi_2^*)$ is the set of all possible $(x_1,x_2)\in \mathbb{R}\times \mathbb{R}^{\mathrm{dim}(m)}$ satisfying 
$$x_1^2\leq \frac{\rho_a(\rho_y+\|x_2\|_2^2)}{1-\rho_a},\quad \|x_2\|_2^2\leq \frac{\rho_m\rho_y}{1-\rho_m}.$$
We have 
$$\|\phi_2^*\|_2^2\leq \frac{\|R_{y\sim u\mid a,m,c}\|_2^2\|R_{m\sim u\mid a,c}\|_2^2}{1-\|R_{m\sim u\mid a,c}\|_2^2}\leq \frac{\rho_m\rho_y}{1-\rho_m}$$
and 
$$R_{y\sim u\mid a,m,c} (I-R_{m\sim u\mid a,c}^{\T}R_{m\sim u\mid a,c})^{-1}R_{y\sim u\mid a,m,c}^\T=\|\phi_2^*\|_2^2+\|R_{y\sim u\mid a,m,c}\|_2^2\leq \|\phi_2^*\|_2^2+\rho_y.$$
Therefore, the range of $(\phi_1^*,\phi_2^*)$ is a subset of all possible $(x_1,x_2)\in \mathbb{R}\times \mathbb{R}^{\mathrm{dim}(m)}$ satisfying
$$x_1^2\leq \frac{\rho_a(\rho_y+\|x_2\|_2^2)}{1-\rho_a},\quad \|x_2\|_2^2\leq \frac{\rho_m\rho_y}{1-\rho_m}.$$
Moreover, by Proposition \ref{prop::scalar-u-range}, $(\phi_1^*,\phi_2^*)$ can take all possible $(x_1,x_2)\in \mathbb{R}\times \mathbb{R}^{\mathrm{dim}(m)}$ satisfying
$$x_1^2\leq \frac{\rho_a(\rho_y+\|x_2\|_2^2)}{1-\rho_a},\quad \|x_2\|_2^2\leq \frac{\rho_m\rho_y}{1-\rho_m}.$$

We then show the range of $(\phi_2^*,\phi_3^*)$ is the set of all possible $(x_2,x_3)\in \mathbb{R}^{\mathrm{dim}(m)}\times \mathbb{R}^{\mathrm{dim}(m)}$ satisfying 
$$\|x_2\|_2^2\leq \frac{\rho_m\rho_y}{1-\rho_m},\quad \|x_3\|_2^2\leq \frac{\rho_a\rho_m}{1-\rho_a}.$$

\begin{itemize}
\item Since 
$$\|\phi_2^*\|_2^2\leq \frac{\|R_{y\sim u\mid a,m,c}\|_2^2\|R_{m\sim u\mid a,c}\|_2^2}{1-\|R_{m\sim u\mid a,c}\|_2^2}\leq \frac{\rho_m\rho_y}{1-\rho_m}$$
$$\|\phi_3^*\|_2^2\leq \frac{\|R_{a\sim u\mid c}\|_2^2\|R_{m\sim u\mid a,c}\|_2^2}{1-\|R_{m\sim u\mid a,c}\|_2^2}\leq \frac{\rho_a\rho_m}{1-\rho_m}$$
the range of $(\phi_2^*,\phi_3^*)$ is a subset of all possible $(x_2,x_3)\in \mathbb{R}^{\mathrm{dim}(m)}\times \mathbb{R}^{\mathrm{dim}(m)}$ satisfying 
$$\|x_2\|_2^2\leq \frac{\rho_m\rho_y}{1-\rho_m},\quad \|x_3\|_2^2\leq \frac{\rho_a\rho_m}{1-\rho_a},$$
\item For any $(x_2,x_3)\in \mathbb{R}^{\mathrm{dim}(m)}\times \mathbb{R}^{\mathrm{dim}(m)}$ satisfying 
$$\|x_2\|_2^2\leq \frac{\rho_m\rho_y}{1-\rho_m},\quad \|x_3\|_2^2\leq \frac{\rho_a\rho_m}{1-\rho_a},$$
Let 
$$t_1=\left(\frac{1-\rho_m}{\rho_m\rho_y}\right)^{1/2}x_2,\quad t_2=\left(\frac{1-\rho_a}{\rho_a\rho_m}\right)^{1/2}x_3$$
and let $\{v_1,v_2\}$ be the orthogonal unit vectors in $\mathbb{R}^{\mathrm{dim}(m)}$ such that $t_1,t_2$ are in the linear space spanned by $\{v_1,v_2\}$. Let $t_1=\lambda_{11}v_1+\lambda_{12}v_2$ and $t_2=\lambda_{21}v_1+\lambda_{22}v_2$. Then, we have $\lambda_{11}^2+\lambda_{12}^2\leq 1$ and $\lambda_{21}^2+\lambda_{22}^2\leq 1$. We can choose
\begin{equation*}
\begin{split}
R_{m\sim u\mid a,c}&=\rho_m^{1/2}\begin{pmatrix}
v_1 & v_2 & 0 & \cdots & 0
\end{pmatrix},\\
R_{y\sim u\mid a,m,c}&=\rho_y^{1/2}\begin{pmatrix}
\lambda_{21} & \lambda_{22} & 0 & \cdots & 0
\end{pmatrix},\\
R_{a\sim u\mid c}&=\left\{\frac{\rho_a}{\rho_a(\lambda_{11}^2+\lambda_{12}^2)-\rho_a+1}\right\}^{1/2}\begin{pmatrix}
\lambda_{11} & \lambda_{12} & 0 & \cdots & 0
\end{pmatrix},\\
\end{split}
\end{equation*}
which satisfy 
\begin{equation*}
\begin{split}
\phi_2^*&=R_{m\sim u\mid a,c}(I-R_{m\sim u\mid a,c}^{\T} R_{m\sim u\mid a,c})^{-1/2}R_{y\sim u\mid a,m,c}^{\T},\\
\phi_3^*&=\frac{R_{m\sim u\mid a,c}R_{a\sim u\mid c}^{\T}}{(1-\|R_{a\sim u\mid c}\|_2^2)^{1/2}},
\end{split}
\end{equation*}
and $\|R_{a\sim u\mid c}\|_2^2\leq \rho_a$, $\|R_{m\sim u\mid a,c}\|_2^2\leq \rho_m$, and $\|R_{y\sim u\mid a,m,c}\|_2^2\leq \rho_y$. Therefore, $(\phi_2^*,\phi_3^*)$ can take all possible $(x_2,x_3)\in \mathbb{R}^{\mathrm{dim}(m)}\times \mathbb{R}^{\mathrm{dim}(m)}$ satisfying 
$$\|x_2\|_2^2\leq \frac{\rho_m\rho_y}{1-\rho_m},\quad \|x_3\|_2^2\leq \frac{\rho_a\rho_m}{1-\rho_a}.$$
\end{itemize}

\end{proof}

Proposition \ref{prop::range-connection} is a direct corollary of Proposition \ref{prop::scalar-u-range} and Proposition \ref{prop::vector-u-range}.

 \begin{proof}[Proof of Proposition \ref{prop::assume-cov-u-ols}]
By Theorem \ref{thm::most-general-ovb}(i), we have
$$R_{y\sim u\mid a,c}\mathrm{cov}(u_{\perp a,c})^{-1/2}\mathrm{cov}(u_{\perp c})^{1/2}R_{a\sim u\mid c}^{\T}=\mathrm{cov}(y_{\perp a,c})^{-1/2}(\tilde\gamma_1-\gamma_1)\mathrm{cov}(a_{\perp c})^{1/2},$$
where $\gamma_1$ is invariant to any inveritable transformation of $u$, and $\tilde{\gamma}_1$, $\mathrm{cov}(y_{\perp a,c})$, and $\mathrm{cov}(a_{\perp c})$ are independent of $u$. Therefore, 
$$f(R_{y\sim u\mid a,c}\mathrm{cov}(u_{\perp a,c})^{-1/2}\mathrm{cov}(u_{\perp c})^{1/2}R_{a\sim u\mid c}^{\T})$$
is invariant to any inveritable transformation of $u$ for any function $f$. For any specific $u$ that attains the minimum of the optimization problem and for any given invertible matrix $\Sigma$, we take $\tilde{u}=\Sigma^{1/2}\mathrm{cov}(u_{\perp c})^{-1/2}u_{\perp c}$ to ensure $\mathrm{cov}(\tilde{u}_{\perp c})=\Sigma$. Moreover, we can verify that $\|R_{y\sim u\mid a,c}\|_2=\|R_{y\sim \tilde{u}\mid a,c}\|_2$ and $\|R_{a\sim u\mid c}\|_2=\|R_{a\sim \tilde{u}\mid c}\|_2$. Therefore, such $\tilde{u}$ also attains the minimum of the optimization problem. 

When $\Sigma=(I-R_{a\sim u\mid c}^{\T}R_{a\sim u\mid c})^{-1}$, we have $\mathrm{cov}(u_{\perp a,c})=I$ and it is equivalent to solving the optimization problem
\begin{equation*}
\begin{aligned}
&\min\quad && f(R_{y\sim u\mid a,c}(I-R_{a\sim u\mid c}^{\T} R_{a\sim u\mid c})^{-1/2}R_{a\sim u\mid c}^{\T})\\
&\text{subject to }\quad &&\|R_{a\sim u\mid c}\|_2^2 \leq \rho_a, \quad \|R_{y\sim u\mid a,c}\|_2^2 \leq \rho_y.
\end{aligned}
\end{equation*}

\end{proof}

The proof of Propositions \ref{prop::assume-cov-u-direct}--\ref{prop::assume-cov-u-indirect} are similar to the proof of Proposition \ref{prop::assume-cov-u-ols}.

 \begin{proof}[Proof of Proposition \ref{prop::general-sensitivity}]

    The covariance matrix of $(a_{\perp c,u}, m_{\perp c,u},y_{\perp c,u})$ is determined by the following six elements: 
    $$\mathrm{var}(a_{\perp c,u}),\quad \mathrm{cov}(m_{\perp c,u}, a_{\perp c,u}),\quad \mathrm{cov}(m_{\perp c,u}),\quad \mathrm{cov}(y_{\perp c,u}, a_{\perp c,u}),\quad \mathrm{cov}(y_{\perp c,u}, m_{\perp c,u}),\quad \mathrm{var}(y_{\perp c,u}).$$
    Below, we will show that each of them depends only on moments of observables and $RR^{\T}$. 
    \begin{enumerate}
        \item We have 
        \begin{equation*}
            \begin{split}
                \mathrm{var}(a_{\perp c,u})=\mathrm{var}(a_{\perp c})-\mathrm{cov}(a_{\perp c},u_{\perp c})\mathrm{cov}(u_{\perp c})^{-1}\mathrm{cov}(u_{\perp c},a_{\perp c})
                =\mathrm{var}(a_{\perp c})(1-R_{a\sim u\mid c}R_{a\sim u\mid c}^{\T})
            \end{split}
        \end{equation*}
        which depends only on the moments of observables and $RR^{\T}$. 
        \item We have
        \begin{equation*}
            \begin{split}
                \mathrm{cov}(m_{\perp c,u},a_{\perp c,u})&=\mathrm{cov}(m_{\perp c},a_{\perp c})-\mathrm{cov}(m_{\perp c},u_{\perp c})\mathrm{cov}(u_{\perp c})^{-1}\mathrm{cov}(u_{\perp c},a_{\perp c})\\
                &=\mathrm{cov}(m_{\perp c},a_{\perp c})-\mathrm{cov}(m_{\perp c})^{1/2}R_{m\sim u\mid c}R_{a\sim u\mid c}^{\T}\mathrm{var}(a_{\perp c})^{1/2},\\
            \end{split}
        \end{equation*}
        where
        \begin{equation*}
            \begin{split}
                R_{m\sim u\mid c}=\mathrm{cov}(m_{\perp c})^{-1/2}\mathrm{cov}(m_{\perp a,c})^{1/2}R_{m\sim u\mid a,c}\mathrm{cov}(u_{\perp a,c})^{1/2}\mathrm{cov}(u_{\perp c})^{-1/2}+R_{m\sim a\mid c}R_{a\sim u\mid c}.
            \end{split}
        \end{equation*}
        Since $\mathrm{cov}(u_{\perp a,c})=I$, we have $\mathrm{cov}(u_{\perp c})=(I-R_{a\sim u\mid c}^{\T}R_{a\sim u\mid c})^{-1}$ and thus
        \begin{equation*}
            \begin{split}
                R_{m\sim u\mid a,c}\mathrm{cov}(u_{\perp a,c})^{1/2}\mathrm{cov}(u_{\perp c})^{-1/2}R_{a\sim u\mid c}^{\T}
                =&R_{m\sim u\mid a,c}(I-R_{a\sim u\mid c}^{\T}R_{a\sim u\mid c})^{-1/2}R_{a\sim u\mid c}^{\T}\\
                =&\frac{R_{m\sim u\mid a,c}R_{a\sim u\mid c}^{\T}}{(1-R_{a\sim u\mid c}R_{a\sim u\mid c}^{\T})^{1/2}}
            \end{split}
        \end{equation*}
        depends only on the moments of observables and $RR^{\T}$. Therefore, $\mathrm{cov}(m_{\perp c,u},a_{\perp c,u})$ depends only on the moments of observables and $RR^{\T}$. 
        \item We have
        \begin{equation*}
                \mathrm{cov}(m_{\perp c,u})=\mathrm{cov}(m_{\perp a,c,u})+\mathrm{cov}(m_{\perp c,u},a_{\perp c,u})\mathrm{var}(a_{\perp c,u})^{-1}\mathrm{cov}(a_{\perp c,u},m_{\perp c,u})
        \end{equation*}
        where $\mathrm{cov}(m_{\perp c,u},a_{\perp c,u})\mathrm{var}(a_{\perp c,u})^{-1}\mathrm{cov}(a_{\perp c,u},m_{\perp c,u})$ depends only on moments of observables and $RR^{\T}$ by the previous parts, and
        \begin{equation*}
            \begin{split}
                \mathrm{cov}(m_{\perp a,c,u})&=\mathrm{cov}(m_{\perp a,c})-\mathrm{cov}(m_{\perp a,c},u_{\perp a,c})\mathrm{cov}(u_{\perp a,c})^{-1}\mathrm{cov}(u_{\perp a,c},m_{\perp a,c})\\
                &=\mathrm{cov}(m_{\perp a,c})^{1/2}(I-R_{m\sim u\mid a,c}R_{m\sim u\mid a,c}^{\T})\mathrm{cov}(m_{\perp a,c})^{1/2}
            \end{split}
        \end{equation*}
        depends only on the moments of observables and $RR^{\T}$. Therefore, $\mathrm{cov}(m_{\perp c,u})$ depends only on the moments of observables and $RR^{\T}$. 

                \item We have
        \begin{equation*}
            \begin{split}
                \mathrm{cov}(y_{\perp c,u},a_{\perp c,u})&=\mathrm{cov}(y_{\perp c},a_{\perp c})-\mathrm{cov}(y_{\perp c},u_{\perp c})\mathrm{cov}(u_{\perp c})^{-1}\mathrm{cov}(u_{\perp c},a_{\perp c})\\
                &=\mathrm{cov}(y_{\perp c},a_{\perp c})-\mathrm{var}(y_{\perp c})^{1/2}R_{y\sim u\mid c}R_{a\sim u\mid c}^{\T}\mathrm{var}(a_{\perp c})^{1/2},\\
            \end{split}
        \end{equation*}
        where 
        \begin{equation*}
            \begin{split}
                R_{y\sim u\mid c}=&\mathrm{var}(y_{\perp c})^{-1/2}\mathrm{var}(y_{\perp a,c})^{1/2}R_{y\sim u\mid a,c}\mathrm{cov}(u_{\perp a,c})^{1/2}\mathrm{cov}(u_{\perp c})^{-1/2}+R_{y\sim a\mid c}R_{a\sim u\mid c},\\
                R_{y\sim u\mid a,c}=&\mathrm{var}(y_{\perp a,c})^{-1/2}\mathrm{var}(y_{\perp a,m,c})^{1/2}R_{y\sim u\mid a,m,c}\mathrm{cov}(u_{\perp a,m,c})^{1/2}\mathrm{cov}(u_{\perp a,c})^{-1/2}\\
                &+R_{y\sim m\mid a,c}R_{m\sim u\mid a,c}.\\
            \end{split}
        \end{equation*}
        We only need to show $R_{y\sim u\mid a,m,c}\mathrm{cov}(u_{\perp a,m,c})^{1/2}\mathrm{cov}(u_{\perp c})^{-1/2}R_{a\sim u\mid c}^{\T}$ depends only on moments of observables and $RR^{\T}$. 
        
        Since $\mathrm{cov}(u_{\perp a,c})=I$, we have $\mathrm{cov}(u_{\perp c})=(I-R_{a\sim u\mid c}^{\T}R_{a\sim u\mid c})^{-1}$ and $\mathrm{cov}(u_{\perp a,m,c})=I-R_{m\sim u\mid a,c}^{\T}R_{m\sim u\mid a,c}$ and thus
        \begin{equation*}
            \begin{split}
                &R_{y\sim u\mid a,m,c}\mathrm{cov}(u_{\perp a,m,c})^{1/2}\mathrm{cov}(u_{\perp c})^{-1/2}R_{a\sim u\mid c}^{\T}\\
                =&R_{y\sim u\mid a,m,c}(I-R_{m\sim u\mid a,c}^{\T}R_{m\sim u\mid a,c})^{1/2}(I-R_{a\sim u\mid c}^{\T}R_{a\sim u\mid c})^{-1/2}R_{a\sim u\mid c}^{\T}\\
                =&\frac{R_{y\sim u\mid a,m,c}(I-R_{m\sim u\mid a,c}^{\T}R_{m\sim u\mid a,c})^{1/2}R_{a\sim u\mid c}^{\T}}{(1-R_{a\sim u\mid c}R_{a\sim u\mid c}^{\T})^{1/2}}
            \end{split}
        \end{equation*}
        depends only on the moments of observables and $RR^{\T}$. Therefore, $\mathrm{cov}(y_{\perp c,u},a_{\perp c,u})$ depends only on the moments of observables and $RR^{\T}$. 

                \item We have
        \begin{equation*}
            \begin{split}
                \mathrm{cov}(y_{\perp c,u},m_{\perp c,u})&=\mathrm{cov}(y_{\perp c},m_{\perp c})-\mathrm{cov}(y_{\perp c},u_{\perp c})\mathrm{cov}(u_{\perp c})^{-1}\mathrm{cov}(u_{\perp c},m_{\perp c})\\
                &=\mathrm{cov}(y_{\perp c},m_{\perp c})-\mathrm{cov}(y_{\perp c})^{1/2}R_{y\sim u\mid c}R_{m\sim u\mid c}^{\T}\mathrm{cov}(m_{\perp c})^{1/2},\\
            \end{split}
        \end{equation*}
        where the formula of $R_{m\sim u\mid c}$ is given in part 2 and the formula of $R_{y\sim u\mid c}$ is given in part 4. We only need to show $R_{m\sim u\mid a,c}\mathrm{cov}(u_{\perp a,c})^{1/2}\mathrm{cov}(u_{\perp c})^{-1}\mathrm{cov}(u_{\perp a,c})^{1/2}R_{m\sim u\mid a,c}^{\T}$ and 
        \newline
        $R_{y\sim u\mid a,m,c}\mathrm{cov}(u_{\perp a,m,c})^{1/2}\mathrm{cov}(u_{\perp c})^{-1}\mathrm{cov}(u_{\perp a,c})^{1/2}R_{m\sim u\mid a,c}^{\T}$ depend only on moments of observables and $RR^{\T}$ by Lemma \ref{lem::linear-algebra-matrix}.

        We have 
        \begin{equation*}
            \begin{split}
                &R_{m\sim u\mid a,c}\mathrm{cov}(u_{\perp a,c})^{1/2}\mathrm{cov}(u_{\perp c})^{-1}\mathrm{cov}(u_{\perp a,c})^{1/2}R_{m\sim u\mid a,c}^{\T}\\
                =&R_{m\sim u\mid a,c}(I-R_{a\sim u\mid c}^{\T}R_{a\sim u\mid c})^{-1/2}R_{m\sim u\mid a,c}^{\T}\\
            \end{split}
        \end{equation*}
        and
        \begin{equation*}
            \begin{split}
                &R_{y\sim u\mid a,m,c}\mathrm{cov}(u_{\perp a,m,c})^{1/2}\mathrm{cov}(u_{\perp c})^{-1}\mathrm{cov}(u_{\perp a,c})^{1/2}R_{m\sim u\mid a,c}^{\T}\\
                =&R_{y\sim u\mid a,m,c}(I-R_{m\sim u\mid a, c}^{\T}R_{m\sim u\mid a,c})^{-1/2}(I-R_{a\sim u\mid c}^{\T}R_{a\sim u\mid c})R_{m\sim u\mid a,c}^{\T}\\
            \end{split}
        \end{equation*}
        which depend only on the moments of observables and $RR^{\T}$, by Lemma \ref{lem::linear-algebra-matrix}. Therefore, $\mathrm{cov}(y_{\perp c,u},m_{\perp c,u})$ depends only on the moments of observables and $RR^{\T}$. 

        \item We have
        \begin{equation*}
                \mathrm{var}(y_{\perp c,u})=\mathrm{var}(y_{\perp a,c,u})+\mathrm{cov}(y_{\perp c,u},a_{\perp c,u})\mathrm{var}(a_{\perp c,u})^{-1}\mathrm{cov}(a_{\perp c,u},y_{\perp c,u})
        \end{equation*}
        where $\mathrm{cov}(y_{\perp c,u},a_{\perp c,u})\mathrm{var}(a_{\perp c,u})^{-1}\mathrm{cov}(a_{\perp c,u},y_{\perp c,u})$ depends only on moments of observables and $RR^{\T}$ by the previous parts; and
        \begin{equation*}
                \mathrm{var}(y_{\perp a,c,u})=\mathrm{var}(y_{\perp a,m,c,u})+\mathrm{cov}(y_{\perp a,c,u},m_{\perp a,c,u})\mathrm{cov}(m_{\perp a,c,u})^{-1}\mathrm{cov}(m_{\perp a,c,u},y_{\perp a,c,u})
        \end{equation*}
        where we can use similar arguments in parts 1-2 to verify that both $\mathrm{cov}(y_{\perp a,c,u},m_{\perp a,c,u})$ and $\mathrm{cov}(m_{\perp a,c,u})$ depend only on moments of observables and $RR^{\T}$; and
        \begin{equation*}
            \begin{split}
                \mathrm{var}(y_{\perp a,m,c,u})&=\mathrm{var}(y_{\perp a,m,c})-\mathrm{cov}(y_{\perp a,m,c},u_{\perp a,m,c})\mathrm{cov}(u_{\perp a,m,c})^{-1}\mathrm{cov}(u_{\perp a,m,c},y_{\perp a,m,c})\\
                &=\mathrm{var}(m_{\perp a,c})(1-R_{y\sim u\mid a,m,c}R_{y\sim u\mid a,m,c}^{\T})
            \end{split}
        \end{equation*}
        depends only on the moments of observables and $RR^{\T}$. Therefore, $\mathrm{var}(y_{\perp c,u})$ depends only on the moments of observables and $RR^{\T}$. 
    \end{enumerate}

\end{proof}

\subsection{Proofs of Corollaries \ref{cor::direct}--\ref{cor::vu-direct} and \ref{cor::benchmarking}}\label{subsec::proofs-corollary}

\begin{proof}[Proof of Corollary \ref{cor::direct}]

For any $R_1,R_3\in (-1,1)$ and $R_2\in \mathbb{B}_{\mathrm{dim}(m)}$, Lemma \ref{lem::block-matrix-psd} implies that
\begin{equation*}\Sigma_1=
\begin{pmatrix}
\mathrm{var}(a_{\perp c}) &  R_1\mathrm{var}(a_{\perp c})^{1/2}\mathrm{var}(u_{\perp c})^{1/2}\\
R_1\mathrm{var}(u_{\perp c})^{1/2}\mathrm{var}(a_{\perp c})^{1/2} & \mathrm{var}(u_{\perp c})
\end{pmatrix},
\end{equation*}
\begin{equation*}
\Sigma_2=
\begin{pmatrix}
\mathrm{cov}(m_{\perp a,c}) & \mathrm{cov}(m_{\perp a,c})^{1/2}R_2\mathrm{var}(u_{\perp a,c})^{1/2}\\
 \mathrm{var}(u_{\perp a,c})^{1/2}R_2^{\T} \mathrm{cov}(m_{\perp a,c})^{1/2} & \mathrm{var}(u_{\perp a,c})
\end{pmatrix},
\end{equation*}
and 
\begin{equation*}
\Sigma_3=
\begin{pmatrix}
\mathrm{var}(y_{\perp a,m,c}) & R_3\mathrm{var}(y_{\perp a,m,c})^{1/2}\mathrm{var}(u_{\perp a,m,c})^{1/2}\\
R_3\mathrm{var}(u_{\perp a,m,c})^{1/2}\mathrm{var}(y_{\perp a,m,c})^{1/2} & \mathrm{var}(u_{\perp a,m,c})
\end{pmatrix}
\end{equation*}
are positive definite. By Lemma \ref{lem::general-matrix-sharpness}, there exists $u$ such that $R_1=R_{a\sim u\mid c}$, $R_2=R_{m\sim u\mid a,c}$, $R_3=R_{y\sim u\mid a,m,c}$. This shows that $\{R_{y\sim u\mid a,m,c},R_{m\sim u\mid a,c},R_{a\sim u\mid c}\}$ can take arbitrary values in $(-1,1)\times \mathbb{B}_{\mathrm{dim}(m)}\times (-1,1)$. 
\end{proof}

The proof of Corollary \ref{cor::vu-direct} is similar to the proof of Theorem \ref{thm::most-general-ovb} and Corollary \ref{cor::direct}.  The proof of Corollary \ref{cor::benchmarking} is similar to the proof of Corollary \ref{cor::direct}. We omit these proofs.

\end{document}